\documentclass[12pt]{article}
\usepackage[utf8]{inputenc}

\PassOptionsToPackage{unicode}{hyperref}
\PassOptionsToPackage{hyphens}{url}
\PassOptionsToPackage{dvipsnames,svgnames,x11names}{xcolor}

\usepackage{bbm}
\usepackage{etoolbox}
\usepackage[T1]{fontenc}
\usepackage[margin = 1in]{geometry}
\usepackage{xr-hyper}
\usepackage{hhline}
\usepackage{varwidth}
\usepackage{graphicx}
\usepackage{pdflscape}
\usepackage{afterpage}
\usepackage{amsmath}
\usepackage{amssymb}
\usepackage{amsfonts}
\usepackage{amsthm}
\usepackage[nodisplayskipstretch]{setspace}
\def\spacingset#1{\renewcommand{\baselinestretch}{#1}\small\normalsize}
\usepackage{subcaption}
\usepackage{adjustbox}
\usepackage{threeparttable}
\usepackage{multirow}
\usepackage{booktabs}
\usepackage{relsize}
\usepackage{tabularx}
\usepackage{soul}
\usepackage[normalem]{ulem}
\usepackage{stackengine}
\usepackage[bottom]{footmisc}
\usepackage{longtable}
\usepackage{hyphenat}
\usepackage{calc}
\usepackage{accents}
\usepackage{upgreek}
\usepackage{mathrsfs}
\usepackage{rotating}
\usepackage{enumerate}
\stackMath

\usepackage{indentfirst}
\theoremstyle{plain}
\newtheorem{theorem}{Theorem}
\newtheorem{corollary}{Corollary}

\newtheorem{example}{Example}

\newtheorem{lemma}{Lemma}

\newtheorem{prop}{Proposition}

\newtheorem{assumption}{Assumption}
\newtheorem{algorithm}{Algorithm}
\theoremstyle{definition}
\newtheorem{remark}{Remark}

\numberwithin{equation}{section}
\numberwithin{assumption}{section}
\numberwithin{theorem}{section}
\numberwithin{algorithm}{section}
\numberwithin{corollary}{section}
\numberwithin{remark}{section}
\numberwithin{proposition}{section}
\numberwithin{prop}{section}
\numberwithin{lemma}{section}
\numberwithin{table}{section}
\numberwithin{figure}{section}

\allowdisplaybreaks
\usepackage[skip=0pt]{caption}
\usepackage{placeins}
\usepackage[footfont=normalsize,font=normalsize]{floatrow}
\usepackage{float}
\floatstyle{plaintop}
\restylefloat{table}

\usepackage{array}
\newcolumntype{H}{>{\setbox0=\hbox\bgroup}c<{\egroup}@{}}

\usepackage[authoryear,longnamesfirst]{natbib}

\usepackage{xcolor}
\usepackage{color}
\definecolor{marineblue2}{rgb}{0.05,0.1,0.4}
\usepackage{hyperref}
\hypersetup{
unicode=false,
pdffitwindow=true,
pdfstartview={FitH},
pdftitle={Testing for Essential Heterogeneity in the Marginal Treatment Effects Model},
pdfnewwindow=true,
colorlinks=true,
linkcolor=marineblue2,
citecolor=marineblue2,
filecolor=marineblue2,
urlcolor=marineblue2
}

\usepackage{bookmark}
\IfFileExists{xurl.sty}{\usepackage{xurl}}{}
\makeatletter
\newcommand*{\addFileDependency}[1]{
  \typeout{(#1)}
  \@addtofilelist{#1}
  \IfFileExists{#1}{}{\typeout{No file #1.}}
}
\makeatother

\AtBeginDocument{
\addtolength{\belowdisplayskip}{-1ex}
\addtolength{\belowdisplayshortskip}{-1ex}
\addtolength{\abovedisplayskip}{-1.5ex}
\addtolength{\abovedisplayshortskip}{-1.5ex}
}

\usepackage{appendix}

\providecommand{\keywords}[1]{\small\textbf{Keywords:} #1}

\newcolumntype{L}[1]{>{\raggedright\let\newline\\arraybackslash\hspace{0pt}}m{#1}}
\newcolumntype{C}[1]{>{\centering\let\newline\\arraybackslash\hspace{0pt}}m{#1}}
\newcolumntype{R}[1]{>{\raggedleft\let\newline\\arraybackslash\hspace{0pt}}m{#1}}


\newcommand{\anon}{1}

\newcommand{\TXP}[1]{T^{\mathit{XP}}_{#1}}
\newcommand{\TP}[1]{T^{\mathit{P}}_{#1}}
\newcommand{\Qpoly}[1]{Q^{\mathrm{poly}}_{#1}}

\spacingset{1}

\begin{document}

\if1\anon
{
  \title{\bf Testing for Essential Heterogeneity in the Marginal Treatment Effects Model}
  \author{Yuhao Li\thanks{Email: \href{mailto:yuhao.li@xjtlu.edu.cn}{yuhao.li@xjtlu.edu.cn}.}\\
    Xi'an Jiaotong--Liverpool University\\
    and\\
    Minghai Mao\thanks{Financial support from the National Science Foundation of China 
    [Grant Number 72203084] is gratefully acknowledged. Email: \href{mailto:maominghai@lnu.edu.cn}{maominghai@lnu.edu.cn}.}\\
    Liaoning University\\
    and\\
    Xiaojun Song\thanks{Financial support from the National Natural Science Foundation of China [Grant Numbers 72373007 and 72333001] is gratefully acknowledged. Email: \href{mailto:sxj@gsm.pku.edu.cn}{sxj@gsm.pku.edu.cn}.}\\
    Peking University}
  \date{}
  \maketitle
} \fi

\if0\anon
{
  \bigskip\bigskip\bigskip
  \begin{center}
    {\LARGE\bf Testing for Essential Heterogeneity in the Marginal Treatment Effects Model}
  \end{center}
  \medskip
} \fi

\thispagestyle{empty}

\begin{abstract}
We propose a test for essential heterogeneity in the marginal treatment effects model without imposing additive separability. On the observed support of the propensity score, no essential heterogeneity implies that the conditional mean of the outcome is affine in the propensity score, and this affine restriction is the sharp testable implication of the null. We implement the test using a projected Gaussian-process statistic that is Neyman-orthogonal to the first-step estimation error, so the multiplier bootstrap avoids re-estimating nuisance parameters. Monte Carlo experiments show accurate size control and good power, and, unlike the fixed-scale integrated conditional moment benchmark, the data-driven median rule retains power as the conditioning set grows in dimension. In an application to the Ugandan ELA programme, the omnibus statistic rejects the null of no essential heterogeneity for aspirations at the midline follow-up.

\bigskip
\noindent\keywords{Reproducing kernel Hilbert space; Specification testing; Conditional moment restrictions; Gaussian process; Multiplier bootstrap} 
\end{abstract}

\newpage
\setcounter{page}{1}
\spacingset{1.8}

\section{Introduction}
A valid instrument identifies the local average treatment effect (LATE) for compliers \citep{Imbens1994IdentificationEffects}, an internally valid causal effect for the subpopulation the instrument moves. A broader question is whether this effect generalizes to other populations. This question depends on essential heterogeneity, which arises when individuals select into treatment partly on the basis of unobserved differences in their treatment effects \citep{Heckman2006UnderstandingHeterogeneity,Mogstad2018UsingParameters}. Characterizing essential heterogeneity is therefore a necessary step toward external validity, and testing for it is the natural first step. The benchmark test is due to \citet{Heckman2010TestingModel} (hereafter HSU). When the null of no essential heterogeneity cannot be rejected, the LATE may be read as the ATE or the ATT; \citet{toivanen2016education} apply the HSU test to the returns to engineering education and conduct policy counterfactuals on that basis.

Our first contribution is to extend the test for essential heterogeneity
beyond the additive separability maintained by the benchmark test of
\citet{Heckman2010TestingModel}. Additive separability rules out nonlinear outcome models such as probit and mixed proportional hazard specifications, and even linear models with heteroskedastic errors. Nevertheless, it remains useful when the object of interest is the entire MTE curve, because it can relax the support requirement for the propensity score. For testing, however, it is stronger than necessary. The reason is that the null itself requires the regression of the observed outcome on the covariates and the propensity score to be affine in the propensity score. Under additive separability, the deviation from this affine restriction is restricted to be a function of the propensity score alone. Thus, the same restriction gives rise to two tests with different conditioning sets: conditioning on the covariates together with the propensity score delivers an \emph{omnibus} test, valid under heteroskedasticity and in nonlinear, non-separable outcome equations; conditioning on the propensity score alone delivers the \emph{directional} test that additive separability permits.

Our second contribution is to characterize the relationship between this affine restriction and the null hypothesis of no essential heterogeneity. In the general case without additive separability,
Proposition~\ref{prop:transformed-null} shows that the affine restriction on the observed support of the propensity score is the strongest necessary condition for the absence of essential heterogeneity that the data can test. Corollary~\ref{cor:mte-on-support} shows that this restriction ensures constancy of the MTE on the observed support, establishing the converse direction only on that support, while Corollary~\ref{cor:separable} shows that, under additive separability and full support, it upgrades to an equivalence with the original null itself. In this way, we identify both the sharp support-restricted implication that the data can test and the additional conditions under which that implication coincides with the global hypothesis of no essential heterogeneity.

We overcome two technical obstacles in turning this restriction into a
feasible test. First, the estimated propensity score affects the test through both the residual and the conditioning argument of the kernel. This generated-regressor problem contaminates the limiting distribution, while re-estimating nuisance parameters in every bootstrap replication is computationally costly. We instead project the kernel onto the orthogonal complement of the span of the first-step score. The projected moment is Neyman orthogonal at the true parameter, so the test is locally insensitive to first-step estimation error, allows first-step estimators slower than the parametric rate, and lets the multiplier bootstrap proceed without re-estimating nuisance parameters. The nonstandard step in our setting is to verify that this orthogonality
continues to hold when the estimated propensity score enters not only the
residual but also the conditioning argument of the test statistic.

The second obstacle is the kernel scale. With a fixed scale, our statistic reduces to the classical integrated conditional moment (ICM) test of \citet{bierens1982consistent}, and such fixed-scale implementations deteriorate as the conditioning set grows. We instead use a data-driven median rule that adapts to the dimension of the conditioning variables. Power still declines with dimension, but the difference is substantial. Our application to the ELA programme in Uganda shows that this distinction matters in practice, because the omnibus test detects covariate-loaded departures that directional tests miss and additive separability cannot be taken for granted.

\textbf{Related literature:} A number of tests for essential heterogeneity
have been proposed. The benchmark is due to \citet{Heckman2010TestingModel},
who derive testable implications of the null in the correlated random
coefficient model and construct Wald-type tests based on polynomial
approximations to the MTE curve. In an experimental setting with a binary
instrument, \citet{kowalski2023examine} develops a test for essential
heterogeneity, whereas our procedure accommodates a continuous propensity
score. Finally, \citet{Mogstad2018UsingParameters}, though mainly concerned with the
partial identification of policy-relevant treatment parameters, also
propose a general framework for specification testing in the MTE model. These
tests, however, require additional restrictions on the MTE function.

\citet{escanciano2024gaussian} develops a Gaussian-process framework that makes kernel-based conditional-moment tests Neyman orthogonal and enables a multiplier bootstrap without re-estimating nuisance parameters, in a setting where the conditioning variable is observed. However, the estimated propensity score enters both the residual and the conditioning argument of the kernel in our setting. We show that the projected moment remains orthogonal in this case, so the bootstrap of Theorem~\ref{Theorem: Bootstrap Consistency} is valid when the conditioning variable is estimated. This extension is what the test requires and is established in Supplement~\ref{app:proofs}. 

The projection is one instance of a broader construction. Orthogonal moments trace back to the $C(\alpha)$ tests of \citet{neyman1959optimal}, and the same property organizes the double and debiased machine learning literature \citep{chernozhukov2018double}, the construction of locally robust GMM moments \citep{chernozhukov2022locally}, influence-function characterizations of
orthogonal moments \citep{ichimura2022influence}, and the theory of orthogonal statistical learning \citep{foster2023orthogonal}.

\textbf{Organization of the paper:} Section~\ref{sec:test-framework} sets out the model and derives the conditional moment restriction.
Section~\ref{sec:test} develops the projected test statistic and its multiplier bootstrap. Section~\ref{sec:simulations} reports Monte Carlo evidence,
Section~\ref{sec:application} applies the test to the ELA programme, and
Section~\ref{sec:conclusion} concludes. Proofs and supplementary results are in Supplement~\ref{app:proofs}--\ref{supp:application}.

\section{Testing Framework}\label{sec:test-framework}
\subsection{The MTE Model and Essential Heterogeneity}\label{sec:model-setup}

Let $(Y(1), Y(0))$ denote the potential outcomes of a binary treatment $D \in \{0,1\}$, and let $X \in \mathbb{R}^K$ be a vector of observed covariates. The observed outcome is $Y = D Y(1) + (1-D) Y(0)$. We allow the potential outcomes to be general nonlinear functions of covariates and unobservables,
\begin{equation} \label{Potential Outcomes: binary}
\begin{split}
    Y(1) &= g_1(X, V(1); \beta_1), \\
    Y(0) &= g_0(X, V(0); \beta_0),
\end{split}
\end{equation}
where, for $d \in \{0,1\}$, $V(d)$ is an unobserved random variable and $g_d$ is a known function parameterized by a finite-dimensional vector $\beta_d$. We write $\mu_d(X;\beta_d) = \mathbb{E}[Y(d) \mid X]$ for the implied conditional mean. Note that the potential outcome notation $Y(d)$ already embeds the exclusion restriction that the instruments introduced below do not directly affect outcomes; we make this explicit in Assumption~\ref{ass:exclusion}.

Specification \eqref{Potential Outcomes: binary} nests the textbook linear model $Y(d) = X^\top \beta_d + V(d)$ as its leading case, but two features of specification \eqref{Potential Outcomes: binary} deserve emphasis. First, $g_d$ may be nonlinear in $X$. Second, the unobservable $V(d)$ need not enter additively. This flexibility matters in practice as many workhorse microeconometric models are intrinsically non-separable in their unobservables. Discrete outcomes are the classic case: MTE analyses with probit-type outcome equations go back at least to \citet{AakvikHeckmanVytlacil2005}. Duration outcomes are also important: hazard-based duration models with unobserved heterogeneity are non-separable by construction \citep{AbbringVanDenBerg2003}, instrumental variable methods for such settings have been developed by \citet{vandenBergBonevMammen2020}, and, most recently, \citet{AcerenzaPossebomSantAnna2026} extend the MTE framework itself to
right-censored duration outcomes.

\begin{example}[Binary outcome]\label{ex:probit}
Suppose $Y(d) = \mathbf{1}\{ X^\top \beta_d + V(d) > 0 \}$ with
$V(d) \mid X \sim N(0,1)$, so that $\mu_d(X;\beta_d) = \Phi(X^\top \beta_d)$. Thus, the unobservable $V(d)$ cannot be separated additively from $X$ in the outcome equation: $Y(d) - \mu_d(X;\beta_d)$ is heteroskedastic and, by construction, its distribution depends on $X$.
As we discuss below, this immediately breaks the additive-separability restriction that underlies existing tests for essential heterogeneity.
\end{example}

\begin{example}[Duration outcome]\label{ex:duration}
Suppose $Y(d)$ follows a Weibull mixed proportional hazard model. Conditional on $(X,V(d))$, let its hazard be
\[
\lambda_d(y \mid X,V(d))
= \kappa_d y^{\kappa_d-1}
  \exp(X^\top\beta_d)V(d),
\]
where $V(d)>0$ is an unobserved frailty term. Equivalently, let
$\tilde E_d$ be an idiosyncratic shock with $\tilde E_d \sim \mathrm{Exp}(1)$, and let $(\tilde E_0,\tilde E_1)$ be independent of $(X,V(0),V(1))$. Then
\[
Y(d)
=
\left[
\frac{\tilde E_d}
{\exp(X^\top\beta_d)V(d)}
\right]^{1/\kappa_d}.
\]
The conditional mean is therefore $\mu_d(X;\beta_d)
= C_d \exp\!\left(-X^\top\beta_d/\kappa_d\right)$, where
$C_d=\Gamma(1+1/\kappa_d)\mathbb{E}\!\left[V(d)^{-1/\kappa_d}\right].$
\end{example}

Let $Z = (X, Z_0)$, where $Z_0$ denotes the excluded instruments. Selection into treatment is governed by the threshold-crossing rule
\begin{equation}\label{eq:selection}
    D = \mathbf{1}\big\{ P(Z) \geq U_D \big\},
\end{equation}
where $P(Z) \equiv \Pr(D = 1 \mid Z)$ is the propensity score and $U_D$ is a scalar latent variable, normalized so that $U_D \mid X \sim \mathrm{Uniform}[0,1]$.\footnote{In our implementation, $P(Z)$ is specified parametrically, e.g., a probit or logit index $P(Z;\gamma)$; see Section~\ref{sec:test}. The results in this section do not depend on that choice.}

In equation \eqref{eq:selection}, $P(Z)$ is the observed index of willingness to participate, while $U_D$ collects the unobserved determinants of participation: motivation, ability, and preferences. Normalizing $U_D \mid X \sim \mathrm{Uniform}[0,1]$ makes $U_D$ the quantile of unobserved resistance, so that low values identify the most willing individuals.

Following \citet{Heckman2005StructuralEvaluation}, the marginal treatment effect (MTE) is the mean treatment effect at a given level of observed and unobserved heterogeneity,
\begin{equation}\label{MTE(x,u)}
\mathrm{MTE}(x, u) = \mathbb{E}\big[ Y(1) - Y(0) \mid X = x,\, U_D = u \big].
\end{equation}
What matters is whether the unobserved determinants of participation also
determine the gains from treatment, so that $\mathrm{MTE}(x,u)$ varies with $u$; \citet{Heckman2006UnderstandingHeterogeneity} call this variation
\emph{essential heterogeneity}.

We therefore consider the null hypothesis of no essential heterogeneity,
\begin{equation}\label{H_0: Null}
H_0:\quad \mathrm{MTE}(x,u) = \mathbb{E}\left[ Y(1) - Y(0) \mid X = x \right],
\end{equation}
where the equality holds for $\mathbb{P}_X$-almost every $x$ and
Lebesgue-almost every $u \in (0,1)$. We test $H_0$ against its negation. Under $H_0$, $\mathrm{MTE}(x,u)$ collapses to the
conditional average treatment effect (CATE).

The testable content of $H_0$ is a single restriction on the conditional mean of the observed outcome. Under Assumptions~\ref{ass:exclusion}--\ref{ass:relevance} introduced below, $H_0$ implies
\begin{equation}\label{H_0: Linear Model}
H_0^{(1)}: \quad \mathbb{E}[Y \mid X, P]
= \mu_0(X;\beta_0) + P\,[\mu_1(X;\beta_1) - \mu_0(X;\beta_0)] \quad \text{a.s.},
\end{equation}
which states that $\mathbb{E}[Y \mid X, P]$ is affine in the propensity score, with intercept $\mu_0(X;\beta_0)$ and slope equal to the CATE:
$\Delta(X) = \mu_1(X;\beta_1) - \mu_0(X;\beta_0)$, which may vary with the covariates through the parametric outcome means. Because the left-hand side is an observed conditional mean, the restriction can be
empirically assessed only on the conditional support of $P$ given $X$, which we write
$\mathcal{S}_x$ throughout (defined formally in Assumption~\ref{ass:relevance}).
Accordingly, no test can evaluate the regression outside that support.

Without imposing this parametric structure, the same restriction takes the nonparametric form
\begin{equation}\label{H_0: affine}
H_0^{\mathrm{aff}}: \quad \mathbb{E}[Y \mid X, P] = a(X) + P\,b(X) \quad
\text{a.s. for some measurable } a(\cdot),\,b(\cdot).
\end{equation}
Proposition~\ref{prop:transformed-null} below shows that \eqref{H_0: affine} is the strongest implication of $H_0$ that the data can test. We nevertheless take \eqref{H_0: Linear Model} to the data. It is
\eqref{H_0: affine} with $a(\cdot)$ and $b(\cdot)$ replaced by the parametric
outcome means, as in the applied MTE literature
\citep{Cornelissen2016FromInterventions} and in the benchmark test of
\citet{Heckman2010TestingModel}. Parametric outcome means and a parametric first
stage avoid a nonparametric first step, but the null we test is then joint: a
rejection is evidence against no essential heterogeneity or against the
parametric specification.

\subsection{From the Null to Testable Restrictions}\label{sec:assumptions}
The purpose of this subsection is to translate the null hypothesis in \eqref{H_0: Null} into a restriction on $\mathbb{E}[Y \mid X, P]$ that the data directly reveal. This translation rests on the following maintained assumptions, adapted from \citet{Heckman2005StructuralEvaluation}.

\begin{assumption}[Instrument Exclusion]\label{ass:exclusion}
For each $d \in \{0,1\}$, $Y(d, z_0) = Y(d, z_0') \equiv Y(d)$ for all
$z_0, z_0'$ in the support of $Z_0$, almost surely.
\end{assumption}

\begin{assumption}[Random Assignment]\label{ass:random-assignment}
$(Y(0), Y(1), U_D) \perp\!\!\!\perp Z_0 \mid X$.
\end{assumption}

\begin{assumption}[Threshold-Crossing Selection]\label{ass:selection}
Treatment is determined by \eqref{eq:selection}, where $U_D$ is continuously
distributed conditional on $X$ and normalized so that
$U_D \mid X \sim \mathrm{Uniform}[0,1]$.
\end{assumption}

\begin{assumption}[Propensity Score Variation]\label{ass:relevance}
Conditional on $X$, the distribution of $P(Z)$ is nondegenerate and
absolutely continuous, and $0 < \Pr(D = 1 \mid X) < 1$ almost surely.
For $\mathbb{P}_X$-almost every $x$, its conditional support
$\mathcal{S}_x := \operatorname{supp}\mathbb{P}_{P \mid X = x}$ is a finite
union of closed intervals.
\end{assumption}

Assumptions~\ref{ass:exclusion}--\ref{ass:relevance} turn $H_0$ into a restriction on the observable regression $\mathbb{E}[Y\mid X,P]$. The local instrumental variables argument gives $\mathbb{E}[Y \mid X = x, P = p]= \mu_0(x;\beta_0) + \int_0^p \mathrm{MTE}(x,u)\, du$ for almost every $p$ in the conditional support of $P$ given $X = x$:
\begin{equation}\label{eq:deviation}
    \begin{split}
   \mathbb{E}[Y\mid X,P]&= \mu_0(X;\beta_0) + P\,\Delta(X) + \psi(X,P)\\
\psi(x,p)&= \int_0^p \big[\mathrm{MTE}(x,u) - \Delta(x)\big]\, du  
    \end{split}
\end{equation}
The first two terms are exactly the right-hand side of \eqref{H_0: Linear Model};
$\psi$ is therefore the gap between the regression the data generate and the affine
shape that $H_0$ implies on the observed support. By the local instrumental
variables representation above, the restriction \eqref{H_0: Linear Model} holds
precisely when $\psi \equiv 0$. Two features of $\psi$ matter below: it may load on the covariates as well as on the propensity score, and it is identified only on the conditional support $\mathcal{S}_x$.\footnote{Absolute continuity and nondegeneracy imply that no component of the finite union degenerates to a point; and $\{f_{P \mid X}(\cdot \mid x) > 0\}$ has full $\mathbb{P}_{P \mid X = x}$-measure and is dense in $\mathcal{S}_x$, every nonempty relatively open subset of the support carrying positive $\mathbb{P}_{P \mid X = x}$-mass.}

\begin{assumption}[Additive Separability]\label{ass:additive}
$\mathrm{MTE}(x, u) = \mu_1(x;\beta_1) - \mu_0(x;\beta_0)
+ \mathbb{E}[V(1) - V(0) \mid U_D = u]$.
\end{assumption}

Additive separability is typically justified by two conditions: (i) outcomes that are additive in their unobservables, $Y(d)=\mu_d(X;\beta_d)+V(d)$; and (ii) full independence, $(V(0),V(1),U_D)\perp\!\!\!\perp Z$. Here, condition~(i) fails mechanically in nonlinear models such as Examples~\ref{ex:probit} and~\ref{ex:duration}. In addition, condition~(ii) fails even in linear models with heteroskedastic errors, where the unobservables are scaled by the covariates.

The role of Assumption~\ref{ass:additive} is therefore a simplification. Under it, the deviation in \eqref{eq:deviation} collapses from a family $\psi(x,p)$ to a single function of the propensity score alone:
\[
\psi(x,p)=\phi(p):=\int_0^p h(u)\,du,
\qquad
h(u)=\mathbb{E}[V(1)-V(0)\mid U_D=u].
\]
Accordingly, with
$\varepsilon = Y - \mu_0(X;\beta_0) - P\,\Delta(X)$
\[
\mathbb{E}[\varepsilon \mid X,P]
=
\psi(X,P)
=
\phi(P)
\]
Relative to the general case, additive separability collapses the deviation to a single function and justifies a directional test based only on the propensity score.

Our approach does \emph{not} impose Assumption~\ref{ass:additive}. Proposition~\ref{prop:transformed-null} establishes that $H_0$ implies \eqref{H_0: affine}, and that this implication is sharp: no test based on the data can distinguish $H_0$ from \eqref{H_0: affine}. Corollary~\ref{cor:mte-on-support} shows that \eqref{H_0: affine} yields constancy of the MTE in $u$ on the observed support $\mathcal{S}_x$, a support-restricted converse of
Proposition~\ref{prop:transformed-null}(i). Corollary~\ref{cor:separable} gives a separate result under additive separability: with full support of the propensity score, the directional restriction $\mathbb{E}[\varepsilon\mid P]=0$ is equivalent to $H_0$.

\begin{prop}\label{prop:transformed-null}
Consider the model \eqref{Potential Outcomes: binary} and \eqref{eq:selection}, and suppose Assumptions~\ref{ass:exclusion}--\ref{ass:relevance} and the integrability condition $\mathbb{E}\big[\,|Y(1)| + |Y(0)| \mid X\,\big] < \infty$ a.s.\ hold. Then
\begin{enumerate}
\item[(i)] \emph{(Necessity.)} $H_0$ in \eqref{H_0: Null} implies
\eqref{H_0: Linear Model}.

\item[(ii)] \emph{(Sharpness.)} Conversely, suppose \eqref{H_0: affine} holds, i.e.\ there exist measurable $a(\cdot)$ and $b(\cdot)$ such that
$\mathbb{E}[Y \mid X = x, P = p] = a(x) + p\, b(x)$ for $\mathbb{P}_X$-a.e.\ $x$ and
$\mathbb{P}_{P \mid X = x}$-a.e.\ $p \in \mathcal{S}_x$. Then there exist potential outcomes $(\tilde{Y}(0), \tilde{Y}(1))$ and a selection unobservable $\tilde U_D$ satisfying Assumptions~\ref{ass:exclusion}--\ref{ass:relevance} and $H_0$, with $\tilde D = \mathbf 1\{P(Z) \ge \tilde U_D\}$ and $\tilde Y = \tilde D\,\tilde Y(1) + (1 - \tilde D)\,\tilde Y(0)$, such that $(\tilde Y, \tilde D, Z)$ has the same distribution as $(Y, D, Z)$. One may take $\tilde U_D = U_D$, in which case $\tilde D = D$. Consequently, no test based on the data can distinguish $H_0$ from \eqref{H_0: affine}.
\end{enumerate}
\end{prop}

The two parts together justify the test, and the asymmetry between them is deliberate. Part (i) shows that the restriction is implied by $H_0$ and thus a legitimate target. Part (ii) shows that it is also the best target the data allow. Its statement is nonparametric and covers the entire distribution of the data: no test can distinguish $H_0$ from \eqref{H_0: affine}, so nothing is lost by testing a
conditional mean.\footnote{For bounded outcomes, as in Example~\ref{ex:probit}, $H_0$ has two further testable implications, which our test does not use; see Remark~\ref{rem:bounded} in the Supplement.} Specifically, \eqref{H_0: Linear Model} is the parametric specialization of \eqref{H_0: affine} obtained by setting $a(X)=\mu_0(X;\beta_0)$ and $b(X)=\mu_1(X;\beta_1)-\mu_0(X;\beta_0)$, where the conditional means are implied by the outcome models in \eqref{Potential Outcomes: binary}.

\begin{corollary}\label{cor:mte-on-support}
Suppose the conditions of Proposition~\ref{prop:transformed-null} hold. If \eqref{H_0: affine} holds, i.e., $\mathbb{E}[Y\mid X=x,P=p] = a(x) + p\,b(x)$ for $\mathbb{P}_X$-a.e.\ $x$ and $\mathbb{P}_{P \mid X = x}$-a.e.\ $p \in \mathcal{S}_x$, then $\mathrm{MTE}(x,u) = b(x)$ for Lebesgue-almost every $u \in \mathcal{S}_x$.
\end{corollary}

Corollary~\ref{cor:mte-on-support} establishes the support-restricted converse
of Proposition~\ref{prop:transformed-null}(i): the affine restriction implies
constancy of the MTE on the observed support $\mathcal{S}_x$. When the region
outside the support is Lebesgue-negligible,
$\mathrm{Leb}((0,1)\setminus\mathcal{S}_x)=0$ for $\mathbb{P}_X$-almost every
$x$, this constancy extends to the whole of $(0,1)$ and the restriction is
equivalent to the global $H_0$. Next, Corollary~\ref{cor:separable} shows that under additive separability this
requirement weakens: only the marginal distribution of $P$ needs full support on
$(0,1)$.

\begin{corollary}[Additive separability]\label{cor:separable}
Suppose the conditions of Proposition~\ref{prop:transformed-null} hold and Assumption~\ref{ass:additive} holds with $h(u) = \mathbb{E}[V(1) - V(0) \mid U_D = u] \in L^1(0,1)$. Then
\[
\mathbb{E}[\varepsilon \mid X, P] = \mathbb{E}[\varepsilon \mid P] = \phi(P),
\]
If, moreover, the propensity score has full support on $(0,1)$, i.e.\
$\mathbb{P}(P \in \mathcal{O}) > 0$ for every nonempty open
$\mathcal{O} \subseteq (0,1)$, then
\begin{equation}\label{eq:separable-full}
\mathbb{E}\big[\varepsilon \mid P\big] = 0 \ \text{a.s.}
\quad \Longleftrightarrow \quad
H_0 \ \text{in } \eqref{H_0: Null}.
\end{equation}
\end{corollary}

For implementation, we rewrite the restriction as a residual-based conditional moment restriction. Define
\begin{equation}\label{eq:residual}
\varepsilon(\theta) = Y - \mu_0(X;\beta_0)
- P(Z;\gamma)\,\big[\mu_1(X;\beta_1) - \mu_0(X;\beta_0)\big],
\end{equation}
where $\theta = (\beta_0^\top,\beta_1^\top,\gamma^\top)^\top \in \Theta$. The
null taken to the data is
\begin{equation}\label{CMR restriction}
H_0^{(2)}:\quad \exists\,\theta_0 \in \Theta \ \text{such that} \
\mathbb{E}\big[\varepsilon(\theta_0)\mid W\big] = 0 \quad \text{a.s.},
\end{equation}
with $W = (X^\top, P)^\top$ in the general case\footnote{In the simulations
and the application we evaluate the kernel on the equivalent feature
$P(Z)\,(1, X^\top)^\top$ rather than on $(X^\top, P)^\top$. Because the
propensity score lies in $(0,1)$ almost surely, the map
$(x,p) \mapsto p(1,x)$ is injective on the support of $(X,P)$; hence it
generates the same conditioning $\sigma$-field as $(X,P)$, and the test
characterizes the same conditional moment restriction
\eqref{CMR restriction}.} and $W = P$ under
Assumption~\ref{ass:additive}; these are the omnibus and directional
conditioning sets, respectively. Both use the residual $\varepsilon(\theta)$
of \eqref{eq:residual} and differ only in $W$.

\begin{remark}[Interpretation of a rejection]\label{rem:joint}
Conditional on the MTE identifying assumptions, the null \eqref{CMR restriction} is a conjunction of three components: no essential heterogeneity, correct specification of the outcome means $\mu_d(X;\beta_d)$, and correct specification of the first stage $P(Z;\gamma)$. Only the first is the object of interest; the other two are maintained, so a rejection may also reflect either of them. In practice, the first stage can be screened by the specification tests of \citet{sant2019specification}, and the outcome means can be enriched without altering the procedure, as in Section~\ref{sec:application}.
\end{remark}

\section{Testing Procedure}\label{sec:test}

\subsection{A Projected Gaussian Process Test Statistic}\label{sec:statistic}

Our test is built directly on the conditional moment restriction
\eqref{CMR restriction}. Momentarily treating $\theta_0$ as known, we consider the
quadratic-form statistic
\begin{equation}\label{Test Statistics}
  \hat{T}^2_n(\theta_0) = \frac{1}{n} \sum_{i=1}^n \sum_{j=1}^n
  \varepsilon_i(\theta_0) \, k(W_i, W_j) \, \varepsilon_j(\theta_0),
\end{equation}
where $k(\cdot,\cdot)$ is an integrally strictly positive definite (ISPD) kernel:
\begin{equation}\label{eq:ispd}
    \int\!\!\int f(w) \, k(w, w') \, f(w') \, d\mathbb{P}_W(w)\, d\mathbb{P}_W(w') > 0
    \quad \text{for all nonzero } f \in L^2(\mathbb{P}_W).
\end{equation}

Leading examples of bounded ISPD kernels on $\mathbb{R}^{d_W}$ include the Gaussian kernel $k(w,w') = \exp(-\lVert w - w'\rVert_2^2/\sigma)$ and the Laplacian kernel $\exp(-\lVert w - w'\rVert_1/\sigma)$; see \citet{Sriperumbudur2011Universality} for characterizations of ISPD kernels. By contrast, polynomial kernels generate finite-dimensional RKHSs and are \emph{not} ISPD when the law of $W$ has infinite support: a test built on them checks only finitely many unconditional moments and cannot be consistent against all alternatives. The ISPD requirement is the RKHS counterpart of a condition with a long history in econometrics: the ICM literature asks which families of weight functions ``reveal'' all deviations from a conditional moment restriction. These families range from complex exponentials in \citet{bierens1982consistent} and logistic weights in \citet{bierens1997asymptotic} to indicator functions in \citet{Stute1997Nonparametric} and projections in \citet{Escanciano2006Projections}. \citet{StinchcombeWhite1998} provide a
general characterization of such revealing classes, and \citet{escanciano2024gaussian} makes the correspondence between these weight families and kernels explicit within the Gaussian process framework we adopt here.

The kernel $k$ is a reproducing kernel, not a smoothing device. In kernel
regression--based specification tests, the bandwidth must shrink to zero and $\mathbb{E}[\varepsilon \mid W]$ is estimated nonparametrically, at
dimension-dependent rates. Here, the ISPD property makes the reproducing kernel Hilbert space (RKHS) $\mathcal{H}_k$ generated by $k$ sufficiently rich to characterize the conditional moment restriction exactly, and no smoothing enters the construction. The following result, adapted from
\citet{escanciano2024gaussian}, formalizes this.

\begin{theorem}[Residual embedding]\label{Theorem: Residual Embedding}
Suppose $\mathbb{E}(\varepsilon(\theta_0)\mid W) \in L^2(\mathbb{P}_W)$,
$\mathbb{E}[\varepsilon(\theta_0)^2] < \infty$, and $k$ is continuous, ISPD, and bounded. Then $\mathbb{E}(\varepsilon(\theta_0)\mid W) = 0$ almost surely if and only if $\boldsymbol{m}_{\varepsilon(\theta_0)}
\equiv \mathbb{E}\left[ \varepsilon(\theta_0)\, k(W, \cdot) \right]
= \boldsymbol{0}$ in $\mathcal{H}_k$, in which case $T^2(\theta_0) \equiv
\lVert \boldsymbol{m}_{\varepsilon(\theta_0)} \rVert_{\mathcal{H}_k}^2
= \mathbb{E}\left[ \varepsilon(\theta_0)\, k(W, W')\, \varepsilon(\theta_0)' \right]
= 0$, where $(W', \varepsilon(\theta_0)')$ is an independent copy of
$(W, \varepsilon(\theta_0))$.
\end{theorem}

Theorem~\ref{Theorem: Residual Embedding} identifies \eqref{Test Statistics} as the sample analog of $\lVert \boldsymbol{m}_{\varepsilon(\theta_0)}\rVert_{\mathcal{H}_k}^2$: the
conditional moment restriction is summarized by a single element of
$\mathcal{H}_k$, and the statistic is a double sum in the data, with no
nonparametric estimation of $\mathbb{E}[\varepsilon \mid W]$. Because nothing is smoothed, it converges at the parametric rate, regardless of the dimension of $W$, which is what makes the omnibus conditioning set $W = (X^\top,P)^\top$ workable; finite-sample power still declines with dimension, as Section~\ref{sec:simulations} shows. The scale $\sigma$ is correspondingly a tuning constant and not a bandwidth: the asymptotics impose no rate on it. Throughout we use the Gaussian kernel $k(W,W') = \exp(-\lVert W - W'\rVert_2^2/\sigma)$ with the median heuristic $\hat\sigma = \mathrm{med}\{\lVert W_i - W_j\rVert_2^2 : i < j\}$.\footnote{The median is taken on the squared-distance scale, so that the kernel's exponent equals $-1$ at the median pair of observations. The rule goes back to \citet[Sec.~8]{Gretton2012Kernel}; the same squared-median scaling under this kernel parameterization is used by, e.g., \citet{Liu2016Stein}.} Remark~\ref{rem:kernel-choice} discusses the choice.

In practice $\theta_0$ is unknown, and replacing it by an estimator $\hat\theta_n$ contaminates the limit distribution of \eqref{Test Statistics}: the residual depends on $(\beta_0,\beta_1)$, and $W$ itself contains the generated regressor $P(Z;\gamma)$. We deal with both estimation effects through a projection, under the following conditions.

\begin{assumption}\label{Assumption: parameter space}
(i) $\Theta$ is a compact convex subset of $\mathbb{R}^{2K+L}$; (ii) the 
true parameter $\theta_0$ is an interior point of $\Theta$; and (iii) the estimator $\hat{\theta}_n$ satisfies $\lVert \hat{\theta}_n - \theta_0 \rVert = O_p (n^{-\alpha})$ for some $\alpha > 1/4$ under the maintained null and under the local alternatives of Theorem~\ref{thm:asymptotics}(iii), and under a fixed alternative with
probability limit $\theta^\ast$ of $\hat{\theta}_n$, $\lVert \hat{\theta}_n - \theta^\ast \rVert = o_p(1)$. 
\end{assumption}

\begin{assumption}\label{Assumption: residual regularity}
(i) $\varepsilon(\theta)$ is three times continuously differentiable in $\theta$ for every observation, with
$\mathbb{E}\big(\sup_{\theta \in \Theta} |\varepsilon(\theta)|^4\big) < \infty$,
$\mathbb{E}\big(\sup_{\theta \in \Theta} \lVert \nabla_{\theta} \varepsilon(\theta) \rVert^4\big) < \infty$,
$\mathbb{E}\big(\sup_{\theta \in \Theta} \lVert \nabla^2_{\theta} \varepsilon(\theta) \rVert^4\big)
< \infty$, and
$\mathbb{E}\big(\sup_{\theta \in \Theta} \lVert \nabla^3_{\theta} \varepsilon(\theta) \rVert^2\big)
< \infty$; the map $\gamma \mapsto P(Z;\gamma)$ is Lipschitz with
$\mathbb{E}\big(\sup_\gamma \lVert \partial_\gamma P(Z;\gamma) \rVert^2\big) < \infty$;
(ii) $\tilde\Gamma_{\theta} = \mathbb{E}\left[\nabla_{\theta} \varepsilon(\theta) \nabla_{\theta} \varepsilon(\theta)^\top\right]$ is
nonsingular in a neighborhood of $\theta_0$; this holds a fortiori whenever the conditional-score covariance $\Gamma_{\theta_0} = \mathbb{E}[g_{\theta_0}(W) g_{\theta_0}(W)^\top]$, with $g_{\theta_0}(w) = \mathbb{E}[\nabla_{\theta} \varepsilon(\theta_0) \mid W = w]$, is
nonsingular, since $\tilde\Gamma_{\theta_0} = \Gamma_{\theta_0} +
\mathbb{E}[\operatorname{Var}(\nabla_{\theta} \varepsilon(\theta_0) \mid W)]\succeq \Gamma_{\theta_0}$.
\end{assumption}

Assumption~\ref{Assumption: parameter space}(iii) is notably weak: it requires only a weaker convergence rate for $\hat\theta_n$, not an asymptotically linear (influence-function) representation. Throughout, the covariate dimension $K$ is fixed, and in this regime the outcome and propensity score parameters are root-$n$-estimable. The rate condition would also tolerate first steps that converge slower than $n^{-1/2}$. Note that regularized estimators such as the LASSO, the natural first steps when $K$ grows with the sample size, satisfy the rate condition in principle, but extending the theory to such high-dimensional regimes is beyond the scope of this paper. Assumption~\ref{Assumption: residual regularity} imposes standard smoothness and moment conditions on the residual and its first three $\theta$-derivatives, as well as nondegeneracy of the raw-score covariance.

The projection is defined as follows. Let $\mathsf{U}=(Y, Z)$ denote a generic observation and $\mathsf{U}'$ an independent copy; for $\theta \in \Theta$, write
$\tilde g_\theta(\mathsf{U}) = \nabla_\theta \varepsilon(\theta)$ for the raw score, a measurable function of $Z$ alone whose $\gamma$-block is
$-\partial_\gamma P(Z; \gamma)\, [\mu_1(X; \beta_1) - \mu_0(X; \beta_0)]$ and which is therefore not a function of $W$. Define
$\tilde b_\theta(w') = \mathbb{E}[\tilde g_\theta(\mathsf{U}) k(W, w')]$,
$\tilde\Gamma_\theta = \mathbb{E}[\tilde g_\theta(\mathsf{U})\tilde g_\theta(\mathsf{U})^\top]$, and $\tilde M_\theta = \mathbb{E}[\tilde g_\theta(\mathsf{U})\tilde g_\theta(\mathsf{U}')^\top k(W, W')]$. At every $\theta$ with nonsingular $\tilde\Gamma_\theta$, the projected kernel feature is the $\mathcal{H}_k$-element
\begin{equation}\label{eq:projected-kernel}
    (\Pi_\theta k)(\mathsf{u}, \cdot) = k(W, \cdot)
    - \tilde g_\theta(\mathsf{u})^\top \tilde\Gamma_\theta^{-1}\, \tilde b_\theta(\cdot),
\end{equation}
which subtracts from $k(W, \cdot)$ a linear combination of the score-weighted embeddings $\tilde b_\theta(\cdot)$; subscripts are dropped at
$\theta = \theta_0$, and the induced bivariate kernel is $\tilde k_{p,\theta}(\mathsf{u}, \mathsf{u}') = \langle (\Pi_\theta k)(\mathsf{u}, \cdot), (\Pi_\theta k)(\mathsf{u}', \cdot) \rangle$ with $\tilde h_{p,\theta}(\mathsf{u}, \mathsf{u}') = \varepsilon(\theta)\, \tilde k_{p,\theta}(\mathsf{u}, \mathsf{u}')\, \varepsilon(\theta)'$.
The resulting moment functional $m(\theta) = \mathbb{E}[\varepsilon(\theta) (\Pi k)(\mathsf{U}, \cdot)]$, with the projection frozen at $\theta_0$, is Neyman orthogonal: $\partial_\theta\, m(\theta)\big|_{\theta=\theta_0} = \mathbb{E}[\tilde g(\mathsf{U}) (\Pi k)(\mathsf{U}, \cdot)] = \boldsymbol{0}$ in $\mathcal{H}_k$. The verification, together with the four-term form of the bivariate kernel $\tilde k_p$, is given in Supplement~\ref{app:proofs} (Lemma~\ref{lem:orthogonal}).

Intuitively, the projection removes those components of the test that could be contaminated by estimation error in $\hat\theta_n$; what remains is locally insensitive to how well $\theta_0$ is estimated. This is what allows the slow rates permitted by Assumption~\ref{Assumption: parameter space}(iii) and, as we exploit in Section~\ref{sec:procedure}, what renders re-estimation of $\theta$ inside the bootstrap unnecessary.

The feasible statistic replaces $k$ by the projected feature $(\Pi k)$ and $\theta_0$ by $\hat\theta_n$:
\begin{equation}\label{Test Statistics local Projection}
        n \hat T_{np}^2(\hat{\theta}_n)
        = \frac{1}{n}\, \boldsymbol{\varepsilon}(\hat\theta_n)^\top
        \hat{\boldsymbol{\Pi}}^\top \hat{\boldsymbol{K}} \hat{\boldsymbol{\Pi}}\,
        \boldsymbol{\varepsilon}(\hat\theta_n),
\end{equation}
where $\boldsymbol{\varepsilon}(\hat\theta_n) =
(\varepsilon_1(\hat\theta_n),\ldots,\varepsilon_n(\hat\theta_n))^\top$,
$(\hat{\boldsymbol{K}})_{ij} = k(\hat W_i,\hat W_j)$ with
$\hat W_i = (X_i^\top, P(Z_i;\hat\gamma_n))^\top$, and
$\hat{\boldsymbol{\Pi}} = \boldsymbol{I}_{n}
- \hat{\mathbf{g}} (\hat{\mathbf{g}}^\top \hat{\mathbf{g}})^{-1}
\hat{\mathbf{g}}^\top$ is the sample projection matrix built from the estimated raw scores $\hat g_{i} = \nabla_{\theta} \varepsilon_{i}(\theta)\big|_{\theta = \hat\theta_n}$, the sample analogs of $\tilde g(\mathsf{U})$. Algorithm~\ref{alg:test} is unchanged by this notation.

The next theorem collects the large-sample properties of
\eqref{Test Statistics local Projection} under the null, under fixed alternatives, and the Pitman-type local alternatives $H_{1n}: \mathbb{E}\left[\varepsilon(\theta_0) \mid W\right] = R(W)/\sqrt{n}$
with $\mathbb{E}[R^4(W)] < \infty$.

\begin{theorem}\label{thm:asymptotics}
Suppose Assumptions~\ref{ass:exclusion}--\ref{ass:relevance},
\ref{Assumption: parameter space}, and
\ref{Assumption: residual regularity} hold, and suppose
$\mathbb{E}\big[\tilde h_{p,\theta}(\mathsf{U}, \mathsf{U}')^2\big] < \infty$ for all $\theta$ in a neighborhood of $\theta_0$. Then

(i) (\emph{Null distribution}) if $H_0^{(2)}$ is true, the first stage is
correctly specified, $P(Z;\gamma_0) = P(Z)$ a.s., and
$\mathbb{E}[\varepsilon(\theta_0) \mid X, P] = 0$ a.s., then
$n \hat{T}_{np}^2(\hat\theta_n) \overset{d}{\longrightarrow}
\sum_{j=1}^\infty \mu_j A_j^2$,
where $\{A_j\}_{j\geq 1}$ are independent standard normals and
$\{\mu_j\}_{j \geq 1}$ are the eigenvalues of $\tilde h_p$;

(ii) (\emph{Consistency}) under any fixed alternative with
$\theta^\ast = \operatorname{plim} \hat\theta_n$ such that
$\tilde\Gamma_{\theta^\ast}$ is nonsingular,
$\mathbb{E}[\tilde h_{p,\theta^\ast}(\mathsf{U}, \mathsf{U}')^2] < \infty$, and there is no
$v \in \mathbb{R}^{2K+L}$ with
$\mathbb{E}[\varepsilon(\theta^\ast) \mid W] = v^\top \tilde g_{\theta^\ast}(\mathsf{U})$
a.s.,
$\Pr\big( n\hat{T}_{np}^2(\hat\theta_n) > t \big) \to 1$ for any $t > 0$;

(iii) (\emph{Local power}) under $H_{1n}$ and the first-order $Z$-centering
$\mathbb{E}[\varepsilon(\theta_0) \mid Z] = \mathbb{E}[R(W) \mid Z]/\sqrt{n} + o(n^{-1/2})$
in $L^2(\mathbb{P})$ (automatic when the perturbation enters additively with
$Z$-centered base residual),
$n \hat{T}_{np}^2(\hat\theta_n) \overset{d}{\longrightarrow}
\sum_{j=1}^\infty \mu_j A_j^2 + 2 B_p
+ \mathbb{E}\big[ R(W)\, \tilde k_p(\mathsf{U}, \mathsf{U}')\, R(W)' \big]$,
where $B_p$ is a centered normal random variable whose variance is given in
Supplement~\ref{app:proofs}.
\end{theorem}

Part (i) shows that, thanks to the projection, the feasible statistic has the same
limiting null distribution as if $\theta_0$ were known. The centering condition in
part (i) is automatic under the joint null of Remark~\ref{rem:joint}: $H_0$
together with the parametric outcome means and the correct first stage gives
$\mathbb{E}[\varepsilon(\theta_0) \mid X, P] = 0$ by
Proposition~\ref{prop:transformed-null}(i). Lemma~\ref{lem:centering} in
Supplement~\ref{app:proofs} then upgrades it to
$\mathbb{E}[\varepsilon(\theta_0) \mid Z] = 0$, which is what the raw-score
projection requires. The projection is not
innocuous for power in one specific direction, however: by construction, the
projected feature annihilates the span of the score-weighted embeddings, so the
test is blind exactly to the fixed alternatives excluded in part~(ii), namely
those whose conditional mean coincides a.s.\ with a linear combination of the raw
scores, and it has nontrivial local power in part~(iii) only along directions
outside the same finite-dimensional set. To first order, such score-proportional
directions are reparameterizations of the parametric null rather than departures
attributable to essential heterogeneity. The limit in (i) is
non-pivotal, however, as the eigenvalues $\{\mu_j\}$ depend on the data distribution;
the next subsection provides a simple bootstrap that sidesteps their estimation.

\subsection{Bootstrap Critical Values and Implementation}\label{sec:procedure}

Because the null limit in Theorem~\ref{thm:asymptotics}(i) is not pivotal, we approximate its critical values by a multiplier bootstrap. A key practical payoff of the projection is that the bootstrap does \emph{not} require re-estimating $\theta$ in any replication: since the projected moment is locally robust to $\theta$, perturbing the residuals alone reproduces the correct null limit. Each replication therefore costs one matrix quadratic form rather than a full re-estimation of the outcome and propensity score models. Specifically, letting $\{V_i\}_{i=1}^n$ be i.i.d.\ multipliers with $\mathbb{E}(V_i) = 0$ and
$\mathbb{E}(V_i^2) = 1$ (e.g., Rademacher, standard normal, or the two-point distribution of \citealp{mammen1993bootstrap}), drawn independently of the data, the bootstrap statistic is
\begin{equation}\label{eq:bootstrap-stat}
    n \hat T_{np}^{2\ast}(\hat\theta_n)
    = \frac{1}{n} \big(\boldsymbol{\varepsilon}(\hat\theta_n)^{\ast}\big)^\top
    \hat{\boldsymbol{\Pi}}^\top \hat{\boldsymbol{K}} \hat{\boldsymbol{\Pi}}\,
    \boldsymbol{\varepsilon}(\hat\theta_n)^{\ast},
    \qquad
    \boldsymbol{\varepsilon}(\hat\theta_n)^{\ast}
    = \big(\varepsilon_1(\hat\theta_n) V_1, \ldots,
    \varepsilon_n(\hat\theta_n) V_n\big)^\top.
\end{equation}

The full procedure is summarized as follows.

\begin{algorithm}\label{alg:test}
\emph{Input:} data $\{(Y_i, D_i, Z_i)\}_{i=1}^n$; kernel scale $\sigma$;
number of bootstrap replications $B$; nominal level $\alpha$.
\begin{enumerate}
    \item Estimate $\hat{\theta}_n = (\hat\beta_{0n}^\top, \hat\beta_{1n}^\top,
    \hat\gamma_n^\top)^\top$; form $\hat P_i = P(Z_i;\hat\gamma_n)$, the residuals
    $\varepsilon_i(\hat\theta_n)$, and $\hat W_i$.
    \item Compute $\hat{\boldsymbol{K}}$ with the Gaussian kernel and the projection
    matrix $\hat{\boldsymbol{\Pi}}$; form the statistic
    $n\hat T_{np}^2(\hat\theta_n)$ in \eqref{Test Statistics local Projection}.
    \item For $b = 1, \ldots, B$: draw multipliers $\{V_i^b\}_{i=1}^n$ and compute
    $n \hat T_{np}^{2\ast b}$ as in \eqref{eq:bootstrap-stat}.
    \item Reject $H_0^{(2)}$ at level $\alpha$ if $n\hat T_{np}^2(\hat\theta_n)$
    exceeds the empirical $(1-\alpha)$-th quantile $c_{B,\alpha}^{\ast}$ of $\{n \hat T_{np}^{2\ast b}\}_{b=1}^B$; equivalently, report the $p$-value
    $B^{-1} \sum_{b=1}^B \mathbf{1}\{ n\hat T_{np}^{2\ast b}
    > n \hat T_{np}^2(\hat\theta_n) \}$.
\end{enumerate}
\end{algorithm}

\begin{remark}[Choice of kernel and scale]\label{rem:kernel-choice}
Any bounded ISPD kernel yields a consistent test, but $\sigma$ directs power. A large $\sigma$ weights smooth, low-frequency departures from the null, whereas a small $\sigma$ weights localized, high-frequency ones. No choice is uniformly best \citep{Escanciano2009Lack}. In this respect, the fixed-scale statistic and local-smoothing tests are less far apart than they appear, since fixed-bandwidth kernel specification tests are themselves ICM-type \citep{fan2000consistent}. The median heuristic adapts the scale to the dimension and spread of $W$ \citep{Gretton2012Kernel, Garreau2017Median}, and \citet{escanciano2024gaussian} documents this
for the Gaussian kernel with $\sigma = 2$, the classical ICM implementation of \citet{bierens1982consistent}. Our simulations report
$\sigma \in \{\mathrm{med}, 1, 2\}$ and match this pattern; aggregating over a grid of scales is a further option \citep{Schrab2023MMDAgg}.
\end{remark}

We next establish the validity of the bootstrap critical values. The following theorem shows that the multiplier bootstrap distribution consistently estimates the null limit of Theorem~\ref{thm:asymptotics}(i), and hence that the test in Algorithm~\ref{alg:test} has asymptotic size $\alpha$ and power approaching one under any fixed alternative.

\begin{theorem}\label{Theorem: Bootstrap Consistency}
Suppose Assumptions~\ref{ass:exclusion}--\ref{ass:relevance},
\ref{Assumption: parameter space}, and
\ref{Assumption: residual regularity} hold, suppose
$\mathbb{E}\big[\tilde h_p(\mathsf{U}, \mathsf{U}')^2\big] < \infty$, and let $H_0^{(2)}$ be true
with the first stage correctly specified, $P(Z;\gamma_0) = P(Z)$ a.s., and
$\mathbb{E}[\varepsilon(\theta_0) \mid X, P] = 0$ a.s. Then, for any
$\alpha \in (0,1)$,
$c_{B,\alpha}^{\ast} \xrightarrow{p} c_\alpha$ as $n, B \to \infty$, where
$c_\alpha$ is the $(1-\alpha)$ quantile of the null limit in
Theorem~\ref{thm:asymptotics}(i). Consequently, the test in
Algorithm~\ref{alg:test} has asymptotic size $\alpha$ and is consistent against the
alternatives in Theorem~\ref{thm:asymptotics}(ii).
\end{theorem}

\section{Simulations}\label{sec:simulations}

This section studies the size and power performance of the projected statistic
across conditioning sets and Gaussian scales, with particular attention to how
its performance changes with the dimension and dispersion of $W$.

\subsection{Design and Implementation}\label{sec:sim-design}

Throughout this section, $T^{W}_{\sigma}$ denotes the projected statistic $n\hat T^2_{np}$ of Section~\ref{sec:test}, implemented with conditioning set $W \in \{(X^\top,P)^\top,\, P\}$ and Gaussian scale $\sigma \in \{\mathrm{med},1,2\}$. The key distinction is that $\hat\sigma_{\mathrm{med}} = \mathrm{med}\{\lVert W_i - W_j\rVert_2^2 : i<j\}$ rescales with the dimension and dispersion of $W$, whereas the fixed scales $\sigma \in \{1,2\}$ do not. The benchmark $T^{W}_{2}$ coincides with the classical ICM implementation of \citet{bierens1982consistent}. The main text reports both the omnibus family $\TXP{\sigma}$ and the directional family $\TP{\sigma}$; the HSU polynomial Wald benchmarks $\Qpoly{2}$ and $\Qpoly{5}$ appear only in Supplement~\ref{supp:additional-sims}.

All designs share the selection equation
\[
D = \mathbf{1}\{ Z_0^\top\gamma_0 + X^\top\gamma_1 \ge U \},
\qquad U \sim N(0,1),\ U \perp (X, Z_0),
\]
so that $P(Z) = \Phi(Z_0^\top\gamma_0 + X^\top\gamma_1)$ and $U_D = \Phi(U)$. The covariates are
$X = (X_1,\ldots,X_K)^\top$ with $X_k \stackrel{i.i.d.}{\sim} N(0,\sigma_X^2)$,
and $Z_0 \sim N(0, I_2)$ is drawn independently of $X$. The baseline sets $K = 5$ and $\sigma_X = 1$. In designs with $K > 5$, the first five coordinates
form a signal block, and the remaining $K - 5$ are nuisance covariates with
zero loadings throughout.

Essential heterogeneity is generated by selection on unobserved gains: for each
direction $j \in \{1,2,3,4\}$,
\[
\eta_j = \big[\,c_1 + c_2\, S_j(X)\,\big]\, U,
\]
with $c_1 = c_2 = 0$ under the null, where $c_1$ and $c_2$ scale a component depending on
the $P$ alone and the covariates $X$. The covariate component is
$S_j(X) = \sum_{k=1}^{5} \tilde f_j(X_k)$, where
\[
f_1(x) = x, \qquad
f_2(x) = e^{x/2} - \mathbb{E}\, e^{G/2}, \qquad
f_3(x) = x^2 - \sigma_X^2, \qquad
f_4(x) = \sin\!\big(2x/\sigma_X\big),
\]
with $G$ an independent copy of the covariate vector $X$ and
$\tilde f_j = f_j / \mathrm{sd}\{f_j(G_j)\}$; hence $\mathrm{Var}(S_j) = 5$ for
every direction and every $\sigma_X$. The four directions represent linear,
exponential, quadratic, and oscillatory departures, and the power tables report
$c_1 \in \{0,3\}$ and $c_2 \in \{0.5,1,2\}$. Because the loading $(c_1+c_2S_j(X))$ also multiplies a mean-zero
individual-specific term, raising $c_1$ or $c_2$ inflates the residual variance
along with the deviation; comparisons across $c_2$ are therefore comparisons of
signal to noise rather than of the size of the alternative alone.

We embed this common heterogeneity term in two outcome equations, for
$d \in \{0,1\}$:
\begin{align}
\textbf{Design L:}\quad
& Y(d) = \alpha_d + X^\top\beta_d + d\cdot\eta + V(d),
  \label{eq:design-L}\\
\textbf{Design N:}\quad
& Y(d) = \mathbf{1}\big\{\alpha_d + X^\top\beta_d + d\cdot\eta + V(d) > 0
  \big\},
  \label{eq:design-N}
\end{align}
with $(V(0), V(1))$ standard normal, independent of 
$(X, Z_0, U)$. Parameter values are
$\beta_1=(1,0,0,0,1)^\top$, $\beta_0=(1,1,1,1,1)^\top$,
$(\alpha_1,\alpha_0)=(1,0)$, $\gamma_1=(0.5,0,0,0,0.5)^\top$, and
$\gamma_0=(0.5,0.5)^\top$.

In each of the $1{,}000$ Monte Carlo replications, we estimate $(\beta_0,\beta_1)$ by least squares (Design~L) or maximum likelihood
(Design~N), estimate the propensity score by probit, and apply Algorithm~\ref{alg:test} with $B=500$ multiplier draws from the two-point
distribution of \citet{mammen1993bootstrap}. Each design is reported at the
three scales $\sigma \in \{\mathrm{med},1,2\}$ that index the statistics
$T^{W}_{\sigma}$. 

Notice that the median heuristic is based on the pairwise distances
$\lVert W_i - W_j\rVert_2^2$, which grow with the dispersion $\sigma_X$ of the
covariates: the realized $\hat\sigma_{\mathrm{med}}$ is $2.56$ at the baseline
$\sigma_X = 1$, close to the fixed $\sigma = 2$, and $0.73$ and $10.92$ at
$\sigma_X = 0.5$ and $2$, respectively, both far from the fixed benchmarks.
Reporting $\sigma_X \in \{0.5,1,2\}$ therefore separates the scale choice from
the remaining design dimensions.

\subsection{Simulation Results}\label{sec:sim-results}

Table~\ref{main:tab:size} reports empirical sizes for the omnibus and directional statistics in the probit design (Design N) across covariate scales
$\sigma_X \in \{0.5,1,2\}$ and Gaussian scales $\sigma \in \{\mathrm{med},1,2\}$. Rejection frequencies are close to nominal throughout, so size control is satisfactory for both statistics across both dimensions. The corresponding linear-design results are reported in Supplement~\ref{supp:additional-sims}.

\begin{table}[!tbp]
\caption{Empirical size under the probit design.}
\label{main:tab:size}
\centering
\singlespacing
\setlength{\tabcolsep}{3.5pt}
\footnotesize
\begin{threeparttable}
\begin{tabular*}{\textwidth}{@{\extracolsep{\fill}}ll*{9}{r}@{}}
\toprule
& & \multicolumn{3}{c}{$\sigma_X=0.5$} & \multicolumn{3}{c}{$\sigma_X=1$} & \multicolumn{3}{c}{$\sigma_X=2$}\\
\cmidrule(lr){3-5}\cmidrule(lr){6-8}\cmidrule(lr){9-11}
$\alpha$ & $n$ & $\mathrm{med}$ & $1$ & $2$ & $\mathrm{med}$ & $1$ & $2$ & $\mathrm{med}$ & $1$ & $2$\\
\midrule
\multicolumn{11}{l}{\textit{Panel A}: omnibus $T^{X\!P}_\sigma$}\\
\midrule
\addlinespace[1pt]
0.1 & 500 & 0.118 & 0.119 & 0.114 & 0.116 & 0.098 & 0.105 & 0.107 & 0.117 & 0.111\\
 & 1,000 & 0.108 & 0.103 & 0.100 & 0.093 & 0.090 & 0.090 & 0.114 & 0.103 & 0.093\\
 & 2,000 & 0.092 & 0.087 & 0.103 & 0.095 & 0.095 & 0.093 & 0.110 & 0.108 & 0.116\\
\addlinespace[2pt]
0.05 & 500 & 0.051 & 0.050 & 0.059 & 0.057 & 0.042 & 0.053 & 0.053 & 0.067 & 0.063\\
 & 1,000 & 0.054 & 0.053 & 0.049 & 0.046 & 0.053 & 0.045 & 0.058 & 0.055 & 0.060\\
 & 2,000 & 0.055 & 0.054 & 0.043 & 0.045 & 0.046 & 0.045 & 0.058 & 0.060 & 0.063\\
\addlinespace[2pt]
0.01 & 500 & 0.017 & 0.015 & 0.013 & 0.014 & 0.009 & 0.012 & 0.014 & 0.014 & 0.016\\
 & 1,000 & 0.008 & 0.009 & 0.006 & 0.009 & 0.010 & 0.009 & 0.011 & 0.012 & 0.014\\
 & 2,000 & 0.011 & 0.009 & 0.009 & 0.009 & 0.010 & 0.006 & 0.013 & 0.011 & 0.013\\
\midrule
\multicolumn{11}{l}{\textit{Panel B}: directional $T^{P}_\sigma$}\\
\midrule
\addlinespace[1pt]
0.1 & 500 & 0.092 & 0.105 & 0.106 & 0.096 & 0.113 & 0.108 & 0.109 & 0.099 & 0.095\\
 & 1,000 & 0.106 & 0.099 & 0.095 & 0.101 & 0.087 & 0.095 & 0.111 & 0.113 & 0.112\\
 & 2,000 & 0.083 & 0.094 & 0.095 & 0.090 & 0.094 & 0.092 & 0.106 & 0.094 & 0.095\\
\addlinespace[2pt]
0.05 & 500 & 0.048 & 0.056 & 0.063 & 0.047 & 0.048 & 0.058 & 0.057 & 0.044 & 0.046\\
 & 1,000 & 0.057 & 0.052 & 0.050 & 0.048 & 0.039 & 0.042 & 0.064 & 0.056 & 0.064\\
 & 2,000 & 0.044 & 0.042 & 0.038 & 0.052 & 0.051 & 0.045 & 0.051 & 0.047 & 0.050\\
\addlinespace[2pt]
0.01 & 500 & 0.010 & 0.010 & 0.014 & 0.006 & 0.012 & 0.009 & 0.011 & 0.015 & 0.013\\
 & 1,000 & 0.008 & 0.011 & 0.007 & 0.012 & 0.008 & 0.007 & 0.013 & 0.013 & 0.009\\
 & 2,000 & 0.009 & 0.012 & 0.009 & 0.015 & 0.008 & 0.008 & 0.007 & 0.007 & 0.005\\
\bottomrule
\end{tabular*}
\begin{tablenotes}[flushleft]\scriptsize
\item \emph{Note}: Empirical size over $1{,}000$ replications with $B = 500$
multiplier-bootstrap draws; $K = 5$ and $c_1=c_2=0$. Columns
$\mathrm{med}$, $1$, and $2$ denote the Gaussian scale $\sigma$.
\end{tablenotes}
\end{threeparttable}
\end{table}

Table~\ref{main:tab:power_main} compares the power of the omnibus and directional statistics at the median kernel scale across the four directions $S_1$--$S_4$ and the two loading regimes $c_1=0$ and $c_1=3$. When the deviation loads on the covariates, the omnibus statistic is uniformly more powerful: in Panel~A with $c_2=2$ and $n=2{,}000$, its rejection rates are $0.527$, $0.879$, $1.000$, and $0.888$ across $S_1$--$S_4$, compared with $0.073$, $0.055$, $0.288$, and $0.100$ for the directional statistic. In Panel~B, where $c_1=3$ adds a large propensity-score component deviation, the ranking reverses: at $(c_1,c_2)=(3,0.5)$ and $n=2{,}000$ the directional statistic rejects with probabilities between $0.764$ and $0.822$, against $0.178$ to $0.692$ for the omnibus statistic. The trade-off is asymmetric: the omnibus statistic retains non-trivial power where the directional statistic is favored, whereas against three of the four covariate-loaded alternatives the directional statistic is barely distinguishable from its own size.

Tables~\ref{main:tab:power_N_c0}--\ref{main:tab:power_di_N_c3} compare
rejection rates at the data-driven median scale with those at the two fixed
Gaussian scales, $\sigma=1$ and $\sigma=2$, as the covariate dispersion
$\sigma_X$ varies. A fixed scale is competitive mainly when it happens to lie
close to the median scale. At $\sigma_X=0.5$, for example, against $S_4$ with
$c_2=1$ and $n=2{,}000$, the rejection rates of
$(\TXP{\mathrm{med}},\TXP{1},\TXP{2})$ are $(0.714,0.537,0.261)$; at
$\sigma_X=2$, against $S_2$ with $c_2=2$, they are $(0.956,0.336,0.650)$. As
anticipated in Section~\ref{sec:sim-design}, the median scale of the omnibus
conditioning set is $0.73$ at $\sigma_X=0.5$ and $10.92$ at $\sigma_X=2$. At
$\sigma_X=0.5$, both fixed scales are too large, so the rejection rates follow
the distance to the median: the median rule outperforms $\sigma=1$, which in
turn outperforms $\sigma=2$. At $\sigma_X=2$, both are too small, and the
ranking reverses: the median rule outperforms $\sigma=2$, which outperforms
$\sigma=1$. At the baseline $\sigma_X=1$, the median scale is $2.56$; the
median rule and the fixed $\sigma=2$ then perform similarly. The median rule
also performs well for the directional statistic in the panels with $c_1=3$,
where that statistic has power: it gives the highest rejection rate in 26 of
the 36 cells (four directions, three values of $c_2$, three sample sizes) at
$\sigma_X=1$ and 28 of 36 at $\sigma_X=2$. At $\sigma_X=0.5$, it is highest in
only 12 of 36; where it is not, the shortfall from the best fixed scale is
usually small, the largest being about $0.18$ for $S_3$.

\begin{table}[!tbp]
\caption{Power of the omnibus and directional statistics at the median scale.}
\label{main:tab:power_main}
\centering
\singlespacing
\setlength{\tabcolsep}{3pt}
\footnotesize
\begin{threeparttable}
\begin{tabular*}{\textwidth}{@{\extracolsep{\fill}}ll*{8}{r}@{}}
\toprule
& & \multicolumn{2}{c}{$S_1$ } & \multicolumn{2}{c}{$S_2$ }
& \multicolumn{2}{c}{$S_3$ } & \multicolumn{2}{c}{$S_4$ }\\
\cmidrule(lr){3-4}\cmidrule(lr){5-6}\cmidrule(lr){7-8}\cmidrule(lr){9-10}
$c_2$ & $n$ & $T^{X\!P}$ & $T^{P}$ & $T^{X\!P}$ & $T^{P}$
& $T^{X\!P}$ & $T^{P}$ & $T^{X\!P}$ & $T^{P}$\\
\midrule
\multicolumn{10}{c}{\textit{Panel A}: $c_1=0$ }\\
\midrule
0.5 & 500 & 0.067 & 0.049 & 0.079 & 0.049 & 0.186 & 0.057 & 0.108 & 0.049\\
 & 1,000 & 0.060 & 0.061 & 0.162 & 0.050 & 0.419 & 0.069 & 0.162 & 0.076\\
 & 2,000 & 0.120 & 0.044 & 0.337 & 0.051 & 0.753 & 0.116 & 0.324 & 0.078\\
\addlinespace[2pt]
1 & 500 & 0.099 & 0.048 & 0.164 & 0.053 & 0.388 & 0.075 & 0.189 & 0.051\\
 & 1,000 & 0.107 & 0.062 & 0.343 & 0.064 & 0.769 & 0.108 & 0.330 & 0.069\\
 & 2,000 & 0.280 & 0.052 & 0.699 & 0.049 & 0.986 & 0.221 & 0.703 & 0.087\\
\addlinespace[2pt]
2 & 500 & 0.146 & 0.044 & 0.244 & 0.055 & 0.587 & 0.100 & 0.257 & 0.050\\
 & 1,000 & 0.225 & 0.074 & 0.507 & 0.065 & 0.920 & 0.143 & 0.502 & 0.070\\
 & 2,000 & 0.527 & 0.073 & 0.879 & 0.055 & 1.000 & 0.288 & 0.888 & 0.100\\
\midrule
\multicolumn{10}{c}{\textit{Panel B}: $c_1=3$ }\\
\midrule
0.5 & 500 & 0.087 & 0.309 & 0.082 & 0.285 & 0.186 & 0.275 & 0.089 & 0.266\\
 & 1,000 & 0.092 & 0.511 & 0.080 & 0.488 & 0.316 & 0.443 & 0.118 & 0.448\\
 & 2,000 & 0.178 & 0.822 & 0.184 & 0.808 & 0.692 & 0.774 & 0.211 & 0.764\\
\addlinespace[2pt]
1 & 500 & 0.105 & 0.218 & 0.101 & 0.211 & 0.676 & 0.185 & 0.143 & 0.176\\
 & 1,000 & 0.164 & 0.408 & 0.134 & 0.369 & 0.961 & 0.331 & 0.262 & 0.307\\
 & 2,000 & 0.304 & 0.709 & 0.312 & 0.697 & 1.000 & 0.613 & 0.554 & 0.569\\
\addlinespace[2pt]
2 & 500 & 0.144 & 0.155 & 0.182 & 0.114 & 0.918 & 0.090 & 0.237 & 0.100\\
 & 1,000 & 0.257 & 0.252 & 0.416 & 0.190 & 1.000 & 0.122 & 0.480 & 0.152\\
 & 2,000 & 0.510 & 0.489 & 0.779 & 0.391 & 1.000 & 0.235 & 0.881 & 0.291\\
\bottomrule
\end{tabular*}
\begin{tablenotes}[flushleft]\scriptsize
\item \emph{Note}: Rejection frequencies at the $5\%$ level over $1{,}000$
replications with $B = 500$ multiplier-bootstrap draws; $K = 5$.  Tables~\ref{main:tab:power_N_c0}--\ref{main:tab:power_di_N_c3}
report all scales and covariate dispersions.
\end{tablenotes}
\end{threeparttable}
\end{table}

Figure~\ref{fig:dimension} shows rejection rates as the covariate dimension
increases under Design~L. We use one illustrative alternative,
$S_1$ with $(c_1,c_2)=(0,1)$, $\sigma_X=1$, and $n=2{,}000$. The qualitative
pattern is similar under the other alternatives. The first five covariates
carry the signal, while additional covariates are nuisance variables. The
median scale rises from about $2.5$ at $K=5$ to $18$ at $K=30$, as its power declines from $0.391$ to $0.099$. By contrast, the power for the fixed ICM scale $\sigma=2$ falls from $0.425$
at $K=5$ to $0.060$ at $K=15$, and for $K\geq 20$ the rejection rate is below
the nominal $5\%$ level.

\begin{figure}[htbp]
\centering
\includegraphics[width=0.9\linewidth]{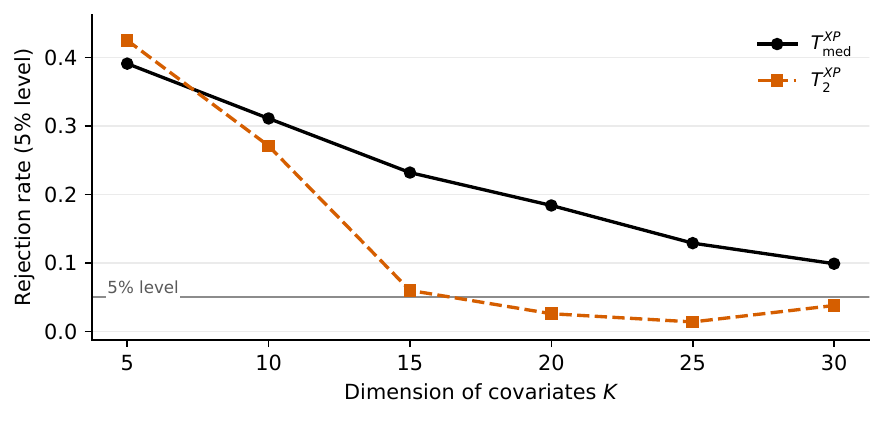}
\caption{Power as the covariate dimension increases.}
\label{fig:dimension}

\vspace{3pt}
\begin{minipage}{\linewidth}\scriptsize \emph{Note}: 
     Empirical rejection rates
at the $5\%$ level under Design~L, direction $S_1$, $(c_1,c_2) = (0,1)$,
$\sigma_X = 1$, $n = 2{,}000$ ($1{,}000$ replications, $B = 500$). The
realized median scale is $\mathrm{med}^2 = 2.5,\, 5.5,\, 8.6,\, 11.7,\,
14.9,\, 18.0$ for $K = 5,\ldots,30$, against the fixed $\sigma = 2$.
\end{minipage}
\end{figure}

\section{Empirical Application}\label{sec:application}

We revisit the Empowerment and Livelihood for Adolescents (ELA) programme of \citet{Bandiera2020WomensEmpowerment} (hereafter BBGGRS). ELA is run by BRAC, a Bangladeshi non-governmental organization, and provides adolescent girls with vocational and life-skills training in a space of their own. Of the $150$ study communities in Uganda, $100$ were randomly assigned a club and $50$ served as controls. Attendance was voluntary: by the midline, 17 percent of girls in treatment communities and 4 percent in control communities had participated, and by the
endline the corresponding figures were 22 and 8 percent. Because club placement was randomized but participation was chosen, girls may sort into treatment on unobserved gains.

Following BBGGRS, the excluded instruments are the distance to the nearest club, its square, the treatment dummy, and the two interactions of distance with treatment. The covariates $X$ collect the baseline value of the outcome
index, nine branch dummies, and age. The estimation samples contain $3{,}984$
and $2{,}942$ observations in $150$ and $148$ villages at midline and endline,
respectively. BBGGRS estimate MTE curves for economic empowerment, control over
the body, and aspirations at the two follow-ups under a joint normal selection
model, and the estimated curves vary in shape across the six outcome--wave
cells: one appears upward-sloping, four downward-sloping, and one nearly flat.
This variation motivates a formal test of essential heterogeneity. Figure~\ref{fig:ps-support} provides context for this qualification: estimated
propensity scores are concentrated at low values, with 96\% of midline
observations and 90\% of endline observations below $0.30$; the first-stage fit and the implied common-support limits are
reported in Table~\ref{tab:application-diag} of the Supplement~\ref{supp:application}. Thus, MTE behavior
beyond the region with substantial empirical support depends heavily on the
parametric specification. Where our test does not reject, the IV estimand may
offer a more defensible summary than the extrapolated MTE curve.

\begin{figure}[htbp]
\centering
\includegraphics[width=0.9\linewidth]{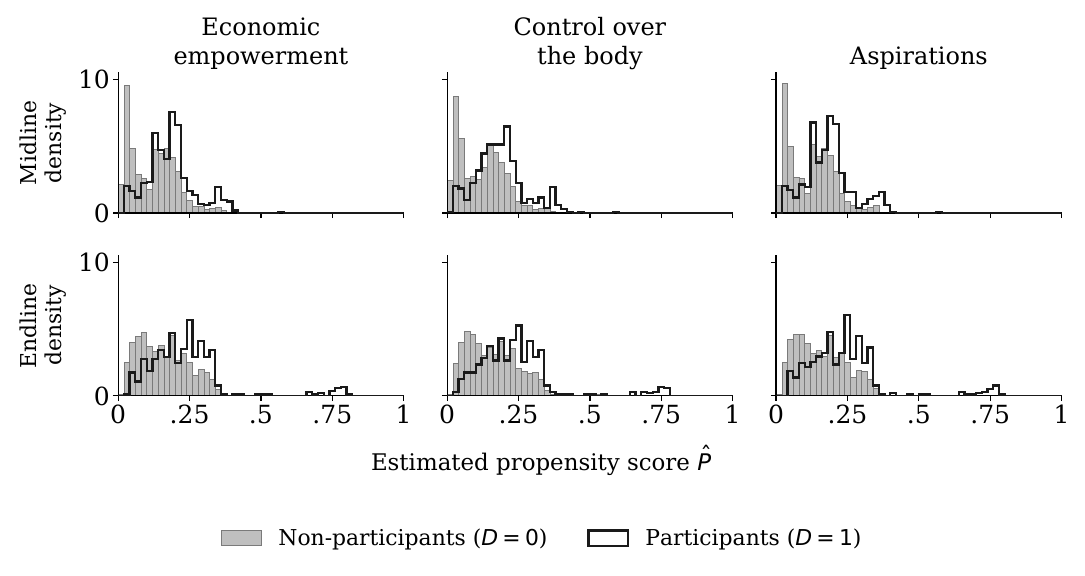}
\caption{Estimated propensity scores in the six ELA samples.}
\label{fig:ps-support}

\vspace{3pt}
\begin{minipage}{\linewidth}\scriptsize \emph{Note}: Estimated first-stage
propensity scores by participation status, on the estimation samples. Rows are
follow-up waves; columns are outcome indices.\end{minipage}
\end{figure}

Table~\ref{main:tab:application} reports bootstrap $p$-values for eight statistics
in each cell: the omnibus $\TXP{\sigma}$ and directional $\TP{\sigma}$ at
$\sigma \in \{\mathrm{med},1,2\}$, together with the polynomial Wald
benchmarks $\Qpoly{2}$ and $\Qpoly{5}$. Inference accounts for village-level
randomization through a cluster bootstrap with $B=999$ draws. The final rows
report the median kernel scales selected by $\TXP{\mathrm{med}}$ and $\TP{\mathrm{med}}$, along with the numbers of observations and villages.

\begin{table}[htbp]
\caption{Tests for essential heterogeneity in the ELA data.}
\label{main:tab:application}
\centering
\singlespacing
\setlength{\tabcolsep}{6pt}
\footnotesize
\begin{threeparttable}
\begin{tabular*}{\textwidth}{@{\extracolsep{\fill}}@{}l rr rr rr@{}}
\toprule
& \multicolumn{2}{c}{Economic} & \multicolumn{2}{c}{Control over}
& \multicolumn{2}{c}{Aspirations}\\
& \multicolumn{2}{c}{empowerment} & \multicolumn{2}{c}{the body} & \multicolumn{2}{c}{}\\
\cmidrule(lr){2-3}\cmidrule(lr){4-5}\cmidrule(lr){6-7}
& Mid & End & Mid & End & Mid & End\\
\midrule
$\TXP{\mathrm{med}}$ & 0.383 & 0.181 & 0.225 & 0.754 & \textbf{0.014} & 0.499\\
$\TXP{1}$ & 0.570 & 0.355 & 0.160 & 0.752 & \textbf{0.023} & 0.784\\
$\TXP{2}$ & 0.376 & 0.203 & 0.226 & 0.764 & \textbf{0.014} & 0.564\\
$\TP{\mathrm{med}}$ & 0.770 & 0.264 & 0.291 & 0.629 & 0.283 & 0.593\\
$\TP{1}$ & 0.053 & 0.246 & 0.415 & 0.962 & 0.791 & 0.523\\
$\TP{2}$ & 0.053 & 0.338 & 0.453 & 0.973 & 0.783 & 0.606\\
$\Qpoly{2}$ & 0.390 & 0.588 & 0.066 & 0.101 & 0.063 & 0.221\\
$\Qpoly{5}$ & 0.204 & 0.351 & 0.322 & 0.612 & 0.373 & 0.623\\
\midrule
\multicolumn{7}{@{}l}{\textit{Robustness: $\mu_d(X)$ of polynomial degree $J$}}\\
$\TXP{\mathrm{med}}$, $J=2$ & 0.556 & 0.154 & 0.389 & 0.841 & \textbf{0.024} & 0.544\\
$\TXP{\mathrm{med}}$, $J=3$ & 0.514 & 0.316 & 0.245 & 0.746 & \textbf{0.021} & 0.618\\
$\TP{\mathrm{med}}$, $J=2$ & 0.674 & 0.215 & 0.680 & 0.553 & 0.246 & 0.578\\
$\TP{\mathrm{med}}$, $J=3$ & 0.659 & 0.149 & 0.616 & 0.512 & 0.244 & 0.598\\
\midrule
Median $\hat\sigma$, $W=(X^\top,P)^\top$ & 1.93 & 2.38 & 1.94 & 2.40 & 1.93 & 2.42\\
Median $\hat\sigma$, $W=P$ & 0.0069 & 0.0081 & 0.0069 & 0.0081 & 0.0069 & 0.0083\\
Observations & 3,984 & 2,943 & 3,985 & 2,942 & 3,985 & 2,942\\
Clusters (villages) & 150 & 148 & 150 & 148 & 150 & 148\\
\bottomrule
\end{tabular*}

\begin{tablenotes}[flushleft]\scriptsize
\item \emph{Note}: Bootstrap $p$-values from $B=999$ two-point multipliers
\citep{mammen1993bootstrap}, one per village. Bold marks rejection at the
$5\%$ level. The robustness block uses degree-$J$ polynomials in the baseline
index and age; $J=1$ is the baseline specification.
\end{tablenotes}
\end{threeparttable}
\end{table}

The clearest result is for aspirations at midline. There, the omnibus statistic
rejects at all three scales ($p=0.014$, $0.023$, $0.014$), whereas the
directional statistic does not reject. These are the
only rejections under the baseline specification. A rejection that conditions on the covariates could also reflect
misspecification of the outcome means $\mu_d(X;\beta_d)$
(Remark~\ref{rem:joint}), but the rejection survives re-estimating them as quadratic and cubic
polynomials: $\TXP{\mathrm{med}}$
returns $p=0.024$ and $0.021$, against $0.014$ under the baseline
specification, while $\TP{\mathrm{med}}$ stays at $p=0.246$ and $0.244$, so the
departure remains visible only to the statistic that conditions on the
covariates. Across the six
cells, neither median-scale statistic changes its decision at the $5\%$
level under either polynomial specification; Table~\ref{stab:robust-J} of
the Supplement reports all cells and statistics for $J = 2$ and $J = 3$.

For the omnibus statistic, whose conditioning set includes the covariates, the median rule selects $\hat\sigma$ between $1.93$ and $2.42$, and the three omnibus columns never disagree at the $5\%$ level. For the directional statistic, which conditions on $P$ alone, the estimated scores lie within a short interval, and the median rule selects $\hat\sigma$ only in the range $0.007$--$0.008$, so the fixed scales $\sigma=1$ and $\sigma=2$ are two orders of magnitude too large. This is the scale mismatch of Section~\ref{sec:sim-results} in an extreme form, which is why we read only $\TP{\mathrm{med}}$ among the directional statistics.

For five of the six outcome--wave cells, neither the omnibus nor the directional test rejects the null of no essential heterogeneity at the $5\%$ level. The exception is aspirations at midline, where the omnibus test rejects at all three scales. Thus, given the limited propensity-score support, we regard the IV estimand as a more credible primary summary than the ATE, ATT, and ATU derived from the extrapolated MTE curve when the test does not reject.

\section{Conclusion}\label{sec:conclusion}
This paper develops a test for essential heterogeneity in the marginal treatment effects model that does not rely on the additive separability restriction maintained by existing procedures. We show that, without additive separability, the null of no essential heterogeneity implies a single conditional moment restriction on the covariates and the propensity score, and that this restriction is sharp in the sense that it is the strongest implication that the data can test. We propose an asymptotically valid multiplier bootstrap procedure based on a projected Gaussian process statistic, whose Neyman-orthogonal moment accommodates the estimated propensity score without requiring nuisance parameter re-estimation in each bootstrap replication.

Monte Carlo simulations demonstrate that the procedure delivers accurate size control and good power against a broad class of alternatives. The omnibus statistic detects covariate-loaded departures to which the directional test is blind by construction, and the data-driven median heuristic keeps the test informative as the conditioning set grows in
dimension, where fixed-scale implementations collapse. In our empirical application to the ELA programme of \citet{Bandiera2020WomensEmpowerment}, the omnibus test rejects the null of no essential heterogeneity for aspirations at the midline follow-up, while the directional test does not. To our knowledge, this is the first essential-heterogeneity test that dispenses with additive separability and conditions on the covariates together with the propensity score. 

Several directions for future work emerge naturally. First, our procedure maintains parametric specifications for the conditional outcome means $\mu_d(X;\beta_d)$, so the tested null is joint with the outcome parametrization. A natural extension is to test the sharp affine condition
\eqref{H_0: affine} directly with nonparametrically estimated $a(\cdot)$ and $b(\cdot)$. Second, the choice of the kernel scale parameter merits further investigation: the median heuristic performs well across our designs, but data-adaptive selection that goes beyond the median is a promising route to more robust power.

\begingroup
\spacingset{1}
\bibliography{references}
\endgroup

\newpage
\setcounter{table}{0}
\renewcommand{\thetable}{A.\arabic{table}}

\newpage

\begin{table}[!tbp]
\caption{Power of $T^{X\!P}_\sigma$ across covariate scales, probit outcome, $c_1=0$ (covariate-dependent heterogeneity only).}
\label{main:tab:power_N_c0}
\centering
\singlespacing
\setlength{\tabcolsep}{3pt}
\footnotesize
\begin{threeparttable}
\begin{tabular*}{\textwidth}{@{\extracolsep{\fill}}ll*{9}{r}@{}}
\toprule
& & \multicolumn{3}{c}{$\sigma_X=0.5$} & \multicolumn{3}{c}{$\sigma_X=1$} & \multicolumn{3}{c}{$\sigma_X=2$}\\
\cmidrule(lr){3-5}\cmidrule(lr){6-8}\cmidrule(lr){9-11}
$c_2$ & $n$ & $\mathrm{med}$ & $1$ & $2$ & $\mathrm{med}$ & $1$ & $2$ & $\mathrm{med}$ & $1$ & $2$\\
\midrule
\multicolumn{11}{c}{\makebox[0pt][c]{\textit{Panel A}: $S_1\propto X$ }}\\
\midrule
0.5 & 500 & 0.070 & 0.065 & 0.062 & 0.067 & 0.067 & 0.067 & 0.067 & 0.060 & 0.070\\
 & 1,000 & 0.077 & 0.068 & 0.054 & 0.060 & 0.076 & 0.059 & 0.079 & 0.058 & 0.072\\
 & 2,000 & 0.131 & 0.120 & 0.094 & 0.120 & 0.113 & 0.131 & 0.106 & 0.093 & 0.100\\
\addlinespace[2pt]
1 & 500 & 0.087 & 0.080 & 0.068 & 0.099 & 0.104 & 0.098 & 0.085 & 0.073 & 0.085\\
 & 1,000 & 0.127 & 0.103 & 0.072 & 0.107 & 0.149 & 0.123 & 0.138 & 0.088 & 0.118\\
 & 2,000 & 0.261 & 0.216 & 0.141 & 0.280 & 0.354 & 0.311 & 0.276 & 0.136 & 0.225\\
\addlinespace[2pt]
2 & 500 & 0.113 & 0.098 & 0.080 & 0.146 & 0.168 & 0.151 & 0.141 & 0.089 & 0.133\\
 & 1,000 & 0.204 & 0.159 & 0.097 & 0.225 & 0.319 & 0.256 & 0.240 & 0.136 & 0.237\\
 & 2,000 & 0.457 & 0.328 & 0.177 & 0.527 & 0.686 & 0.594 & 0.579 & 0.315 & 0.578\\
\midrule
\multicolumn{11}{c}{\makebox[0pt][c]{\textit{Panel B}: $S_2\propto e^{X/2}-Ee^{X/2}$ }}\\
\addlinespace[1pt]
0.5 & 500 & 0.078 & 0.074 & 0.064 & 0.079 & 0.078 & 0.087 & 0.097 & 0.071 & 0.081\\
 & 1,000 & 0.117 & 0.097 & 0.090 & 0.162 & 0.149 & 0.160 & 0.140 & 0.076 & 0.096\\
 & 2,000 & 0.256 & 0.238 & 0.214 & 0.337 & 0.292 & 0.342 & 0.211 & 0.113 & 0.160\\
\addlinespace[2pt]
1 & 500 & 0.117 & 0.109 & 0.088 & 0.164 & 0.168 & 0.178 & 0.190 & 0.078 & 0.103\\
 & 1,000 & 0.209 & 0.186 & 0.141 & 0.343 & 0.315 & 0.356 & 0.331 & 0.113 & 0.174\\
 & 2,000 & 0.476 & 0.424 & 0.322 & 0.699 & 0.646 & 0.696 & 0.621 & 0.186 & 0.330\\
\addlinespace[2pt]
2 & 500 & 0.137 & 0.126 & 0.086 & 0.244 & 0.251 & 0.262 & 0.355 & 0.094 & 0.161\\
 & 1,000 & 0.302 & 0.237 & 0.162 & 0.507 & 0.535 & 0.537 & 0.654 & 0.167 & 0.314\\
 & 2,000 & 0.639 & 0.553 & 0.351 & 0.879 & 0.902 & 0.897 & 0.956 & 0.336 & 0.650\\
\midrule
\multicolumn{11}{c}{\makebox[0pt][c]{\textit{Panel C}: $S_3\propto X^2-\mathrm{Var}(X)$ }}\\
\addlinespace[1pt]
0.5 & 500 & 0.289 & 0.289 & 0.211 & 0.186 & 0.181 & 0.188 & 0.109 & 0.082 & 0.108\\
 & 1,000 & 0.706 & 0.687 & 0.558 & 0.419 & 0.380 & 0.422 & 0.158 & 0.102 & 0.133\\
 & 2,000 & 0.967 & 0.965 & 0.930 & 0.753 & 0.708 & 0.760 & 0.252 & 0.166 & 0.240\\
\addlinespace[2pt]
1 & 500 & 0.613 & 0.608 & 0.497 & 0.388 & 0.350 & 0.392 & 0.177 & 0.094 & 0.154\\
 & 1,000 & 0.950 & 0.946 & 0.880 & 0.769 & 0.694 & 0.764 & 0.336 & 0.155 & 0.231\\
 & 2,000 & 1.000 & 1.000 & 1.000 & 0.986 & 0.963 & 0.983 & 0.610 & 0.279 & 0.437\\
\addlinespace[2pt]
2 & 500 & 0.802 & 0.788 & 0.680 & 0.587 & 0.511 & 0.581 & 0.293 & 0.116 & 0.187\\
 & 1,000 & 0.991 & 0.990 & 0.979 & 0.920 & 0.874 & 0.915 & 0.546 & 0.190 & 0.336\\
 & 2,000 & 1.000 & 1.000 & 1.000 & 1.000 & 0.998 & 1.000 & 0.857 & 0.403 & 0.629\\
\midrule
\multicolumn{11}{c}{\makebox[0pt][c]{\textit{Panel D}: $S_4\propto\sin(2X/\sigma_X)$ }}\\
\addlinespace[1pt]
0.5 & 500 & 0.119 & 0.100 & 0.099 & 0.108 & 0.134 & 0.113 & 0.087 & 0.077 & 0.078\\
 & 1,000 & 0.140 & 0.115 & 0.075 & 0.162 & 0.215 & 0.184 & 0.105 & 0.098 & 0.117\\
 & 2,000 & 0.359 & 0.243 & 0.132 & 0.324 & 0.489 & 0.378 & 0.188 & 0.185 & 0.238\\
\addlinespace[2pt]
1 & 500 & 0.151 & 0.137 & 0.108 & 0.189 & 0.248 & 0.209 & 0.129 & 0.107 & 0.141\\
 & 1,000 & 0.321 & 0.227 & 0.121 & 0.330 & 0.497 & 0.390 & 0.198 & 0.154 & 0.244\\
 & 2,000 & 0.714 & 0.537 & 0.261 & 0.703 & 0.892 & 0.786 & 0.491 & 0.409 & 0.606\\
\addlinespace[2pt]
2 & 500 & 0.210 & 0.171 & 0.125 & 0.257 & 0.366 & 0.293 & 0.169 & 0.112 & 0.207\\
 & 1,000 & 0.461 & 0.337 & 0.173 & 0.502 & 0.708 & 0.564 & 0.338 & 0.243 & 0.423\\
 & 2,000 & 0.867 & 0.721 & 0.353 & 0.888 & 0.980 & 0.935 & 0.723 & 0.657 & 0.866\\
\bottomrule
\end{tabular*}
\begin{tablenotes}[flushleft]\scriptsize
\item \emph{Note}: Rejection frequencies at the $5\%$ level over $1{,}000$
replications with $B = 500$ multiplier-bootstrap draws; $K = 5$. The
heterogeneity term is $\eta=(c_1+c_2S_j(X))U$ with $\mathrm{Var}(S_j)=5$.
Columns $\mathrm{med}$, $1$, and $2$ denote $\sigma$; their realised median
scales are $0.73$, $2.56$, and $10.92$. 
\end{tablenotes}
\end{threeparttable}
\end{table}

\begin{table}[!tbp]
\caption{Power of $T^{X\!P}_\sigma$ across covariate scales, probit outcome, $c_1=3$ (constant plus covariate-dependent heterogeneity).}
\label{main:tab:power_N_c3}
\centering
\singlespacing
\setlength{\tabcolsep}{3pt}
\footnotesize
\begin{threeparttable}
\begin{tabular*}{\textwidth}{@{\extracolsep{\fill}}ll*{9}{r}@{}}
\toprule
& & \multicolumn{3}{c}{$\sigma_X=0.5$} & \multicolumn{3}{c}{$\sigma_X=1$} & \multicolumn{3}{c}{$\sigma_X=2$}\\
\cmidrule(lr){3-5}\cmidrule(lr){6-8}\cmidrule(lr){9-11}
$c_2$ & $n$ & $\mathrm{med}$ & $1$ & $2$ & $\mathrm{med}$ & $1$ & $2$ & $\mathrm{med}$ & $1$ & $2$\\
\midrule
\multicolumn{11}{c}{\makebox[0pt][c]{\textit{Panel A}: $S_1\propto X$ }}\\
\midrule
0.5 & 500 & 0.174 & 0.165 & 0.144 & 0.087 & 0.092 & 0.082 & 0.093 & 0.076 & 0.069\\
 & 1,000 & 0.314 & 0.275 & 0.223 & 0.092 & 0.155 & 0.097 & 0.120 & 0.097 & 0.080\\
 & 2,000 & 0.590 & 0.509 & 0.393 & 0.178 & 0.295 & 0.191 & 0.230 & 0.193 & 0.129\\
\addlinespace[2pt]
1 & 500 & 0.214 & 0.200 & 0.169 & 0.105 & 0.103 & 0.098 & 0.102 & 0.072 & 0.081\\
 & 1,000 & 0.382 & 0.357 & 0.298 & 0.164 & 0.177 & 0.165 & 0.146 & 0.098 & 0.087\\
 & 2,000 & 0.702 & 0.643 & 0.523 & 0.304 & 0.363 & 0.310 & 0.264 & 0.201 & 0.198\\
\addlinespace[2pt]
2 & 500 & 0.194 & 0.165 & 0.130 & 0.144 & 0.143 & 0.145 & 0.133 & 0.093 & 0.115\\
 & 1,000 & 0.359 & 0.323 & 0.233 & 0.257 & 0.270 & 0.272 & 0.208 & 0.129 & 0.167\\
 & 2,000 & 0.705 & 0.599 & 0.400 & 0.510 & 0.570 & 0.534 & 0.485 & 0.261 & 0.411\\
\midrule
\multicolumn{11}{c}{\makebox[0pt][c]{\textit{Panel B}: $S_2\propto e^{X/2}-Ee^{X/2}$ }}\\
\addlinespace[1pt]
0.5 & 500 & 0.160 & 0.149 & 0.116 & 0.082 & 0.114 & 0.086 & 0.077 & 0.096 & 0.080\\
 & 1,000 & 0.234 & 0.192 & 0.152 & 0.080 & 0.177 & 0.094 & 0.088 & 0.119 & 0.092\\
 & 2,000 & 0.506 & 0.403 & 0.288 & 0.184 & 0.350 & 0.202 & 0.143 & 0.257 & 0.183\\
\addlinespace[2pt]
1 & 500 & 0.158 & 0.137 & 0.126 & 0.101 & 0.131 & 0.102 & 0.104 & 0.113 & 0.100\\
 & 1,000 & 0.260 & 0.219 & 0.177 & 0.134 & 0.208 & 0.153 & 0.179 & 0.161 & 0.166\\
 & 2,000 & 0.555 & 0.489 & 0.366 & 0.312 & 0.459 & 0.346 & 0.419 & 0.356 & 0.384\\
\addlinespace[2pt]
2 & 500 & 0.148 & 0.129 & 0.105 & 0.182 & 0.237 & 0.210 & 0.292 & 0.163 & 0.231\\
 & 1,000 & 0.292 & 0.241 & 0.170 & 0.416 & 0.470 & 0.429 & 0.568 & 0.298 & 0.445\\
 & 2,000 & 0.639 & 0.519 & 0.312 & 0.779 & 0.850 & 0.807 & 0.945 & 0.643 & 0.849\\
\midrule
\multicolumn{11}{c}{\makebox[0pt][c]{\textit{Panel C}: $S_3\propto X^2-\mathrm{Var}(X)$ }}\\
\addlinespace[1pt]
0.5 & 500 & 0.163 & 0.140 & 0.107 & 0.186 & 0.223 & 0.197 & 0.106 & 0.119 & 0.106\\
 & 1,000 & 0.315 & 0.257 & 0.167 & 0.316 & 0.414 & 0.349 & 0.167 & 0.170 & 0.147\\
 & 2,000 & 0.711 & 0.617 & 0.433 & 0.692 & 0.795 & 0.714 & 0.375 & 0.366 & 0.371\\
\addlinespace[2pt]
1 & 500 & 0.669 & 0.620 & 0.499 & 0.676 & 0.690 & 0.692 & 0.448 & 0.227 & 0.297\\
 & 1,000 & 0.957 & 0.936 & 0.852 & 0.961 & 0.970 & 0.966 & 0.798 & 0.464 & 0.626\\
 & 2,000 & 1.000 & 1.000 & 0.994 & 1.000 & 1.000 & 1.000 & 0.988 & 0.798 & 0.946\\
\addlinespace[2pt]
2 & 500 & 0.952 & 0.937 & 0.854 & 0.918 & 0.907 & 0.926 & 0.731 & 0.291 & 0.450\\
 & 1,000 & 1.000 & 1.000 & 0.997 & 1.000 & 0.999 & 1.000 & 0.972 & 0.568 & 0.796\\
 & 2,000 & 1.000 & 1.000 & 1.000 & 1.000 & 1.000 & 1.000 & 1.000 & 0.930 & 0.997\\
\midrule
\multicolumn{11}{c}{\makebox[0pt][c]{\textit{Panel D}: $S_4\propto\sin(2X/\sigma_X)$ }}\\
\addlinespace[1pt]
0.5 & 500 & 0.195 & 0.176 & 0.141 & 0.089 & 0.101 & 0.089 & 0.098 & 0.079 & 0.065\\
 & 1,000 & 0.330 & 0.283 & 0.217 & 0.118 & 0.160 & 0.126 & 0.120 & 0.094 & 0.087\\
 & 2,000 & 0.648 & 0.556 & 0.442 & 0.211 & 0.344 & 0.235 & 0.235 & 0.190 & 0.152\\
\addlinespace[2pt]
1 & 500 & 0.242 & 0.218 & 0.179 & 0.143 & 0.185 & 0.148 & 0.112 & 0.096 & 0.098\\
 & 1,000 & 0.464 & 0.389 & 0.290 & 0.262 & 0.350 & 0.288 & 0.157 & 0.153 & 0.182\\
 & 2,000 & 0.840 & 0.769 & 0.580 & 0.554 & 0.741 & 0.614 & 0.360 & 0.337 & 0.413\\
\addlinespace[2pt]
2 & 500 & 0.258 & 0.228 & 0.175 & 0.237 & 0.313 & 0.264 & 0.166 & 0.120 & 0.187\\
 & 1,000 & 0.570 & 0.457 & 0.288 & 0.480 & 0.659 & 0.526 & 0.304 & 0.252 & 0.380\\
 & 2,000 & 0.941 & 0.839 & 0.613 & 0.881 & 0.975 & 0.919 & 0.682 & 0.618 & 0.808\\
\bottomrule
\end{tabular*}
\begin{tablenotes}[flushleft]\scriptsize
\item \emph{Note}: Design, notation, and scale values are as in
Table~\ref{main:tab:power_N_c0}. 
\end{tablenotes}
\end{threeparttable}
\end{table}

\begin{table}[!tbp]
\caption{Power of $T^{P}_\sigma$ across covariate scales, probit outcome, $c_1=0$ (covariate-dependent heterogeneity only).}
\label{main:tab:power_di_N_c0}
\centering
\singlespacing
\setlength{\tabcolsep}{3pt}
\footnotesize
\begin{threeparttable}
\begin{tabular*}{\textwidth}{@{\extracolsep{\fill}}ll*{9}{r}@{}}
\toprule
& & \multicolumn{3}{c}{$\sigma_X=0.5$} & \multicolumn{3}{c}{$\sigma_X=1$} & \multicolumn{3}{c}{$\sigma_X=2$}\\
\cmidrule(lr){3-5}\cmidrule(lr){6-8}\cmidrule(lr){9-11}
$c_2$ & $n$ & $\mathrm{med}$ & $1$ & $2$ & $\mathrm{med}$ & $1$ & $2$ & $\mathrm{med}$ & $1$ & $2$\\
\midrule
\multicolumn{11}{c}{\makebox[0pt][c]{\textit{Panel A}: $S_1\propto X$ }}\\
\midrule
0.5 & 500 & 0.050 & 0.051 & 0.058 & 0.049 & 0.060 & 0.054 & 0.048 & 0.044 & 0.042\\
 & 1,000 & 0.061 & 0.046 & 0.053 & 0.061 & 0.052 & 0.051 & 0.071 & 0.055 & 0.058\\
 & 2,000 & 0.058 & 0.053 & 0.055 & 0.044 & 0.050 & 0.054 & 0.068 & 0.063 & 0.060\\
\addlinespace[2pt]
1 & 500 & 0.051 & 0.057 & 0.057 & 0.048 & 0.068 & 0.073 & 0.052 & 0.054 & 0.052\\
 & 1,000 & 0.055 & 0.061 & 0.066 & 0.062 & 0.056 & 0.053 & 0.071 & 0.061 & 0.055\\
 & 2,000 & 0.050 & 0.055 & 0.080 & 0.052 & 0.053 & 0.059 & 0.068 & 0.050 & 0.048\\
\addlinespace[2pt]
2 & 500 & 0.055 & 0.071 & 0.068 & 0.044 & 0.067 & 0.057 & 0.060 & 0.048 & 0.048\\
 & 1,000 & 0.075 & 0.070 & 0.084 & 0.074 & 0.054 & 0.045 & 0.079 & 0.056 & 0.060\\
 & 2,000 & 0.052 & 0.095 & 0.096 & 0.073 & 0.057 & 0.053 & 0.091 & 0.054 & 0.047\\
\midrule
\multicolumn{11}{c}{\makebox[0pt][c]{\textit{Panel B}: $S_2\propto e^{X/2}-Ee^{X/2}$ }}\\
\addlinespace[1pt]
0.5 & 500 & 0.051 & 0.057 & 0.051 & 0.049 & 0.072 & 0.067 & 0.067 & 0.046 & 0.044\\
 & 1,000 & 0.062 & 0.047 & 0.057 & 0.050 & 0.060 & 0.061 & 0.059 & 0.064 & 0.062\\
 & 2,000 & 0.053 & 0.058 & 0.080 & 0.051 & 0.075 & 0.077 & 0.071 & 0.068 & 0.062\\
\addlinespace[2pt]
1 & 500 & 0.048 & 0.065 & 0.061 & 0.053 & 0.062 & 0.071 & 0.060 & 0.065 & 0.056\\
 & 1,000 & 0.056 & 0.068 & 0.085 & 0.064 & 0.078 & 0.083 & 0.086 & 0.089 & 0.079\\
 & 2,000 & 0.049 & 0.081 & 0.113 & 0.049 & 0.112 & 0.127 & 0.112 & 0.147 & 0.131\\
\addlinespace[2pt]
2 & 500 & 0.050 & 0.069 & 0.078 & 0.055 & 0.063 & 0.069 & 0.062 & 0.062 & 0.056\\
 & 1,000 & 0.072 & 0.082 & 0.109 & 0.065 & 0.079 & 0.081 & 0.099 & 0.096 & 0.087\\
 & 2,000 & 0.058 & 0.121 & 0.150 & 0.055 & 0.111 & 0.120 & 0.117 & 0.119 & 0.095\\
\midrule
\multicolumn{11}{c}{\makebox[0pt][c]{\textit{Panel C}: $S_3\propto X^2-\mathrm{Var}(X)$ }}\\
\addlinespace[1pt]
0.5 & 500 & 0.055 & 0.062 & 0.068 & 0.057 & 0.078 & 0.080 & 0.054 & 0.049 & 0.047\\
 & 1,000 & 0.079 & 0.120 & 0.147 & 0.069 & 0.094 & 0.104 & 0.076 & 0.061 & 0.066\\
 & 2,000 & 0.083 & 0.255 & 0.336 & 0.116 & 0.207 & 0.213 & 0.098 & 0.079 & 0.079\\
\addlinespace[2pt]
1 & 500 & 0.064 & 0.072 & 0.089 & 0.075 & 0.077 & 0.083 & 0.092 & 0.062 & 0.064\\
 & 1,000 & 0.110 & 0.146 & 0.183 & 0.108 & 0.122 & 0.136 & 0.141 & 0.104 & 0.100\\
 & 2,000 & 0.144 & 0.255 & 0.355 & 0.221 & 0.257 & 0.257 & 0.274 & 0.200 & 0.147\\
\addlinespace[2pt]
2 & 500 & 0.093 & 0.082 & 0.100 & 0.100 & 0.073 & 0.074 & 0.115 & 0.078 & 0.057\\
 & 1,000 & 0.137 & 0.128 & 0.167 & 0.143 & 0.099 & 0.087 & 0.227 & 0.103 & 0.077\\
 & 2,000 & 0.231 & 0.182 & 0.208 & 0.288 & 0.120 & 0.093 & 0.409 & 0.150 & 0.090\\
\midrule
\multicolumn{11}{c}{\makebox[0pt][c]{\textit{Panel D}: $S_4\propto\sin(2X/\sigma_X)$ }}\\
\addlinespace[1pt]
0.5 & 500 & 0.048 & 0.066 & 0.058 & 0.049 & 0.072 & 0.070 & 0.055 & 0.060 & 0.065\\
 & 1,000 & 0.057 & 0.055 & 0.060 & 0.076 & 0.060 & 0.065 & 0.080 & 0.067 & 0.069\\
 & 2,000 & 0.061 & 0.054 & 0.057 & 0.078 & 0.074 & 0.075 & 0.074 & 0.092 & 0.093\\
\addlinespace[2pt]
1 & 500 & 0.044 & 0.050 & 0.049 & 0.051 & 0.075 & 0.081 & 0.051 & 0.064 & 0.071\\
 & 1,000 & 0.062 & 0.061 & 0.062 & 0.069 & 0.081 & 0.083 & 0.082 & 0.078 & 0.078\\
 & 2,000 & 0.057 & 0.059 & 0.061 & 0.087 & 0.102 & 0.093 & 0.085 & 0.085 & 0.100\\
\addlinespace[2pt]
2 & 500 & 0.048 & 0.050 & 0.051 & 0.050 & 0.082 & 0.089 & 0.058 & 0.063 & 0.062\\
 & 1,000 & 0.067 & 0.068 & 0.068 & 0.070 & 0.089 & 0.104 & 0.071 & 0.065 & 0.065\\
 & 2,000 & 0.052 & 0.058 & 0.061 & 0.100 & 0.125 & 0.128 & 0.097 & 0.084 & 0.093\\
\bottomrule
\end{tabular*}
\begin{tablenotes}[flushleft]\scriptsize
\item \emph{Note}: Design, notation, and scale values are as in
Table~\ref{main:tab:power_N_c0}. 
\end{tablenotes}
\end{threeparttable}
\end{table}

\begin{table}[!tbp]
\caption{Power of $T^{P}_\sigma$ across covariate scales, probit outcome, $c_1=3$ (constant plus covariate-dependent heterogeneity).}
\label{main:tab:power_di_N_c3}
\centering
\singlespacing
\setlength{\tabcolsep}{3pt}
\footnotesize
\begin{threeparttable}
\begin{tabular*}{\textwidth}{@{\extracolsep{\fill}}ll*{9}{r}@{}}
\toprule
& & \multicolumn{3}{c}{$\sigma_X=0.5$} & \multicolumn{3}{c}{$\sigma_X=1$} & \multicolumn{3}{c}{$\sigma_X=2$}\\
\cmidrule(lr){3-5}\cmidrule(lr){6-8}\cmidrule(lr){9-11}
$c_2$ & $n$ & $\mathrm{med}$ & $1$ & $2$ & $\mathrm{med}$ & $1$ & $2$ & $\mathrm{med}$ & $1$ & $2$\\
\midrule
\multicolumn{11}{c}{\makebox[0pt][c]{\textit{Panel A}: $S_1\propto X$ }}\\
\midrule
0.5 & 500 & 0.380 & 0.377 & 0.223 & 0.309 & 0.211 & 0.113 & 0.213 & 0.157 & 0.130\\
 & 1,000 & 0.619 & 0.677 & 0.475 & 0.511 & 0.398 & 0.200 & 0.379 & 0.272 & 0.217\\
 & 2,000 & 0.904 & 0.931 & 0.812 & 0.822 & 0.739 & 0.350 & 0.707 & 0.536 & 0.428\\
\addlinespace[2pt]
1 & 500 & 0.293 & 0.265 & 0.166 & 0.218 & 0.173 & 0.094 & 0.194 & 0.087 & 0.075\\
 & 1,000 & 0.456 & 0.465 & 0.294 & 0.408 & 0.255 & 0.107 & 0.349 & 0.146 & 0.118\\
 & 2,000 & 0.786 & 0.774 & 0.603 & 0.709 & 0.558 & 0.153 & 0.654 & 0.269 & 0.180\\
\addlinespace[2pt]
2 & 500 & 0.167 & 0.125 & 0.086 & 0.155 & 0.094 & 0.074 & 0.140 & 0.069 & 0.061\\
 & 1,000 & 0.284 & 0.253 & 0.151 & 0.252 & 0.127 & 0.072 & 0.262 & 0.075 & 0.064\\
 & 2,000 & 0.528 & 0.461 & 0.308 & 0.489 & 0.236 & 0.111 & 0.497 & 0.092 & 0.060\\
\midrule
\multicolumn{11}{c}{\makebox[0pt][c]{\textit{Panel B}: $S_2\propto e^{X/2}-Ee^{X/2}$ }}\\
\addlinespace[1pt]
0.5 & 500 & 0.369 & 0.380 & 0.227 & 0.285 & 0.213 & 0.121 & 0.199 & 0.162 & 0.136\\
 & 1,000 & 0.606 & 0.679 & 0.477 & 0.488 & 0.429 & 0.217 & 0.352 & 0.307 & 0.269\\
 & 2,000 & 0.905 & 0.932 & 0.837 & 0.808 & 0.767 & 0.393 & 0.690 & 0.672 & 0.563\\
\addlinespace[2pt]
1 & 500 & 0.273 & 0.258 & 0.168 & 0.211 & 0.183 & 0.115 & 0.178 & 0.133 & 0.125\\
 & 1,000 & 0.454 & 0.472 & 0.332 & 0.369 & 0.321 & 0.162 & 0.317 & 0.294 & 0.252\\
 & 2,000 & 0.775 & 0.772 & 0.617 & 0.697 & 0.616 & 0.255 & 0.619 & 0.614 & 0.498\\
\addlinespace[2pt]
2 & 500 & 0.147 & 0.130 & 0.080 & 0.114 & 0.100 & 0.071 & 0.108 & 0.071 & 0.065\\
 & 1,000 & 0.260 & 0.223 & 0.149 & 0.190 & 0.135 & 0.076 & 0.176 & 0.120 & 0.107\\
 & 2,000 & 0.482 & 0.416 & 0.251 & 0.391 & 0.222 & 0.093 & 0.365 & 0.224 & 0.176\\
\midrule
\multicolumn{11}{c}{\makebox[0pt][c]{\textit{Panel C}: $S_3\propto X^2-\mathrm{Var}(X)$ }}\\
\addlinespace[1pt]
0.5 & 500 & 0.340 & 0.381 & 0.251 & 0.275 & 0.221 & 0.129 & 0.176 & 0.154 & 0.137\\
 & 1,000 & 0.580 & 0.668 & 0.499 & 0.443 & 0.423 & 0.236 & 0.307 & 0.314 & 0.274\\
 & 2,000 & 0.882 & 0.939 & 0.846 & 0.774 & 0.746 & 0.421 & 0.608 & 0.658 & 0.575\\
\addlinespace[2pt]
1 & 500 & 0.241 & 0.353 & 0.290 & 0.185 & 0.228 & 0.168 & 0.135 & 0.151 & 0.139\\
 & 1,000 & 0.412 & 0.592 & 0.524 & 0.331 & 0.396 & 0.264 & 0.195 & 0.360 & 0.337\\
 & 2,000 & 0.716 & 0.852 & 0.787 & 0.613 & 0.663 & 0.422 & 0.476 & 0.759 & 0.682\\
\addlinespace[2pt]
2 & 500 & 0.098 & 0.203 & 0.235 & 0.090 & 0.127 & 0.124 & 0.079 & 0.078 & 0.086\\
 & 1,000 & 0.149 & 0.299 & 0.331 & 0.122 & 0.166 & 0.160 & 0.084 & 0.156 & 0.158\\
 & 2,000 & 0.243 & 0.397 & 0.419 & 0.235 & 0.273 & 0.207 & 0.143 & 0.349 & 0.346\\
\midrule
\multicolumn{11}{c}{\makebox[0pt][c]{\textit{Panel D}: $S_4\propto\sin(2X/\sigma_X)$ }}\\
\addlinespace[1pt]
0.5 & 500 & 0.367 & 0.358 & 0.205 & 0.266 & 0.185 & 0.097 & 0.164 & 0.123 & 0.097\\
 & 1,000 & 0.601 & 0.650 & 0.450 & 0.448 & 0.361 & 0.175 & 0.302 & 0.185 & 0.150\\
 & 2,000 & 0.894 & 0.930 & 0.800 & 0.764 & 0.680 & 0.319 & 0.583 & 0.364 & 0.257\\
\addlinespace[2pt]
1 & 500 & 0.269 & 0.284 & 0.181 & 0.176 & 0.165 & 0.099 & 0.127 & 0.076 & 0.066\\
 & 1,000 & 0.451 & 0.516 & 0.373 & 0.307 & 0.274 & 0.162 & 0.190 & 0.099 & 0.080\\
 & 2,000 & 0.768 & 0.825 & 0.709 & 0.569 & 0.572 & 0.275 & 0.378 & 0.140 & 0.091\\
\addlinespace[2pt]
2 & 500 & 0.137 & 0.161 & 0.115 & 0.100 & 0.124 & 0.096 & 0.086 & 0.062 & 0.055\\
 & 1,000 & 0.235 & 0.266 & 0.212 & 0.152 & 0.199 & 0.148 & 0.121 & 0.069 & 0.065\\
 & 2,000 & 0.418 & 0.504 & 0.443 & 0.291 & 0.365 & 0.238 & 0.241 & 0.081 & 0.067\\
\bottomrule
\end{tabular*}
\begin{tablenotes}[flushleft]\scriptsize
\item \emph{Note}: Design, notation, and scale values are as in
Table~\ref{main:tab:power_di_N_c0}. 
\end{tablenotes}
\end{threeparttable}
\end{table}

\clearpage
\setcounter{section}{0}
\counterwithout{table}{section}
\counterwithout{figure}{section}
\counterwithout{equation}{section}
\renewcommand{\thesection}{S\arabic{section}}
\renewcommand{\thetable}{S.\arabic{table}}
\renewcommand{\thefigure}{S.\arabic{figure}}
\renewcommand{\theequation}{S.\arabic{equation}}
\setcounter{table}{0}
\setcounter{figure}{0}
\setcounter{equation}{0}
\section{Proofs}\label{app:proofs}

Throughout this appendix, $\lVert \cdot \rVert$ denotes the $\mathcal{H}_k$-norm for $\mathcal{H}_k$-valued arguments, the Euclidean norm for vectors, and the spectral norm for matrices, unless stated otherwise, and $c$ denotes a generic finite constant that may change from line to line. We maintain $\mathbb{E}[\varepsilon(\theta_0)^2] < \infty$ and the second-moment condition $\mathbb{E}[\tilde h_p(\mathsf{U}, \mathsf{U}')^2] < \infty$ required by the V-statistic limit theorems. Primes on scalar random variables mark evaluation at the independent copy: $\varepsilon(\theta_0)' = \varepsilon(\theta_0)(\mathsf{U}')$ and $R(W)' = R(W')$ for $\mathsf{U}' = (Y', Z')$ an independent copy of $\mathsf{U}$; they never denote transposition.

\subsection{Proof for Section~\ref{sec:test-framework}}

\begin{lemma}\label{lem:zeroset}
Let $F:[0,1] \to \mathbb{R}$ be absolutely continuous and let $A \subseteq [0,1]$
be Lebesgue measurable with $F$ constant on $A$. Then $F' = 0$ almost everywhere
on $A$.
\end{lemma}

\begin{proof}

Since $F$ is absolutely continuous, it is differentiable at almost every point
of $(0,1)$; that is, the set
\[
N_1 = \left\{ x \in (0,1) : \text{$F$ is not differentiable at } x \right\}
\]
has Lebesgue measure zero. Separately, the Lebesgue density theorem applies to
the measurable set $A$ \citep{stein2009real}: almost every point of
$A$ is a density point of $A$, so the set
\[
N_2 = \left\{ x \in A : \text{$x$ is not a density point of } A \right\}
\]
also has measure zero. Hence, outside the zero measure set $N_1 \cup N_2 \cup \{0,1\}$,
every point of $A$ is both a differentiability point of $F$ and a density point
of $A$.

Let $p \in A \cap (0,1) \setminus (N_1 \cup N_2 \cup \{0,1\})$. It suffices to prove that $F'(p)=0$. Since $p$ is a density point of $A$, for every $n\in\mathbb{N}$ the set
\[
A \cap \left(p-\frac1n,p+\frac1n\right)
\]
has positive measure. In particular, it contains a point $p_n\neq p$. Choosing
such a point for each $n$, we have $p_n\in A$ and $p_n\to p$.

Because $F$ is constant on $A$,
\[
\frac{F(p_n)-F(p)}{p_n-p}=0
\]
for every $n$. Since $F$ is differentiable at $p$, these difference quotients
converge to $F'(p)$. Hence $F'(p)=0$.

Since $A \cap (N_1 \cup N_2 \cup \{0,1\})$ has Lebesgue measure zero, $F'=0$
almost everywhere on $A$.
\end{proof}

\begin{proof}[Proof of Proposition \ref{prop:transformed-null}]
For $d \in \{0,1\}$ let $m_d(x,u) = \mathbb{E}[Y(d) \mid X = x, U_D = u]$, so that
$\mathrm{MTE}(x,u) = m_1(x,u)-m_0(x,u)$.

\emph{Part (i): Necessity.}  By the integrability condition in
the proposition, the local instrumental variables representation holds under Assumptions~\ref{ass:exclusion}--\ref{ass:relevance},
\begin{equation}\label{eq:liv}
\mathbb{E}[Y \mid X = x, P = p]
= \mu_0(x;\beta_0) + \int_0^p \mathrm{MTE}(x,u)\,du ,
\end{equation}
for $\mathbb{P}_{P \mid X = x}$-almost every $p$
\citep{Heckman2005StructuralEvaluation}. Indeed, $YD = Y(1)\mathbf 1\{U_D \le P(Z)\}$
and $(Y(1), U_D) \perp P(Z) \mid X$ by
Assumptions~\ref{ass:exclusion} and~\ref{ass:random-assignment}, so the tower
property gives
\[
\mathbb{E}[Y(1)D \mid X = x, P = p]
= \mathbb{E}\big[Y(1)\mathbf 1\{U_D \le p\} \mid X = x\big]
= \int_0^p m_1(x,u)\,du ,
\]
the last equality by $U_D \mid X \sim \mathrm{Uniform}[0,1]$ from
Assumption~\ref{ass:selection}, the integral being finite because
$\int_0^1 \lvert m_d(x,u)\rvert\,du \le \mathbb{E}[\lvert Y(d)\rvert \mid X = x] < \infty$.
Symmetrically $\mathbb{E}[Y(0)(1-D) \mid X = x, P = p] = \int_p^1 m_0(x,u)\,du$,
and $\int_0^1 m_0(x,u)\,du = \mathbb{E}[Y(0) \mid X = x] = \mu_0(x;\beta_0)$ by
the maintained outcome means; adding the two pieces gives \eqref{eq:liv}.

Under $H_0$, the MTE does not depend on $u$, so that
$\mathrm{MTE}(x,u) = \Delta(x) \equiv \mathbb{E}[Y(1) - Y(0) \mid X = x]$ for almost
every $u \in (0,1)$, and the integral in \eqref{eq:liv} equals $p\,\Delta(x)$. Substituting the specified outcome means
$\mathbb{E}[Y(d) \mid X = x] = \mu_d(x;\beta_d)$, $d \in \{0,1\}$, then gives
$\Delta(x) = \mu_1(x;\beta_1) - \mu_0(x;\beta_0)$ and hence
\[
\mathbb{E}[Y \mid X = x, P = p]
= \mu_0(x;\beta_0) + p\,\big[\mu_1(x;\beta_1) - \mu_0(x;\beta_0)\big]
\]
for almost every $x$ and $\mathbb{P}_{P \mid X = x}$-almost every $p$, which is
\eqref{H_0: Linear Model}.

\emph{Part (ii): Sharpness.}
Part (ii) is the converse of (i): it shows that \eqref{H_0: affine} is not merely
necessary but the sharp characterization the data can deliver. The proof proceeds in steps: we first collect
the objects and identities established in Part~(i) (Step~1); deduce that $G_x$ is
affine on the support (Step~2); construct the potential outcomes (Step~3); verify
that the constructed model satisfies $H_0$ (Step~4), reproduces the observed
distribution (Step~5), and is admissible (Step~6); and conclude.

\emph{Step 1: Notation and preliminaries.} 
Recall $m_d(x,u) = \mathbb{E}[Y(d) \mid X = x, U_D = u]$, $d \in \{0,1\}$, and the partial-mean identities
$\mathbb{E}[YD \mid X = x, P = p] = M_1^x(p)$ and
$\mathbb{E}[Y(1-D) \mid X = x, P = p] = M_0^x(p)$; then we write
\begin{equation}\label{eq:lativ}
    G_x(p) = \mu_0(x;\beta_0) + \int_0^p \mathrm{MTE}(x,u)\,du,
\qquad p \in [0,1].
\end{equation}
Adding the two partial-mean identities gives
$\mathbb{E}[Y \mid X = x, P = p] = G_x(p)$ for
$\mathbb{P}_{P \mid X = x}$-almost every $p \in \mathcal{S}_x$:

Since $m_0(x,\cdot), m_1(x,\cdot) \in L^1(0,1)$, the functions $M_1^x$, $M_0^x$
and $G_x$ are absolutely continuous on $[0,1]$, with
$G_x' = \mathrm{MTE}(x,\cdot)$ and $(M_1^x)' = m_1(x,\cdot)$,
$(M_0^x)' = -m_0(x,\cdot)$ almost everywhere; no smoothness beyond this is
assumed. Substituting the indicator-valued versions
$(y_1,u) \mapsto \mathbf 1\{y_1 \le y\}\mathbf 1\{u \le p\}$ and
$(y_0,u) \mapsto \mathbf 1\{y_0 \le y\}\mathbf 1\{u > p\}$ for
$(y_1,u) \mapsto y_1\mathbf 1\{u \le p\}$ and
$(y_0,u) \mapsto y_0\mathbf 1\{u > p\}$ in the derivation of $M_1^x$ and $M_0^x$ gives, for
every $y$ and almost every $p \in \mathcal{S}_x$,
\begin{equation}\label{eq:partial-laws}
    \begin{split}
\Pr(Y \le y, D = 1 \mid X = x, P = p) &= \int_0^p F_1(y \mid x,u)\,du,\\
\Pr(Y \le y, D = 0 \mid X = x, P = p) &= \int_p^1 F_0(y \mid x,u)\,du,
\end{split}
\end{equation}
where $F_d(\cdot \mid x,u)$ denotes the conditional law of $Y(d)$ given
$(X, U_D) = (x,u)$, so that $m_d(x,u) = \int y\,F_d(dy\mid x,u)$; these partial
laws are what the construction below must reproduce.

\emph{Step 2: The affine conditional mean forces $G_x$ to be affine on the support.} 
Write $F_x(p) \equiv G_x(p) - a(x) - p\,b(x)$, an absolutely continuous function
on $[0,1]$. By \eqref{H_0: affine} and the local IV representation,
$F_x(P) = 0$ almost surely under $\mathbb{P}_{P \mid X = x}$.
 Continuity
carries this from $\mathbb{P}_{P \mid X = x}$-almost sure to Lebesgue-almost
everywhere on $\mathcal{S}_x$. Indeed, if $F_x(p_0) \ne 0$ for some
$p_0 \in \mathcal{S}_x$, continuity would give an open neighborhood
$\mathcal{O}$ of $p_0$ on which $F_x$ is nonzero. Because $p_0$ belongs
to the support, $\mathbb{P}_{P \mid X = x}(\mathcal{O}) > 0$, contradicting
$F_x(P) = 0$ almost surely. Thus $F_x = 0$ on $\mathcal{S}_x$.

By Assumption~\ref{ass:relevance}, $\mathcal{S}_x$ is a finite union of
intervals. Since it is the support of the conditional distribution, it is
closed; write $\mathcal{S}_x = \bigcup_{j=1}^J [s_j,t_j]$. By the argument
above, $F_x$ vanishes at every point of $\mathcal{S}_x$, including the
interval endpoints. Hence
\[
C_x := \overline{\mathcal{S}_x} = \mathcal{S}_x,
\qquad
G_x(p) = a(x) + p\,b(x) \quad \text{for every } p \in C_x.
\]

\emph{Step 3: Construction of the potential outcomes.}
It remains to construct potential outcomes and a selection unobservable that
satisfy Assumptions~\ref{ass:exclusion}--\ref{ass:relevance} and $H_0$ in \eqref{H_0: affine}
whose observed distribution of $(Y, D, Z)$ coincides with that of the data. We
keep $Z$, $U_D$, and the threshold-crossing rule \eqref{eq:selection}
unchanged, so that $\tilde U_D = U_D$ and $\tilde D = D$; only the conditional
laws $\tilde F_d(\cdot \mid x,u)$ of $\tilde Y(d)$ given $(X, U_D) = (x,u)$
remain to be chosen. The affinity of $G_x$ on $C_x$ established above is
what makes a suitable choice available.

We define $\tilde F_d$ separately on the support, the gaps between its
components, and the regions outside the support. First, on $C_x$, retain the
original conditional laws:
\[
\tilde F_d(\cdot \mid x,u)=F_d(\cdot \mid x,u), \qquad u\in C_x.
\]
Next, on each gap $(t_j,s_{j+1})$ between adjacent support intervals, replace
the conditional law by its average over that gap:
\[
\tilde F_d(\cdot \mid x,u)
= \frac{1}{s_{j+1}-t_j}
  \int_{t_j}^{s_{j+1}} F_d(\cdot \mid x,v)\,dv,
\qquad u\in(t_j,s_{j+1}).
\]
This preserves the integrated law over each gap. Below the support, if
$\underline p>0$, set $\tilde F_1$ equal to the average of the original
treated-outcome laws on $[0,\underline p]$, and define $\tilde F_0$ as the
law shifted downward by $b(x)$:
\[
\tilde F_1(\cdot \mid x,u)
= \frac{1}{\underline p}\int_0^{\underline p}
  F_1(\cdot \mid x,v)\,dv,
\qquad
\tilde F_0(y \mid x,u)=\tilde F_1(y+b(x)\mid x,u),
\qquad u<\underline p.
\]
Similarly, above the support, if $\bar p<1$, set $\tilde F_0$ equal to the
average of the original control-outcome laws on $[\bar p,1]$, and define
$\tilde F_1$ as the law shifted upward by $b(x)$:
\[
\tilde F_0(\cdot \mid x,u)
= \frac{1}{1-\bar p}\int_{\bar p}^{1}
  F_0(\cdot \mid x,v)\,dv,
\qquad
\tilde F_1(y \mid x,u)=\tilde F_0(y-b(x)\mid x,u),
\qquad u>\bar p.
\]

\emph{Step 4: Verify
that the constructed model satisfies $H_0$.}
We now verify that the constructed model satisfies $H_0$. Fix $x$, and define
the conditional mean associated with $\tilde F_d$ by
\[
\tilde m_d(x,u)
:= \int_{\mathbb R} y\,\tilde F_d(dy \mid x,u),
\qquad d\in\{0,1\}.
\]
For $u\in C_x$, the construction gives
$\tilde F_d(\cdot\mid x,u)=F_d(\cdot\mid x,u)$, and hence
$\tilde m_d(x,u)=m_d(x,u)$. Moreover, $G_x$ is absolutely continuous and
affine on each interval making up $C_x$, so
$\mathrm{MTE}(x,u)=G_x'(u)=b(x)$ for Lebesgue-almost every $u\in C_x$.
Therefore,
\[
\tilde m_1(x,u)-\tilde m_0(x,u)=b(x)
\qquad\text{for Lebesgue-almost every }u\in C_x.
\]

For $u\in (t_j, s_{j+1})$, inserting the definition of $\tilde F_d$ and
applying Fubini's theorem,
\[
\tilde m_d(x,u) = \int y\,\tilde F_d(dy \mid x,u)
= \frac{1}{s_{j+1}-t_j}\int_{t_j}^{s_{j+1}}\int y\,F_d(dy \mid x,v)\,dv
= \frac{1}{s_{j+1}-t_j}\int_{t_j}^{s_{j+1}} m_d(x,v)\,dv ,
\]
and therefore
\begin{align*}
\tilde m_1(x,u) - \tilde m_0(x,u)
&= \frac{1}{s_{j+1}-t_j}\int_{t_j}^{s_{j+1}} \mathrm{MTE}(x,v)\,dv
 = \frac{G_x(s_{j+1}) - G_x(t_j)}{s_{j+1}-t_j} \\
&= \frac{\{a(x) + s_{j+1} b(x)\} - \{a(x) + t_j b(x)\}}{s_{j+1}-t_j}
 = b(x),
\end{align*}
where the second equality is the definition of $G_x$ and the third uses the
affinity of $G_x$ at the endpoints $t_j, s_{j+1} \in C_x$.

Below the support, fix $u<\underline p$ and let $W_1 \sim \tilde F_1(\cdot \mid x,u) = \frac{1}{\underline p}\int_0^{\underline p}
  F_1(\cdot\mid x,v)\,dv$. By the construction of $\tilde F_0$, its distribution is that of
$W_1-b(x)$, since
\[
\Pr\{W_1-b(x)\le y\}
= \tilde F_1(y+b(x)\mid x,u)
= \tilde F_0(y\mid x,u).
\]
Consequently, $\tilde m_0(x,u)
= \mathbb E[W_1]-b(x)
= \tilde m_1(x,u)-b(x)$. 

Symmetrically, above the support, $u > \bar p$, $\tilde F_1(\cdot \mid x,u)$
is the law of $W_0 + b(x)$ with $W_0 \sim \tilde F_0(\cdot \mid x,u)= \frac{1}{1-\bar p}\int_{\bar p}^{1} F_0(\cdot\mid x,v)\,dv$, so
$\tilde m_1(x,u) = \tilde m_0(x,u) + b(x)$. The marginal treatment effect of
the constructed model is therefore $b(x)$ for almost every $u \in (0,1)$.

\emph{Step 5: The data are reproduced.} Under Assumptions~\ref{ass:exclusion}--\ref{ass:selection}, the conditional
law of $(Y,D)$ given $Z$ depends on $Z$ only through $(X,P(Z))$. Since the
law of $Z$ is unchanged, it suffices to reproduce the conditional law of
$(Y,D)$ given $(X,P)$.
By \eqref{eq:partial-laws}, this amounts to showing that replacing $F_d$
by $\tilde F_d$ leaves both expressions unchanged for every $y\in\mathbb R$
and almost every $p\in\mathcal S_x$.

Fix such $y$ and $p$, with $p\in[s_k,t_k]$.
Splitting the integral over $[0,p]$ into the lower tail, the support
intervals, the intervening gaps, and the final segment gives
\[
\begin{aligned}
\int_0^p \tilde F_1(y\mid x,u)\,du
={}& \int_0^{\underline p}\tilde F_1(y\mid x,u)\,du \\
&+ \sum_{j<k}\left\{
\int_{s_j}^{t_j}\tilde F_1(y\mid x,u)\,du
+\int_{t_j}^{s_{j+1}}\tilde F_1(y\mid x,u)\,du
\right\}
+\int_{s_k}^{p}\tilde F_1(y\mid x,u)\,du .
\end{aligned}
\]
Each part agrees with the corresponding integral under $F_1$. On the lower
tail, the construction makes $\tilde F_1(y\mid x,u)$ constant in $u$ and
equal to the average of $F_1(y\mid x,\cdot)$ over $[0,\underline p]$; hence
\[
\int_0^{\underline p}\tilde F_1(y\mid x,u)\,du
=\underline p\left\{\frac{1}{\underline p}
\int_0^{\underline p}F_1(y\mid x,v)\,dv\right\}
=\int_0^{\underline p}F_1(y\mid x,v)\,dv.
\]
On each support interval, $\tilde F_1=F_1$. On each gap
$(t_j,s_{j+1})$, $\tilde F_1(y\mid x,u)$ is the average of
$F_1(y\mid x,\cdot)$ over that gap, so its integral over the gap equals
the integral of $F_1$ there. Therefore,
\[
\int_0^p \tilde F_1(y\mid x,u)\,du
=\int_0^p F_1(y\mid x,u)\,du.
\]
The same argument, applied to the upper tail, support intervals, and gaps
in $[p,1]$, gives
\[
\int_p^1 \tilde F_0(y\mid x,u)\,du
=\int_p^1 F_0(y\mid x,u)\,du.
\]

\emph{Step 6: The construction is admissible.} Retain $(X, Z_0, U_D)$ and the
selection rule \eqref{eq:selection}, so that $\tilde U_D = U_D$,
$\tilde D = D$, and the law of $Z$ is unchanged. Let
$\xi_0, \xi_1 \sim \mathrm{Uniform}[0,1]$ be independent of each other and of
$(X, Z_0, U_D)$, and set
$\tilde Y(d) = \tilde F_d^{-1}(\xi_d \mid X, U_D)$ with $\tilde F_d^{-1}$ the
quantile function. Since $\tilde F_0$ and $\tilde F_1$ are measurable
conditional distributions, the conditional joint law of
$(\tilde Y(0),\tilde Y(1))$ given $X=x,U_D=u$ is the product measure
$\tilde F_0(\cdot\mid x,u)\otimes \tilde F_1(\cdot\mid x,u)$. 

Because the outcomes are real-valued, this conditional law can be realized
on an extension of the probability space, independently of $Z_0$
conditional on $(X,U_D)$. Its marginals are the prescribed
$\tilde F_d$, so the constructed potential outcomes are well defined and
satisfy Assumption~\ref{ass:exclusion}. $(\tilde Y(0), \tilde Y(1), U_D)$ is a
function of $(X, U_D, \xi_0, \xi_1)$ with
$(U_D, \xi_0, \xi_1) \perp Z_0 \mid X$, giving
Assumption~\ref{ass:random-assignment}; Assumptions~\ref{ass:selection}
and~\ref{ass:relevance} restrict only the law of $(U_D, P)$, which is
unchanged. The constructed
model therefore satisfies Assumptions~\ref{ass:exclusion}--\ref{ass:relevance},
and it satisfies $H_0$, as shown above.

\emph{Conclusion.} By the argument following \eqref{eq:partial-laws}, which
applies to any model satisfying these assumptions, the partial laws determine
the joint law of $(Y, D, Z)$. Since they coincide in the two models, as
verified above, $(\tilde Y, D, Z)$ and $(Y, D, Z)$ have the same distribution.
A model satisfying Assumptions~\ref{ass:exclusion}--\ref{ass:relevance} and
$H_0$ thus generates the distribution of the data, and no procedure based on
that distribution can separate $H_0$ from \eqref{H_0: affine}. \qedhere
\end{proof}

\begin{remark}[Bounded outcomes]\label{rem:bounded}

Below the lower support bound, every unit is treated at every observed
propensity score, so the data identify the average law of $Y(1)$ over
$[0,\underline p]$, but not the law of $Y(0)$ there. Above the upper support
bound, the analogous statement holds for $Y(0)$ over $[\bar p,1]$. Interior
gaps between support intervals are handled separately by the averaging
construction in Step~3. Under $H_0$, the unobserved potential-outcome laws
in these two tails can be constructed by shifting the corresponding observed
average laws by the constant MTE $b(x)$. Write
\[
A_1(x)=\frac{M_1^x(\underline p)}{\underline p},
\qquad
A_0(x)=\frac{M_0^x(\bar p)}{1-\bar p}.
\]
The corresponding unobserved means under this constant-MTE construction
are $A_1(x)-b(x)$ and $A_0(x)+b(x)$. If outcomes are bounded,
$Y(d)\in[y_L,y_U]$, these means must also lie in $[y_L,y_U]$; otherwise no
bounded-outcome model satisfying $H_0$ can reproduce the data. Thus, when
$\underline p>0$ and $\bar p<1$, boundedness imposes the additional
restrictions
\[
A_1(x)-y_U \le b(x) \le A_1(x)-y_L,
\qquad
y_L-A_0(x) \le b(x) \le y_U-A_0(x).
\]

\end{remark}

\begin{proof}[Proof of Corollary \ref{cor:mte-on-support}]
Suppose \eqref{H_0: affine} holds and define
$F_x(p):=G_x(p)-a(x)-p\,b(x)$. By Step~1, $F_x$ is absolutely continuous
on $[0,1]$ and
$F_x'=\mathrm{MTE}(x,\cdot)-b(x)$ almost everywhere. Moreover,
\eqref{eq:lativ} and \eqref{H_0: affine} imply
$F_x(P)=0$ almost surely conditional on $X=x$. Since $F_x$ is continuous,
the support argument in Step~2 gives $F_x(p)=0$ for every
$p\in\mathcal S_x$.

Applying Lemma~\ref{lem:zeroset} with
$A = \{p \in [0,1] : F_x(p) = 0\}$, which contains $\mathcal{S}_x$,
gives $F_x' = 0$ almost everywhere on $\mathcal{S}_x$, that is,
\[
\mathrm{MTE}(x,u) = b(x)
\qquad \text{for almost every } u \in \mathcal{S}_x,
\]
for almost every $x$. \qedhere
\end{proof}

\begin{proof}[Proof of Corollary \ref{cor:separable}]
Under Assumption~\ref{ass:additive} the deviation collapses:
$\mathrm{MTE}(x,u) - \Delta(x) = h(u)$, where
$\Delta(x) = \mu_1(x;\beta_1) - \mu_0(x;\beta_0)$ and
$h(u) = \mathbb{E}[V(1) - V(0) \mid U_D = u]$. The local IV representation of
Step 1 therefore gives the explicit regression
\[
\mathbb{E}[Y \mid X = x, P = p]
= \mu_0(x;\beta_0) + p\,\Delta(x) + \phi(p),
\qquad
\phi(p) := \int_0^p h(u)\,du .
\]
The right-hand side is an absolutely continuous function of $p$ defined on all
of $[0,1]$. With
$\varepsilon = Y - \mu_0(X;\beta_0) - P\,\Delta(X)$ this yields the collapse
$\mathbb{E}[\varepsilon \mid X, P] = \phi(P)$ and, by the law of iterated
expectations, $\mathbb{E}[\varepsilon \mid P] = \phi(P)$.

It remains to prove \eqref{eq:separable-full}. Suppose
$\mathbb{E}[\varepsilon \mid P] = 0$ almost surely, that is $\phi(P) = 0$
almost surely. By full support of $P$ on $(0,1)$, this is equivalent to
$\phi(p) = 0$ for Lebesgue-almost every $p \in (0,1)$. Since $\phi$ is
absolutely continuous, Lemma~\ref{lem:zeroset} applied with
$A = \{p \in (0,1) : \phi(p) = 0\}$ gives $h = \phi' = 0$ almost everywhere
on $(0,1)$. Thus
$\mathrm{MTE}(x,u) = \Delta(x)$ for every $x$ and almost every $u \in (0,1)$,
which is $H_0$ in \eqref{H_0: Null}. Conversely, if $H_0$ holds, then $\mathrm{MTE}(x,u) = \mathrm{CATE}(x)$ for
almost every $u$, and
$\mathrm{CATE}(x) = \mathbb{E}[Y(1)-Y(0) \mid X = x]
= \Delta(x)$ by the definition. Under \eqref{ass:additive} the
same object equals $\Delta(x)+h(u)$, so $h = 0$ almost everywhere. Hence
$\phi \equiv 0$ and $\mathbb{E}[\varepsilon \mid P] = 0$.
\qedhere
\end{proof}

\subsection{Proof for Section~\ref{sec:test}}

\begin{proof}[Proof of Theorem \ref{Theorem: Residual Embedding}]
Let $m(W) = \mathbb{E}[\varepsilon(\theta_0)\mid W] \in L^2(\mathbb{P}_W)$ and consider the linear map $M: L^2(\mathbb{P}_W) \to \mathcal{H}_k$, $M(f) = \int f(w)\, k(w,\cdot)\, d\mathbb{P}_W(w)$. By the Cauchy--Schwarz inequality and the reproducing property,
\[
\lVert M(f)\rVert \le \int \lvert f(w)\rvert\, k(w,w)^{1/2}\, d\mathbb{P}_W(w) \le \lVert f\rVert_{L^2}\left(\int k(w,w)\, d\mathbb{P}_W(w)\right)^{1/2} < \infty,
\]
where the last bound uses boundedness of $k$; hence $M$ is a well-defined bounded Bochner integral, and $\boldsymbol{m}_{\varepsilon(\theta_0)} = M(m)$ is well defined under $\mathbb{E}\lvert\varepsilon(\theta_0)\rvert < \infty$.

If $\mathbb{E}[\varepsilon(\theta_0)\mid W] = 0$, the law of iterated expectations gives $\boldsymbol{m}_{\varepsilon(\theta_0)} = \mathbb{E}[\mathbb{E}[\varepsilon(\theta_0)\mid W]\,k(W,\cdot)] = \boldsymbol{0}$. Conversely, suppose $\boldsymbol{m}_{\varepsilon(\theta_0)} = \boldsymbol{0}$. By the reproducing property, for all $f, g \in L^2(\mathbb{P}_W)$,
\[
\left\langle M(f), M(g)\right\rangle = \iint f(w)\, k(w,w')\, g(w')\, d\mathbb{P}_W(w)\, d\mathbb{P}_W(w').
\]
Taking $f = g = m$, the ISPD property \eqref{eq:ispd} yields $\lVert M(m)\rVert^2 = \iint m(w)k(w,w')m(w')\,d\mathbb{P}_W\,d\mathbb{P}_W > 0$ whenever $m \ne 0$ in $L^2(\mathbb{P}_W)$; hence $M$ is injective and $m = 0$ almost surely.

It remains to verify the identification of $T^2(\theta_0)$ with the squared norm of the embedding. Write $\varepsilon(\theta_0) = m(W) + \eta$ with $\mathbb{E}[\eta \mid W] = 0$, and let $(W',\varepsilon')$ be an independent copy. The cross terms vanish pointwise: $\mathbb{E}[\eta' k(W,W')\mid W] = \mathbb{E}[k(W,W')\mathbb{E}[\eta'\mid W']\mid W] = 0$, and similarly for $\mathbb{E}[\eta k(W,W')\mid W']$ and for the $\eta\eta'$ term. Therefore
\[
T^2(\theta_0) = \mathbb{E}\left[\varepsilon(\theta_0) k(W,W') \varepsilon(\theta_0)'\right] = \iint m(w)k(w,w')m(w')\, d\mathbb{P}_W(w)\, d\mathbb{P}_W(w') = \lVert \boldsymbol{m}_{\varepsilon(\theta_0)}\rVert^2,
\]
which vanishes if and only if $m = 0$ almost surely, i.e.\ exactly under the null. \qedhere
\end{proof}

\emph{Discussion.} Theorem~\ref{Theorem: Residual Embedding} underlines why the statistic \eqref{Test Statistics} is not a local-smoothing device: the conditional moment restriction is summarized by a single element of $\mathcal{H}_k$, no nonparametric estimator of $\mathbb{E}[\varepsilon \mid W]$ is involved, and the kernel scale is a fixed tuning constant rather than a bandwidth tending to zero. Consequently the statistic converges at the parametric rate regardless of the dimension of $W$, which is what makes the omnibus case $W = (X^\top, P)^\top$ workable; see also Remark~\ref{rem:kernel-choice}.

The sample projection in \eqref{Test Statistics local Projection} is built from the estimated scores $\hat g_i = \nabla_\theta \varepsilon_i(\theta)\big\vert_{\theta = \hat\theta_n}$, whose probability limits are the raw derivatives $\nabla_\theta \varepsilon(\theta_0)$ rather than conditional expectations thereof: the $\gamma$-block of $\nabla_\theta \varepsilon(\theta_0)$ is $-\partial_\gamma P(Z; \gamma_0)\, \Delta(X)$, which is not measurable with respect to $W$. The population analysis is therefore developed in terms of these raw scores. Let $\mathsf{U} = (Y, Z)$ denote a generic observation and $\mathsf{U}'$ an independent copy. For $\theta \in \Theta$, write $\tilde g_\theta(\mathsf{U}) = \nabla_\theta \varepsilon(\theta)$, a measurable function of $Z$ alone, and define
\begin{equation}\label{eq:raw-objects}
\tilde b_\theta(w') = \mathbb{E}\big[\tilde g_\theta(\mathsf{U})\, k(W, w')\big],
\qquad
\tilde\Gamma_\theta = \mathbb{E}\big[\tilde g_\theta(\mathsf{U})\tilde g_\theta(\mathsf{U})^\top\big],
\qquad
\tilde M_\theta = \mathbb{E}\big[\tilde g_\theta(\mathsf{U})\tilde g_\theta(\mathsf{U}')^\top k(W, W')\big],
\end{equation}
so that $\tilde b_\theta$ is $\mathcal{H}_k$-valued in each coordinate and $\tilde\Gamma_\theta$, $\tilde M_\theta$ are $(2K+L) \times (2K+L)$ matrices; integrability of $\tilde g_\theta$ follows from Assumption~\ref{Assumption: residual regularity}(i) below. At every $\theta$ with nonsingular $\tilde\Gamma_\theta$, define the projected feature
\begin{equation}\label{eq:Pi-section}
(\Pi_\theta k)(\mathsf{u}, \cdot) = k(W, \cdot) - \tilde g_\theta(\mathsf{u})^\top \tilde\Gamma_\theta^{-1}\, \tilde b_\theta(\cdot) \;\in\; \mathcal{H}_k,
\end{equation}
the induced bivariate kernel $\tilde k_{p,\theta}(\mathsf{u}, \mathsf{u}') = \big\langle (\Pi_\theta k)(\mathsf{u}, \cdot), (\Pi_\theta k)(\mathsf{u}', \cdot) \big\rangle$, and the kernel $\tilde h_{p,\theta}(\mathsf{u}, \mathsf{u}') = \varepsilon(\theta)\, \tilde k_{p,\theta}(\mathsf{u}, \mathsf{u}')\, \varepsilon(\theta)'$. Throughout, subscripts are dropped when these objects are evaluated at $\theta_0$.

\begin{lemma}\label{lem:centering}
Suppose Assumptions~\ref{ass:exclusion}--\ref{ass:relevance} hold, the first
stage is correctly specified, $P(Z;\gamma_0) = P(Z)$, and
\eqref{H_0: Linear Model} holds at $\theta_0$, and
Assumption~\ref{Assumption: residual regularity} holds. Then
$\mathbb{E}[\varepsilon(\theta_0) \mid Z] = 0$ almost surely, and consequently
\[
\mathbb{E}\big[\varepsilon(\theta_0)
\, (\Pi k)(\mathsf{U}, \cdot)\big] = \boldsymbol{0} \quad \text{in } \mathcal{H}_k,
\qquad
\mathbb{E}\big[\varepsilon(\theta_0)\, \tilde k_p(\mathsf{U}, \mathsf{U}') \mid \mathsf{U}'\big] = 0 \quad \text{a.s.}
\]
\end{lemma}

\begin{proof}

By Assumptions~\ref{ass:exclusion}--\ref{ass:selection} and the
threshold-crossing model, the conditional distribution of $Y$ given $Z$
depends on $Z$ only through $(X,P)$. Hence
$\mathbb{E}[Y\mid Z]=\mathbb{E}[Y\mid X,P]$ almost surely. Under
\eqref{H_0: Linear Model} and the correctly specified first stage,
\[
\mathbb{E}[Y\mid Z]
=\mu_0(X;\beta_0)+P\,\Delta(X).
\]
Moreover, \eqref{eq:residual} gives
$Y-\varepsilon(\theta_0)=\mu_0(X;\beta_0)+P\,\Delta(X)$ almost surely.
Taking conditional expectations given $Z$ and comparing the two expressions
yields $\mathbb{E}[\varepsilon(\theta_0)\mid Z]=0$ almost surely.

For the two consequences, note that $(\Pi k)(\mathsf{U}, \cdot)$ is $\sigma(Z)$-measurable
and $\tilde k_p(\mathsf{U}, \mathsf{U}')$ is $\sigma(Z, Z')$-measurable, since $W$ and $\tilde g$
are functions of $Z$ and $\tilde b$ of $Z'$. The tower property for Bochner
integrals, which follows by approximating the integrand by simple functions, gives
$\mathbb{E}[\varepsilon(\theta_0) (\Pi k)(\mathsf{U}, \cdot)] =
\mathbb{E}\big[\mathbb{E}[\varepsilon(\theta_0) \mid Z] (\Pi k)(\mathsf{U}, \cdot)\big] =
\boldsymbol{0}$. For the second display, fix $\mathsf{u}'$ and compute the section:
$\int \varepsilon(\mathsf{u})\, \tilde k_p(\mathsf{u}, \mathsf{u}')\, d\mathbb{P}_{\mathsf{U}}(\mathsf{u}) =
\mathbb{E}\big[\tilde k_p(\mathsf{U}, \mathsf{u}')\, \mathbb{E}[\varepsilon(\theta_0) \mid Z]\big] = 0$,
using $\mathsf{U} \perp \mathsf{U}'$; since this holds for every $\mathsf{u}'$, the conditional expectation
vanishes almost surely.
\end{proof}

\begin{lemma}\label{lem:orthogonal}
Suppose that $\tilde\Gamma_{\theta_0}$ is nonsingular,
$\mathbb{E}[\lVert \tilde g(\mathsf{U}) \rVert^2] < \infty$, and
Assumption~\ref{Assumption: residual regularity} holds. Then:
(i) $(\Pi k)(\mathsf{u}, \cdot) \in \mathcal{H}_k$ for every $\mathsf{u}$, and the bivariate projected kernel is
\begin{equation}\label{eq:bivariate-kp}
\begin{aligned}
    \tilde k_p(\mathsf{u}, \mathsf{u}') &:= \left\langle (\Pi k)(\mathsf{u}, \cdot), (\Pi k)(\mathsf{u}', \cdot)\right\rangle \\
    &= k(w, w') - \tilde g(\mathsf{u})^\top \tilde\Gamma^{-1} \tilde b(w') - \tilde g(\mathsf{u}')^\top \tilde\Gamma^{-1} \tilde b(w) \\
    &\quad + \tilde g(\mathsf{u})^\top \tilde\Gamma^{-1} \tilde M \tilde\Gamma^{-1} \tilde g(\mathsf{u}'),
\end{aligned}
\end{equation}
which is the population counterpart of the sample entries
$(\hat{\boldsymbol{\Pi}}^\top \hat{\boldsymbol{K}} \hat{\boldsymbol{\Pi}})_{ij}$ of
\eqref{Test Statistics local Projection} (Lemma~\ref{lem:sampleprojection} below);
(ii) the moment functional $m(\theta) = \mathbb{E}[\varepsilon(\theta) (\Pi k)(\mathsf{U}, \cdot)]$,
with the projection frozen at $\theta_0$, is Neyman orthogonal:
\[
\partial_\theta\, m(\theta)\Big\vert_{\theta = \theta_0}
= \mathbb{E}\big[\tilde g(\mathsf{U}) (\Pi k)(\mathsf{U}, \cdot)\big]
= \tilde b - \tilde\Gamma \tilde\Gamma^{-1} \tilde b
= \boldsymbol{0} \quad \text{in } \mathcal{H}_k .
\]
\end{lemma}
\begin{proof}
(i) Each coordinate of $\tilde b(\cdot) = \mathbb{E}[\tilde g(\mathsf{U}) k(W, \cdot)]$ is a
Bochner integral of $\mathcal{H}_k$-valued elements, bounded in norm by
$\mathbb{E}\lVert \tilde g(\mathsf{U}) \rVert\, k(w, w)^{1/2} < \infty$ by the
Cauchy--Schwarz inequality and boundedness of $k$; hence
$\tilde b \in \mathcal{H}_k \otimes \mathbb{R}^{2K+L}$ and the sections in
\eqref{eq:Pi-section} belong to $\mathcal{H}_k$ for every $\mathsf{u}$; no membership
condition on $\tilde g$ itself is required. The four-term form follows from the
reproducing property: by symmetry of $k$,
$\langle k(w, \cdot), \tilde b(\cdot) \rangle = \mathbb{E}[k(w, W)\tilde g(\mathsf{U})] = \tilde b(w)$
coordinatewise, and
$\big\langle \tilde b(\cdot), \tilde b(\cdot) \big\rangle =
\mathbb{E}[\tilde g(\mathsf{U})\tilde g(\mathsf{U}')^\top k(W, W')] = \tilde M$. For the sample
counterpart,
$(\hat{\boldsymbol{\Pi}}^\top \hat{\boldsymbol{K}} \hat{\boldsymbol{\Pi}})_{ij}
= \big\langle (\hat{\boldsymbol{\Pi}} \hat k)(\hat W_i, \cdot),
(\hat{\boldsymbol{\Pi}} \hat k)(\hat W_j, \cdot) \big\rangle$
with $(\hat{\boldsymbol{\Pi}} \hat k)(\hat W_i, \cdot) = k(\hat W_i, \cdot) -
\hat g_i^\top (\hat{\mathbf g}^\top \hat{\mathbf g} / n)^{-1} \hat b_n(\cdot)$ and
$\hat b_n(\cdot) = n^{-1} \sum_{\ell=1}^n \hat g_\ell k(\hat W_\ell, \cdot)$,
which reduces entrywise to the four-term form with population objects replaced by
sample analogs; Lemma~\ref{lem:sampleprojection} below makes the convergence
precise.

(ii) The projection is frozen at $\theta_0$, so only $\varepsilon(\theta)$
differentiates. For each $\mathsf{u}$, the difference quotients of
$\theta \mapsto \mathbb{E}[\varepsilon(\theta)\, (\Pi k)(\mathsf{U}, \cdot)]$ are dominated
by $\sup_{\theta \in \Theta} \lVert \nabla_\theta \varepsilon(\theta) \rVert\,
\lVert (\Pi k)(\mathsf{U}, \cdot) \rVert$, which is integrable by the Cauchy--Schwarz
inequality, since $\mathbb{E}[\sup_\theta \lVert \nabla_\theta \varepsilon(\theta) \rVert^2] < \infty$
and $\mathbb{E}\lVert (\Pi k)(\mathsf{U}, \cdot) \rVert^2 < \infty$ under
Assumption~\ref{Assumption: residual regularity}; dominated convergence
interchanges differentiation and expectation, and by \eqref{eq:raw-objects},
\begin{align*}
\partial_\theta\, m(\theta)\big\vert_{\theta_0}
&= \mathbb{E}\big[\tilde g(\mathsf{U})\, k(W, \cdot)\big]
   - \mathbb{E}\big[\tilde g(\mathsf{U})\tilde g(\mathsf{U})^\top\big]\, \tilde\Gamma^{-1} \tilde b(\cdot) \\
&= \tilde b(\cdot) - \tilde\Gamma\tilde\Gamma^{-1} \tilde b(\cdot) = \boldsymbol{0} .
\end{align*}
\end{proof}

For raw scores the conditional centering
$\mathbb{E}[\tilde g(\mathsf{U})\, \tilde k_p(\mathsf{U}, \mathsf{U}') \mid \mathsf{U}'] = \boldsymbol{0}$ need not
hold, and no step below uses it: the arguments consume only the orthogonality in
(ii) and, under the null, the economic centering of Lemma~\ref{lem:centering}.

Under the omnibus conditioning set $W = (X^\top, P)^\top$, the hypothesis of
Lemma~\ref{lem:centering} is $H_0^{(2)}$ under the correctly specified first
stage, i.e.\
$\mathbb{E}[\varepsilon(\theta_0) \mid X, P] = 0$ almost surely; the same holds for the directional set $W = P$,
since $\mathbb{E}[\varepsilon(\theta_0) \mid P]
= \mathbb{E}\big[\mathbb{E}[\varepsilon(\theta_0) \mid X, P] \mid P\big] = 0$
by the tower property.

\begin{lemma}\label{lem:invariance}
Suppose Assumptions~\ref{Assumption: parameter space} and~\ref{Assumption: residual regularity} hold
and $\lVert \sqrt{n}\,\hat m(\theta_0) \rVert = O_p(1)$. Then
\[
n\left(\hat T_{np}^2(\hat\theta_n) - \hat T_{np}^2(\theta_0)\right) = o_p(1),
\]
where $\hat T_{np}^2(\theta) = n^{-2}\,\boldsymbol{\varepsilon}(\theta)^\top \hat{\boldsymbol{\Pi}}^\top \hat{\boldsymbol{K}} \hat{\boldsymbol{\Pi}}\,\boldsymbol{\varepsilon}(\theta)$.
\end{lemma}
\begin{proof}
Write $(\hat{\boldsymbol{\Pi}} \hat k)(\hat W_i, \cdot)$ for the sample section defined in the proof of Lemma~\ref{lem:orthogonal} and
$\hat m(\theta) = n^{-1}\sum_{i=1}^n \varepsilon_i(\theta)
\, (\hat{\boldsymbol{\Pi}}\hat k)(\hat W_i,\cdot) \in \mathcal{H}_k$, so that
$n\hat T_{np}^2(\theta) = \lVert \sqrt{n}\,\hat m(\theta)\rVert^2$. The projection,
the kernel matrix, and $\hat W$ are frozen as $\theta$ varies, so
$\theta \mapsto \hat m(\theta)$ enters only through $\varepsilon_i(\theta)$, and
componentwise application of the mean value theorem gives
\[
\sqrt{n}\left(\hat m(\hat\theta_n) - \hat m(\theta_0)\right)
= \sqrt{n}\,(\hat\theta_n - \theta_0)^\top \frac{\partial \hat m}{\partial \theta}(\bar\theta_n)
\]
for some $\bar\theta_n$ on the segment joining $\hat\theta_n$ and $\theta_0$. The
key estimate bounds the derivative at $\theta_0$: by the decomposition of
$(\hat{\boldsymbol{\Pi}}\hat k)(\hat W_i, \cdot)$ into its population and
replacement parts,
\[
\frac{\partial \hat m}{\partial \theta}(\theta_0)
= \frac{1}{n}\sum_{i=1}^n \tilde g(\mathsf{U}_i) (\Pi k)(\mathsf{U}_i, \cdot) + \bar r_n,
\qquad
\lVert \bar r_n \rVert = O_p(n^{-\alpha} + n^{-1/2}) .
\]
Indeed, the first term is $O_p(n^{-1/2})$ by the central limit theorem, its mean
$\mathbb{E}[\tilde g(\mathsf{U})(\Pi k)(\mathsf{U}, \cdot)]$ vanishing by
Lemma~\ref{lem:orthogonal}(ii). Writing
$L_{2,i} = \sup_{\theta \in \Theta} \lVert \nabla^2_\theta \varepsilon_i(\theta) \rVert$,
the mean value theorem gives $\hat g_i - \tilde g(\mathsf{U}_i) = A_i (\hat\theta_n - \theta_0)$
with $\lVert A_i \rVert \le L_{2,i}$, and the Cauchy--Schwarz inequality bounds every
average that enters $\bar r_n$ without recourse to maxima over $i$:
$n^{-1} \sum_i \lVert \hat g_i - \tilde g(\mathsf{U}_i) \rVert \, \lVert (\Pi k)(\mathsf{U}_i, \cdot) \rVert
\le \lVert \hat\theta_n - \theta_0 \rVert
\big( n^{-1} \sum_i L_{2,i}^2 \big)^{1/2}
\big( n^{-1} \sum_i \lVert (\Pi k)(\mathsf{U}_i, \cdot) \rVert^2 \big)^{1/2} = O_p(n^{-\alpha})$, the
first factor by Assumption~\ref{Assumption: parameter space}(iii), the second by
Chebyshev's inequality under
$\mathbb{E}[\sup_\theta \lVert \nabla^2_\theta \varepsilon(\theta) \rVert^2] < \infty$,
and the third by the same route, since
$\mathbb{E}\lVert (\Pi k)(\mathsf{U}, \cdot) \rVert^2 < \infty$ follows from boundedness of
$k$ and the gradient moment;
the displaced kernel terms obey
$\lVert k(\hat W_i, \cdot) - k(W_i, \cdot) \rVert \le c(Z_i)\, \lVert \hat\gamma_n - \gamma_0 \rVert$
pointwise in $i$, with Lipschitz constant
$c(Z_i) = \sup_\gamma \lVert \partial_\gamma P(Z_i; \gamma) \rVert$ satisfying
$\mathbb{E}[c(Z)^2] < \infty$ under Assumption~\ref{Assumption: residual regularity},
so their averages with $\tilde g(\mathsf{U}_i)$ or
$\hat g_i$ are $O_p(n^{-\alpha})$ by Cauchy--Schwarz likewise; and the common factors
$\hat{\mathbf g}^\top \hat{\mathbf g}/n - \tilde\Gamma$ and
$\hat b_n - \tilde b$, with $\hat b_n(\cdot) = n^{-1} \sum_\ell \hat g_\ell k(\hat W_\ell, \cdot)$,
are $O_p(n^{-1/2})$ by Chebyshev's inequality under the maintained second moments
and Slutsky's theorem, the latter
also picking up an $O_p(n^{-\alpha})$ displacement from $\hat W$.

Second, the mean value theorem applied to $\partial_\theta \hat m$ gives
$\lVert \partial_\theta \hat m(\bar\theta_n) - \partial_\theta \hat m(\theta_0) \rVert = O_p(n^{-\alpha})$,
because $\sup_{\theta} \lVert \partial^2_\theta \hat m(\theta) \rVert
\le n^{-1} \sum_{i=1}^n L_{2,i} \lVert (\hat{\boldsymbol{\Pi}} \hat k)(\hat W_i, \cdot) \rVert = O_p(1)$:
the feature norms obey
$\lVert (\hat{\boldsymbol{\Pi}} \hat k)(\hat W_i, \cdot) \rVert \le 1 + \lVert \hat g_i \rVert
\lVert (\hat{\mathbf{g}}^\top \hat{\mathbf{g}}/n)^{-1} \rVert \lVert \hat b_n \rVert$
with $\lVert \hat{\mathbf{g}}^\top \hat{\mathbf{g}}/n - \tilde\Gamma \rVert = O_p(n^{-1/2})$,
$\tilde\Gamma$ nonsingular, and $\lVert \hat b_n - \tilde b \rVert = O_p(n^{-1/2})$
by the law of large numbers, while $\lVert \hat g_i \rVert \le \lVert \tilde g(\mathsf{U}_i) \rVert
+ L_{2,i} \lVert \hat\theta_n - \theta_0 \rVert$; the Cauchy--Schwarz inequality,
with $\mathbb{E}[\sup_\theta \lVert \nabla^2_\theta \varepsilon(\theta) \rVert^2] < \infty$ and
$\mathbb{E}[\lVert \tilde g(\mathsf{U}) \rVert^2] < \infty$ under
Assumption~\ref{Assumption: residual regularity}, then bounds the average term
by term and gives $O_p(1)$. Combining the two estimates,
\[
\big\lVert \sqrt{n}\big(\hat m(\hat\theta_n) - \hat m(\theta_0)\big) \big\rVert
\le \sqrt{n}\, \lVert \hat\theta_n - \theta_0 \rVert
\big( \lVert \partial_\theta \hat m(\theta_0) \rVert + O_p(n^{-\alpha}) \big)
= O_p(n^{-\alpha}) + O_p(n^{1/2 - 2\alpha}) = o_p(1),
\]
since $n^{-\alpha} \to 0$ and $n^{1/2 - 2\alpha} \to 0$ for $\alpha > 1/4$.
Consequently,
\begin{align*}
n\left(\hat T_{np}^2(\hat\theta_n) - \hat T_{np}^2(\theta_0)\right)
&= 2\left\langle \sqrt{n}\,\hat m(\theta_0),\, \sqrt{n}\left(\hat m(\hat\theta_n) - \hat m(\theta_0)\right)\right\rangle + \left\lVert \sqrt{n}\left(\hat m(\hat\theta_n) - \hat m(\theta_0)\right)\right\rVert^2 \\
&= O_p(1)\cdot o_p(1) + o_p(1) = o_p(1),
\end{align*}
where the last step uses the Cauchy--Schwarz inequality. \qedhere
\end{proof}

\begin{lemma}\label{lem:sampleprojection}
Suppose Assumptions~\ref{ass:exclusion}--\ref{ass:relevance},
\ref{Assumption: parameter space}, and~\ref{Assumption: residual regularity} hold,
$H_0^{(2)}$ is true, the first stage is correctly specified,
$P(Z;\gamma_0) = P(Z)$ a.s., $\mathbb{E}[\varepsilon(\theta_0) \mid X, P] = 0$
a.s., and $\mathbb{E}[\tilde h_p(\mathsf{U}, \mathsf{U}')^2] < \infty$. Then,
with $\{\mu_j\}_{j\ge 1}$ the eigenvalues of the integral operator with kernel
$\tilde h_p$ on $L^2(\mathbb{P}_{\mathsf{U}})$, i.e.\ the solutions of
$\mu_j\psi_j(\mathsf{u}) = \int \tilde h_p(\mathsf{u}, \mathsf{u}')\psi_j(\mathsf{u}')\, d\mathbb{P}_{\mathsf{U}}(\mathsf{u}')$,
\[
n\hat T_{np}^2(\theta_0) = \frac{1}{n}\,\boldsymbol{\varepsilon}(\theta_0)^\top \hat{\boldsymbol{\Pi}}^\top \hat{\boldsymbol{K}} \hat{\boldsymbol{\Pi}}\,\boldsymbol{\varepsilon}(\theta_0) \xrightarrow{d} \sum_{j\ge 1} \mu_j A_j^2,
\]
with $\{A_j\}_{j\ge 1}$ independent standard normal.
\end{lemma}
\begin{proof}
First, replace the sample sandwich by the population features. Write
$\Phi_i = (\Pi k)(\mathsf{U}_i, \cdot)$, $\hat\Phi_i = (\hat{\boldsymbol{\Pi}}\hat k)(\hat W_i, \cdot)$,
so that $n \hat T_{np}^2(\theta_0) = \lVert \sqrt{n}\, \hat m(\theta_0) \rVert^2$ with
$\hat m(\theta_0) = n^{-1} \sum_{i=1}^n \varepsilon_i(\theta_0) \hat\Phi_i
= n^{-1} \sum_{i=1}^n \varepsilon_i(\theta_0) \Phi_i + \rho_n$ and
\[
n \hat T_{np}^2(\theta_0)
= \Big\lVert \frac{1}{\sqrt{n}} \sum_{i=1}^n \varepsilon_i(\theta_0) \Phi_i \Big\rVert^2 + o_p(1)
= \frac{1}{n} \sum_{i,j} \varepsilon_i(\theta_0) \varepsilon_j(\theta_0) \tilde k_p(\mathsf{u}_i, \mathsf{u}_j) + o_p(1),
\]
provided $\sqrt{n}\, \rho_n = o_p(1)$. To verify the rate, decompose the feature
deviation $\hat\Phi_i - \Phi_i = \delta_i^{\mathrm{id}} + \delta^{\mathrm{co}}$.
The idiosyncratic part collects
$(\hat g_i - \tilde g(\mathsf{U}_i))^\top (\hat{\mathbf g}^\top \hat{\mathbf g}/n)^{-1} \tilde b$
and $k(\hat W_i, \cdot) - k(W_i, \cdot)$ (the remaining cross products of an
idiosyncratic and a common deviation are bounded by Cauchy--Schwarz at
$O_p(n^{-1/2 - \alpha})$, the same order as the idiosyncratic piece, and are
absorbed in the final bound). The kernel
displacement $k(\hat W_i, \cdot) - k(W_i, \cdot)$ is a function of $\{Z_\ell\}$
alone, Lipschitz in $\hat\gamma_n - \gamma_0$ pointwise in $i$ under
Assumption~\ref{Assumption: residual regularity} and of size
$O_p(n^{-\alpha})$. For the score displacement, write
$\bar\mu(\theta) = \mu_0(X;\beta_0) + P(Z;\gamma)\,\Delta(X)$, a function of $Z$
alone, so that $\varepsilon(\theta) = Y - \bar\mu(\theta)$ and
$\nabla^2_\theta \varepsilon(\theta) = -\nabla^2_\theta \bar\mu(\theta)$ never involves
$Y$; the mean value theorem gives
$\hat g_i - \tilde g(\mathsf{U}_i) = -\nabla^2_\theta \bar\mu(\bar\theta_i)(Z_i)(\hat\theta_n - \theta_0)$
with $\lVert \nabla^2_\theta \bar\mu(\bar\theta_i)(Z_i) \rVert \le L_{2,i}
:= \sup_{\theta \in \Theta} \lVert \nabla^2_\theta \varepsilon_i(\theta) \rVert$. The
contribution splits exactly:
\[
n^{-1} \sum_i \varepsilon_i(\theta_0) \big( \hat g_i - \tilde g(\mathsf{U}_i) \big)^\top \tilde\Gamma^{-1} \tilde b
= - (\hat\theta_n - \theta_0)^\top \Big[ n^{-1} \sum_i \varepsilon_i(\theta_0) \nabla^2_\theta \bar\mu(\theta_0)(Z_i) \Big] \tilde\Gamma^{-1} \tilde b
   + O_p(n^{-2\alpha}),
\]
where the first term is $O_p(n^{-1/2 - \alpha})$: the bracket has mean zero,
$\mathbb{E}[\varepsilon(\theta_0) \nabla^2_\theta \bar\mu(\theta_0)] =
\mathbb{E}\big[ \mathbb{E}[\varepsilon(\theta_0) \mid Z]\, \nabla^2_\theta \bar\mu(\theta_0)(Z) \big] = \boldsymbol{0}$
by Lemma~\ref{lem:centering}, $\nabla^2_\theta \bar\mu(\theta_0)(Z)$ being a function of
$Z$, its variance is finite under Assumption~\ref{Assumption: residual regularity},
and the central limit theorem applies. The remainder is bounded by Cauchy--Schwarz:
writing $L_{3,i} = \sup_{\theta \in \Theta} \lVert \nabla^3_\theta \varepsilon_i(\theta) \rVert$,
the displacement obeys
$\lVert \nabla^2_\theta \bar\mu(\bar\theta_i)(Z_i) - \nabla^2_\theta \bar\mu(\theta_0)(Z_i) \rVert
\le L_{3,i} \lVert \hat\theta_n - \theta_0 \rVert$, so the remainder is at most
$\lVert \hat\theta_n - \theta_0 \rVert^2
\big( n^{-1} \sum_i \varepsilon_i(\theta_0)^2 \big)^{1/2}
\big( n^{-1} \sum_i L_{3,i}^2 \big)^{1/2} = O_p(n^{-2\alpha})$. The common part
collects $\tilde g(\mathsf{U}_i)^\top \big( (\hat{\mathbf g}^\top \hat{\mathbf g}/n)^{-1} -
\tilde\Gamma^{-1} \big) \tilde b$ and $\tilde g(\mathsf{U}_i)^\top \tilde\Gamma^{-1}
(\hat b_n - \tilde b)$, with the matrix deviation $O_p(n^{-1/2})$ by the law of
large numbers, $\lVert \hat b_n - \tilde b \rVert = O_p(n^{-1/2})$ by the law of
large numbers for Bochner integrals, and both refined by Slutsky's theorem,
$\tilde\Gamma$ being nonsingular. The kernel-displacement contribution is handled by linearizing in
$\hat\gamma_n - \gamma_0$. For kernels twice continuously differentiable in
their first argument (in particular the Gaussian kernel used throughout),
$k(\hat W_i, \cdot) - k(W_i, \cdot) = \partial_\gamma k(W_i, \cdot)(Z_i)
(\hat\gamma_n - \gamma_0)$ up to a second-order term that is bounded, uniformly
in $i$, by Cauchy--Schwarz at $O_p(n^{-2\alpha})$ under
Assumption~\ref{Assumption: residual regularity}. Hence
\[
n^{-1} \sum_i \varepsilon_i(\theta_0) \big[ k(\hat W_i, \cdot) - k(W_i, \cdot) \big]
= \Big[ n^{-1} \sum_i \varepsilon_i(\theta_0) \, \partial_\gamma k(W_i, \cdot)(Z_i) \Big] (\hat\gamma_n - \gamma_0)
   + O_p(n^{-2\alpha}),
\]
and the bracket has mean zero exactly,
$\mathbb{E}[\varepsilon(\theta_0) \, \partial_\gamma k(W, \cdot)(Z)] =
\mathbb{E}\big[ \mathbb{E}[\varepsilon(\theta_0) \mid Z]\, \partial_\gamma k(W, \cdot)(Z) \big] = \boldsymbol{0}$
by Lemma~\ref{lem:centering} -- the summand being a function of $Z$ -- and is
$O_p(n^{-1/2})$ by the central limit theorem; conditional Chebyshev on
$n^{-2} \sum_i \mathbb{E}[\varepsilon(\theta_0)^2 \mid Z_i]\,
\lVert \partial_\gamma k(W_i, \cdot)(Z_i) \rVert^2 = O_p(n^{-1})$ gives the same
rate conditionally. The contribution is therefore
$O_p(n^{-1/2 - \alpha} + n^{-2\alpha})$.
The common contribution factors as
$\big( n^{-1} \sum_i \varepsilon_i(\theta_0)\, \tilde g(\mathsf{U}_i)^\top \big) \, O_p(n^{-1/2})$,
whose bracket has mean zero by Lemma~\ref{lem:centering}
($\mathbb{E}[\varepsilon(\theta_0)\tilde g(\mathsf{U})^\top] =
\mathbb{E}[\mathbb{E}[\varepsilon(\theta_0) \mid Z] \tilde g(\mathsf{U})^\top] = \boldsymbol{0}$)
and is $O_p(n^{-1/2})$ by the central limit theorem, giving $O_p(n^{-1})$ in
total. Hence $\lVert \rho_n \rVert = O_p(n^{-1/2 - \alpha} + n^{-2\alpha})$ and
$\sqrt{n}\, \rho_n = O_p(n^{-\alpha} + n^{1/2 - 2\alpha}) = o_p(1)$ for
$\alpha > 1/4$.
Second, $\tilde h_p$ is degenerate under the maintained hypotheses: by
Lemma~\ref{lem:centering},
\[
\mathbb{E}\big[\tilde h_p(\mathsf{U}, \mathsf{U}') \mid \mathsf{U}'\big]
= \varepsilon(\theta_0)\,'\; \mathbb{E}\big[\varepsilon(\theta_0)\, \tilde k_p(\mathsf{U}, \mathsf{U}') \mid \mathsf{U}'\big] = 0
\quad \text{a.s.},
\]
and $\mathbb{E}[\tilde h_p(\mathsf{U}, \mathsf{U}')^2] < \infty$ is part of the hypothesis. Third, apply the
spectral expansion for degenerate U-statistics (Serfling, 1980, Theorem 5.5.2):
\[
n \cdot \frac{2}{n(n-1)} \sum_{i < j} \tilde h_p(\mathsf{U}_i, \mathsf{U}_j)
\; \xrightarrow{d} \; \sum_{j \ge 1} \mu_j \big( A_j^2 - 1 \big),
\]
a result stated by Serfling for U-statistics and transposed here to
V-statistics, the diagonal contributing $O_p(n^{-1})$. The V-statistic differs
from $n$ times the U-statistic only by the diagonal:
\[
\frac{1}{n} \sum_{i,j} \tilde h_p(\mathsf{U}_i, \mathsf{U}_j)
= \frac{1}{n} \sum_i \tilde h_p(\mathsf{U}_i, \mathsf{U}_i)
+ \frac{n-1}{n} \cdot n \cdot \frac{2}{n(n-1)} \sum_{i < j} \tilde h_p(\mathsf{U}_i, \mathsf{U}_j),
\]
and the diagonal average converges in probability to
$\mathbb{E}[\tilde h_p(\mathsf{U}, \mathsf{U})] = \sum_{j \ge 1} \mu_j$ by the law of large
numbers, this being the trace of the positive trace-class operator with kernel
$\tilde h_p$, equal to the integral of the kernel over the diagonal; finiteness
holds since $\tilde h_p(\mathsf{U}, \mathsf{U}) = \varepsilon(\theta_0)^2 \lVert (\Pi k)(\mathsf{U}, \cdot) \rVert^2$
and $\mathbb{E}[\varepsilon(\theta_0)^2 \lVert (\Pi k)(\mathsf{U}, \cdot) \rVert^2]
\le \big( \mathbb{E}[\varepsilon(\theta_0)^4] \big)^{1/2}
\big( \mathbb{E}[\tilde k_p(\mathsf{U}, \mathsf{U})^2] \big)^{1/2} < \infty$ under the maintained
moments, so the trace is finite and $\{\mu_j\} \in \ell^1$. Slutsky's
theorem combines the two displays and gives the claim. \qedhere
\end{proof}

\begin{proof}[Proof of Theorem \ref{thm:asymptotics}]
(i) The hypothesis $\lVert \sqrt{n}\,\hat m(\theta_0) \rVert = O_p(1)$ of
Lemma~\ref{lem:invariance} holds under the maintained null: by the reduction in
the proof of Lemma~\ref{lem:sampleprojection},
$\sqrt{n}\,\hat m(\theta_0) = n^{-1/2} \sum_i \varepsilon_i(\theta_0) (\Pi k)(\mathsf{U}_i, \cdot) + o_p(1)$,
whose first term has mean zero in $\mathcal{H}_k$ by Lemma~\ref{lem:centering}
and conditional variance $\mathbb{E}[\varepsilon(\theta_0)^2 \tilde k_p(\mathsf{U}, \mathsf{U})] < \infty$,
the cross terms vanishing by the same lemma. Then, by
Lemmas~\ref{lem:invariance} and~\ref{lem:sampleprojection},
$n\hat T_{np}^2(\hat\theta_n) = n\hat T_{np}^2(\theta_0) + o_p(1) \to_d \sum_{j\ge 1}\mu_j A_j^2$, where $\{\mu_j\}$ are the eigenvalues of $\tilde h_p$.

(ii) Under a fixed alternative, $\hat\theta_n \to_p \theta^\ast$ and, by the same replacement argument as in the proof of Lemma~\ref{lem:sampleprojection} evaluated at $\theta^\ast$, whose population objects are well defined under the hypotheses of part (ii), $\hat T_{np}^2(\theta^\ast) \to_p T_p^2(\theta^\ast)$, where
\[
T_p^2(\theta^\ast) = \mathbb{E}\big[\varepsilon(\theta^\ast)\, \tilde k_{p,\theta^\ast}(\mathsf{U}, \mathsf{U}')\, \varepsilon(\theta^\ast)'\big] = \left\lVert \mathbb{E}\left[\varepsilon(\theta^\ast) (\Pi_{\theta^\ast} k)(\mathsf{U}, \cdot)\right]\right\rVert^2 \ge 0,
\]
the equality following from the four-term form \eqref{eq:bivariate-kp} and Fubini's theorem, exactly as in the proof of Lemma~\ref{lem:orthogonal}. By \eqref{eq:Pi-section} and \eqref{eq:raw-objects},
\[
\mathbb{E}\big[\varepsilon(\theta^\ast) (\Pi_{\theta^\ast} k)(\mathsf{U}, \cdot)\big]
= \mathbb{E}\big[\varepsilon(\theta^\ast) k(W, \cdot)\big] - \mathbb{E}\big[\varepsilon(\theta^\ast) \tilde g_{\theta^\ast}(\mathsf{U})^\top\big] \tilde\Gamma_{\theta^\ast}^{-1} \tilde b_{\theta^\ast}(\cdot),
\]
and, by the tower property applied to the first term and the linearity of the reproducing map, this vanishes if and only if
$\mathbb{E}\big[ \big( m_{\theta^\ast}(W) - v^\top \tilde g_{\theta^\ast}(\mathsf{U}) \big) k(W, \cdot) \big] = \boldsymbol{0}$
for $m_{\theta^\ast}(w) = \mathbb{E}[\varepsilon(\theta^\ast) \mid W = w]$ and
$v^\top = \mathbb{E}[\varepsilon(\theta^\ast) \tilde g_{\theta^\ast}(\mathsf{U})^\top] \tilde\Gamma_{\theta^\ast}^{-1}$;
by the injectivity argument of Theorem~\ref{Theorem: Residual Embedding} (ISPD),
applied to the $W$-measurable function $m_{\theta^\ast}(W) - v^\top a(W)$ with
$a(W) = \mathbb{E}[\tilde g_{\theta^\ast}(\mathsf{U}) \mid W]$, this holds if and only if
$m_{\theta^\ast}(W) = v^\top a(W)$ a.s. The conditional form upgrades to the raw
one: $v^\top a(W) = \mathbb{E}[v^\top \tilde g_{\theta^\ast}(\mathsf{U}) \mid W]$ by the tower
property, and $\mathbb{E}[m_{\theta^\ast}(W) a(W)^\top] = \mathbb{E}[\varepsilon(\theta^\ast) a(W)^\top]
= v^\top \tilde\Gamma_{\theta^\ast}$, so $m_{\theta^\ast}(W) = v^\top a(W)$ a.s. forces
$v^\top (\tilde\Gamma_{\theta^\ast} - C) = \boldsymbol{0}$ with
$C = \mathbb{E}[a(W) a(W)^\top]$, i.e. $v^\top \Omega = \boldsymbol{0}$ for
$\Omega = \mathbb{E}[\operatorname{Var}(\tilde g_{\theta^\ast}(\mathsf{U}) \mid W)] \succeq 0$; positive
semidefiniteness gives $\operatorname{Var}(v^\top \tilde g_{\theta^\ast}(\mathsf{U}) \mid W) =
v^\top \Omega v = 0$ almost surely, hence $v^\top a(W) = v^\top \tilde g_{\theta^\ast}(\mathsf{U})$
a.s. and $m_{\theta^\ast}(W) = v^\top \tilde g_{\theta^\ast}(\mathsf{U})$ a.s., which the
hypothesis of part (ii) excludes. Hence $T_p^2(\theta^\ast) > 0$. Then
\[
n\hat T_{np}^2(\hat\theta_n) = n \hat T_{np}^2(\theta^\ast) + O_p(n^{1-\alpha})
= n\, T_p^2(\theta^\ast) + O_p(\sqrt{n}) + O_p(n^{1-\alpha}) \to_p \infty :
\]
the first equality holds because
$\lvert \hat T_{np}^2(\hat\theta_n) - \hat T_{np}^2(\theta^\ast) \rvert \le
\lVert \hat m(\hat\theta_n) - \hat m(\theta^\ast) \rVert
\big( \lVert \hat m(\hat\theta_n) \rVert + \lVert \hat m(\theta^\ast) \rVert \big)$
with $\lVert \hat m(\theta) \rVert = O_p(1)$ near $\theta^\ast$ and
$\lVert \hat m(\hat\theta_n) - \hat m(\theta^\ast) \rVert
\le \lVert \hat\theta_n - \theta^\ast \rVert
\big( n^{-1} \sum_i \lvert \varepsilon_i(\bar\theta_i) \rvert^2 \big)^{1/2}
\big( n^{-1} \sum_i \lVert \hat\Phi_i \rVert^2 \big)^{1/2} = O_p(\lVert \hat\theta_n - \theta^\ast \rVert) = O_p(n^{-\alpha})$
by the mean value theorem and the Cauchy--Schwarz inequality under
Assumption~\ref{Assumption: residual regularity}; and for the $O_p(\sqrt{n})$
term, the Hoeffding projection gives
\[
n\big( \hat T_{np}^2(\theta^\ast) - T_p^2(\theta^\ast) \big)
= 2 \sum_{i=1}^n \big( h_1(\mathsf{U}_i) - T_p^2(\theta^\ast) \big) + O_p(1) = O_p(\sqrt{n}),
\qquad
h_1(\mathsf{u}) = \mathbb{E}\big[\tilde h_{p,\theta^\ast}(\mathsf{u}, \mathsf{U}')\big],
\]
the bound using $\mathbb{E}[\tilde h_{p,\theta^\ast}(\mathsf{U}, \mathsf{U}')^2] < \infty$
(Serfling, 1980, Theorem 5.5.1A, stated for U-statistics and transposed to
V-statistics by adding the diagonal, which is $O_p(n^{-1})$). The linear term
genuinely fluctuates at order $\sqrt{n}$: were the kernel degenerate,
$\mathbb{E}[\tilde h_{p,\theta^\ast}(\mathsf{U}, \mathsf{U}') \mid \mathsf{U}'] \equiv 0$, the tower
property would give $T_p^2(\theta^\ast) = \mathbb{E}\big[\varepsilon(\theta^\ast)'\;
\mathbb{E}[\tilde h_{p,\theta^\ast}(\mathsf{U}, \mathsf{U}') \mid \mathsf{U}']\big] = 0$, contradicting
$T_p^2(\theta^\ast) > 0$. Hence $\Pr(n\hat T_{np}^2(\hat\theta_n) > t) \to 1$ for
every fixed $t > 0$.

(iii) Decompose the residual under $H_{1n}$ as $\varepsilon(\theta_0) = \eta(\theta_0) + R(W)/\sqrt{n}$ with $\mathbb{E}[\eta(\theta_0)\mid W] = 0$ and
\[
\mathbb{E}[\varepsilon(\theta_0) \mid Z] = \mathbb{E}[R(W) \mid Z]/\sqrt{n} + o(n^{-1/2}),
\]
in $L^2(\mathbb{P})$; a first-order $Z$-centering that is automatic when the perturbation enters
additively, $\varepsilon(\theta_0) = \eta_0 + R(W)/\sqrt{n}$ with
$\mathbb{E}[\eta_0 \mid Z] = 0$, since then
$\mathbb{E}[\varepsilon(\theta_0) \mid Z] = \mathbb{E}[R(W) \mid Z]/\sqrt{n}$; it
gives $\mathbb{E}[\eta(\theta_0) \mid Z] = o(n^{-1/2})$ and is satisfied for the
$W$-measurable perturbations of Section~\ref{sec:sim-design}; the reduction in the proof of
Lemma~\ref{lem:sampleprojection} then applies with its idiosyncratic centering
holding at the perturbed residual by this clause, and gives,
\[
n\hat T_{np}^2(\theta_0) = \frac{1}{n}\sum_{i=1}^n\sum_{j=1}^n \left(\eta_i + \frac{R_i}{\sqrt{n}}\right) \tilde k_p(\mathsf{u}_i,\mathsf{u}_j) \left(\eta_j + \frac{R_j}{\sqrt{n}}\right) = \left\lVert G_{n,p} + \hat m_{R,p}\right\rVert^2,
\]
where $G_{n,p} = n^{-1/2}\sum_{i=1}^n \eta_i (\Pi k)(\mathsf{U}_i,\cdot)$ and $\hat m_{R,p} = n^{-1}\sum_{i=1}^n R_i (\Pi k)(\mathsf{U}_i,\cdot)$.

The $\mathcal{H}_k$-valued empirical process $G_{n,p}$ converges weakly to the centered Gaussian process $G_p$ with covariance operator $C_{\tilde p} = \mathbb{E}[\eta(\theta_0)^2 (\Pi k)(\mathsf{U},\cdot)\otimes(\Pi k)(\mathsf{U},\cdot)]$: finite-dimensional projections are i.i.d.\ sums, to which the Lindeberg--L\'evy central limit theorem applies under Assumption~\ref{Assumption: residual regularity}, and tightness holds because $\mathbb{E}\lVert G_{n,p}\rVert^2 = \mathbb{E}[\eta(\theta_0)^2\, \tilde k_p(\mathsf{U}, \mathsf{U})] < \infty$, the covariance operator being constant in $n$.

The nonzero eigenvalues of $C_{\tilde p}$ coincide with those of the integral operator with kernel $\eta(\theta_0)\,\eta(\theta_0)'\,\tilde k_p$ on $L^2(\mathbb{P}_{\mathsf{U}})$: define $B : L^2(\mathbb{P}_{\mathsf{U}}) \to \mathcal{H}_k$ by $Bf = \int f(\mathsf{u})\, \eta(\theta_0)(\mathsf{u})\, (\Pi k)(\mathsf{u},\cdot)\, d\mathbb{P}_{\mathsf{U}}(\mathsf{u})$, a bounded operator because $\mathbb{E}[\eta(\theta_0)^2\, \tilde k_p(\mathsf{U},\mathsf{U})] < \infty$; then $C_{\tilde p} = BB^\ast$, while $B^\ast B$ is the integral operator with kernel $\eta(\theta_0)\,\eta(\theta_0)'\,\tilde k_p$, and the nonzero spectra of $BB^\ast$ and $B^\ast B$ coincide, being the squared singular values of the compact operator $B$ (spectral theory of compact operators). The Karhunen--Lo\`eve expansion of $G_p$ therefore gives $\lVert G_p\rVert^2 = \sum_{j\ge 1}\mu_j A_j^2$ with $\{\mu_j\}$ these eigenvalues. Moreover,
\[
\varepsilon(\theta_0)\,\varepsilon(\theta_0)' - \eta(\theta_0)\,\eta(\theta_0)'
= \frac{\eta(\theta_0)\, R(W)' + R(W)\, \eta(\theta_0)'}{\sqrt{n}} + \frac{R(W)\, R(W)'}{n},
\]
so the integral operators with kernels $\tilde h_p$ and $\eta(\theta_0)\,\eta(\theta_0)'\,\tilde k_p$ differ in Hilbert--Schmidt norm by $O(n^{-1/2})$ under Assumption~\ref{Assumption: residual regularity} and $\mathbb{E}[R^4] < \infty$; Weyl-type perturbation bounds for Hilbert--Schmidt operators then imply that the eigenvalues of the $\eta\eta'\tilde k_p$-operator along $H_{1n}$ drift from those of part (i) by at most $O(n^{-1/2})$ in each coordinate, so the same symbols $\{\mu_j\}$ may be used at the precision of the stated limit.

By the law of large numbers for Bochner integrals, which uses
$\mathbb{E}[R(W)^2 \lVert (\Pi k)(\mathsf{U}, \cdot) \rVert^2] < \infty$ under
$\mathbb{E}[R^4] < \infty$ and Assumption~\ref{Assumption: residual regularity},
$\hat m_{R,p} \to_p m_{R,p} := \mathbb{E}[R(W) (\Pi k)(\mathsf{U}, \cdot)]$ in
$\mathcal{H}_k$, and $\lVert m_{R,p}\rVert^2 = \mathbb{E}[R(W)\, \tilde k_p(\mathsf{U}, \mathsf{U}')\, R(W)']$
by the reproducing property. For the cross term, with
$h_1^R(\mathsf{u}) = \mathbb{E}[R(W')\, \tilde k_p(\mathsf{U}', \mathsf{u})] = \big\langle (\Pi k)(\mathsf{u}, \cdot), m_{R,p} \big\rangle$,
\[
\left\langle G_{n,p}, \hat m_{R,p}\right\rangle
= \frac{1}{\sqrt{n}} \sum_{i=1}^n \eta_i\, h_1^R(\mathsf{U}_i) + O_p(n^{-1/2})
\; \xrightarrow{d} \; B_p :
\]
writing $d_i = n^{-1}\sum_j R(W_j)\, \tilde k_p(\mathsf{U}_j, \mathsf{u}_i) - h_1^R(\mathsf{u}_i)$ for the
deviation of the inner average, $\mathbb{E}[d_i^2 \mid \mathsf{U}_i] = O(n^{-1})$ by the
i.i.d.\ structure of the summands over $j$ and $\mathbb{E}[\eta_i\eta_\ell\, d_i d_\ell] = o(n^{-1})$
for $i \ne \ell$ (condition on $\{Z_\ell\}$ and use
$\mathbb{E}[\eta_i \mid Z_i] = o(n^{-1/2})$ under the displayed $Z$-centering,
the $d_i$ being functions of $\{Z_\ell\}$), while the conditional variance is
$n^{-1} \sum_i \mathbb{E}[\eta_i^2 \mid Z_i]\, d_i^2 = O_p(n^{-1})$, so the
remainder is $O_p(n^{-1/2})$ with mean $o(n^{-1/2})$; the leading term is an
i.i.d.\ sum of zero-mean variables with finite variance, and the
Lindeberg--L\'evy central limit theorem gives $B_p \sim \mathcal{N}(0, S_p^2)$ with variance
\[
S_p^2 = \mathrm{Var}\left(\eta(\theta_0)\, \mathbb{E}\big[ R(W')\, \tilde k_p(\mathsf{U}', \mathsf{U}) \big]\right),
\]
which is the variance of $B_p$ referred to in the statement, and
$B_p = \langle G_p, m_{R,p}\rangle$ in distribution. Finally, the estimation
effect is negligible under $H_{1n}$ by Lemma~\ref{lem:invariance}, whose
hypothesis holds because $\sqrt{n}\,\hat m(\theta_0) = G_{n,p} + \hat m_{R,p}
= O_p(1)$. The continuous mapping theorem
applied to $(G_{n,p}, \hat m_{R,p}) \mapsto \lVert G_{n,p} + \hat m_{R,p}\rVert^2$ yields
\[
n\hat T_{np}^2(\hat\theta_n) \xrightarrow{d} \sum_{j\ge 1}\mu_j A_j^2 + 2B_p + \mathbb{E}\left[R(W)\, \tilde k_p(\mathsf{U}, \mathsf{U}')\, R(W)'\right],
\]
with $B_p = \langle G_p, m_{R,p}\rangle \sim \mathcal{N}(0, S_p^2)$. This is the claimed limit. \qedhere
\end{proof}

\begin{proof}[Proof of Theorem \ref{Theorem: Bootstrap Consistency}]
Conditional on the data, the bootstrap statistic equals $\lVert G_n^\ast\rVert^2$ with $G_n^\ast = n^{-1/2}\sum_{i=1}^n \varepsilon_i(\hat\theta_n)V_i\,(\hat{\boldsymbol{\Pi}}\hat k)(\hat W_i,\cdot)$. Since $\mathbb{E}[V_i] = 0$ and $\mathbb{E}[V_i^2] = 1$, the multiplier central limit theorem in Hilbert space applies conditionally on the data once two conditions are verified: the conditional covariance operator converges,
\begin{equation}\label{eq:cov-lln}
\frac{1}{n}\sum_{i=1}^n \varepsilon_i(\hat\theta_n)^2\,(\hat{\boldsymbol{\Pi}}\hat k)(\hat W_i,\cdot)\otimes(\hat{\boldsymbol{\Pi}}\hat k)(\hat W_i,\cdot) \;\to_p\; C_{\tilde p} := \mathbb{E}\left[\varepsilon(\theta_0)^2\,(\Pi k)(\mathsf{U},\cdot)\otimes(\Pi k)(\mathsf{U},\cdot)\right],
\end{equation}
and no observation dominates,
\begin{equation}\label{eq:max-leverage}
\max_{1 \le i \le n} \frac{\varepsilon_i(\hat\theta_n)^2\, \lVert (\hat{\boldsymbol{\Pi}}\hat k)(\hat W_i, \cdot) \rVert^2}{n} \;\to_p\; 0 .
\end{equation}
It is this covariance-operator convergence, a law of large numbers with an
$O_p(n^{-\alpha})$ perturbation, that justifies evaluating the bootstrap at the
fixed $\hat\theta_n$; the invariance argument of Lemma~\ref{lem:invariance} is
not consumed here. To verify \eqref{eq:cov-lln}, replace the estimated residuals
and features by their population counterparts:
$\varepsilon_i(\hat\theta_n)^2 - \varepsilon_i(\theta_0)^2 = O_p(n^{-\alpha})$
uniformly in $i$ by the mean value theorem and
Assumption~\ref{Assumption: residual regularity}, and
$\hat\Phi_i - \Phi_i$ with $\Phi_i = (\Pi k)(\mathsf{U}_i, \cdot)$ admits the
idiosyncratic-plus-common decomposition of the proof of
Lemma~\ref{lem:sampleprojection}, with sizes $O_p(n^{-\alpha})$ and
$O_p(n^{-1/2})$ respectively; since $\lVert \Phi_i \rVert = O_p(1)$ uniformly
under the maintained moments, the Cauchy--Schwarz inequality bounds the resulting
drift of \eqref{eq:cov-lln} by $O_p(n^{-\alpha} + n^{-1/2})$, and the population
average converges to $C_{\tilde p}$ by the law of large numbers for Bochner
integrals, which uses
$\mathbb{E}[\varepsilon(\theta_0)^2 \lVert (\Pi k)(\mathsf{U}, \cdot) \rVert^2] < \infty$.
For the norm bounds, the triangle inequality followed by squaring,
$(a + b)^2 \le 2a^2 + 2b^2$, and the norm computations in the proof of
Lemma~\ref{lem:orthogonal} (namely $\lVert k(w, \cdot) \rVert^2 = k(w, w)$ and
$\lVert \tilde g(\mathsf{U})^\top \tilde\Gamma^{-1} \tilde b(\cdot) \rVert^2 =
\tilde g(\mathsf{U})^\top \tilde\Gamma^{-1} \tilde M \tilde\Gamma^{-1} \tilde g(\mathsf{U})$) give
\[
\lVert (\Pi k)(\mathsf{U}, \cdot) \rVert^2
\le 2 + 2\, \mathrm{tr}(\tilde\Gamma^{-1} \tilde M)\; \tilde g(\mathsf{U})^\top \tilde\Gamma^{-1} \tilde g(\mathsf{U}),
\]
where, with $k \ge 0$ and the elementary bound $k\, x^\top y \le \tfrac{1}{2}\, k(x^\top x + y^\top y)$ (from $2\sqrt{ab} \le a + b$),
\[
\mathrm{tr}(\tilde\Gamma^{-1} \tilde M)
= \mathbb{E}\big[ k(W, W')\, \tilde g(\mathsf{U}')^\top \tilde\Gamma^{-1} \tilde g(\mathsf{U}) \big]
\le \mathbb{E}\big[ \tilde g(\mathsf{U})^\top \tilde\Gamma^{-1} \tilde g(\mathsf{U}) \big]
= \mathrm{tr}(\tilde\Gamma^{-1} \tilde\Gamma) = 2K + L,
\]
using $k(w, w') \le 1$, which holds for the Gaussian kernel since the exponent
$-\lVert w - w'\rVert_2^2/\sigma \le 0$ (for a general bounded kernel, the bound becomes
$\mathrm{tr}(\tilde\Gamma^{-1} \tilde M) \le (2K+L) \sup_{w, w'} k(w, w')$). Hence
$\lVert (\Pi k)(\mathsf{U}, \cdot) \rVert^2 \le 2 + 2(2K+L)\, \tilde g(\mathsf{U})^\top \tilde\Gamma^{-1} \tilde g(\mathsf{U})$,
and $\mathbb{E}[\varepsilon(\theta_0)^4 \lVert (\Pi k)(\mathsf{U}, \cdot) \rVert^4] < \infty$
under Assumption~\ref{Assumption: residual regularity}, which also yields
\eqref{eq:max-leverage} by Markov's inequality applied to the $n$ summands.

With \eqref{eq:cov-lln} and \eqref{eq:max-leverage} in hand, the multiplier central limit theorem in Hilbert space yields, conditionally on the data and hence unconditionally, $G_n^\ast \Rightarrow G^\ast$, the centered Gaussian process with covariance operator $C_{\tilde p}$, so that $n\hat T_{np}^{2\ast}(\hat\theta_n) = \lVert G_n^\ast\rVert^2 \xrightarrow{d} \lVert G^\ast\rVert^2$ conditionally on the data. The nonzero eigenvalues of $C_{\tilde p}$ are the eigenvalues $\{\mu_j\}$ of $\tilde h_p$: $C_{\tilde p} = BB^\ast$ with $Bf = \int f(\mathsf{u})\, \varepsilon(\theta_0)(\mathsf{u}) (\Pi k)(\mathsf{u}, \cdot)\, d\mathbb{P}_{\mathsf{U}}(\mathsf{u})$ and $B^\ast B$ has kernel $\tilde h_p$, the same identification as in the proof of Theorem~\ref{thm:asymptotics}(iii). Hence, by the Karhunen--Lo\`eve expansion, conditionally on the data,
\[
n\hat T_{np}^{2\ast}(\hat\theta_n) \xrightarrow{d} \lVert G^\ast\rVert^2 = \sum_{j\ge 1}\mu_j A_j^2,
\]
with $\{A_j\}$ independent standard normal. This is the null limit of $n\hat T_{np}^2(\hat\theta_n)$ by Theorem~\ref{thm:asymptotics}(i). The conditional distribution functions therefore converge uniformly in probability to the distribution function $F$ of $\sum_{j\ge 1}\mu_j A_j^2$, which is continuous whenever at least one $\mu_j > 0$ (a weighted sum of
independent $\chi^2_1$ variables); quantile convergence then gives $c_{B,\alpha}^{\ast} \to_p F^{-1}(1-\alpha) = c_\alpha$. The technical tightness and uniform conditional-convergence steps follow from Theorem~4.4 of \citet{escanciano2024gaussian}, whose operative conditions hold here: the max-leverage condition \eqref{eq:max-leverage}, delivered by the fourth moments of Assumption~\ref{Assumption: residual regularity} and the uniform feature bound, and the conditional covariance-operator convergence \eqref{eq:cov-lln}.

Finally, the test is consistent against the alternatives of
Theorem~\ref{thm:asymptotics}(ii): there $n\hat T_{np}^2(\hat\theta_n) \to_p \infty$,
while the bootstrap quantile is bounded in probability, since conditional on the
data $\mathbb{E}[\lVert G_n^\ast \rVert^2 \mid \text{data}] = n^{-1}\sum_i
\varepsilon_i(\hat\theta_n)^2 \lVert \hat\Phi_i \rVert^2 = O_p(1)$ and Markov's
inequality applies to each bootstrap draw. Hence
$\Pr\big( n\hat T_{np}^2(\hat\theta_n) > c_{B,\alpha}^{\ast} \big) \to 1$.
\qedhere
\end{proof}

\newpage
\section{Additional Simulation Results}\label{supp:additional-sims}
The main text reports simulation results for the probit design (Design~N):
Table~\ref{main:tab:size} reports empirical size, and
Table~\ref{main:tab:power_main} compares power at the median  kernel scale.
The main-text complements these results with four scale-by-scale
power tables for the omnibus and directional statistics under $c_1=0$ and
$c_1=3$ (Tables~\ref{main:tab:power_N_c0}--\ref{main:tab:power_di_N_c3}).
This supplement reports the corresponding results for the linear design
(Design~L): Section~\ref{supp:linear} covers the omnibus statistics $T_{\sigma}^{XP}$, while
Section~\ref{supp:di-linear} covers both the directional statistics $T_{\sigma}^{P}$ and the
polynomial Wald benchmarks: $\Qpoly{2}$ and $\Qpoly{5}$ proposed by HSU.

\subsection{Omnibus statistics: linear design}\label{supp:linear}

Table~\ref{supp:tab:size_main_L} reports the empirical size of the omnibus
statistics in the linear design, with rows given by the nominal level $\alpha$
and the sample size $n$, and columns given by the covariate scale $\sigma_X$
crossed with the kernel scale $\sigma$. Rejection frequencies are close to nominal
at $\sigma_X \in \{0.5, 1\}$. At $\sigma_X = 2$, the fixed scales
$\sigma \in \{1,2\}$ are conservative (roughly $0.02$--$0.03$ at the $5\%$
level), whereas the median rule stays between $0.036$ and $0.049$.

With $c_1 = 0$, the median rule has a modest advantage at $\sigma_X = 0.5$. At
$\sigma_X = 1$, the median scale is close to $2$ and thus their performance
tends to be slightly better than that of $\sigma = 1$ for $S_1$--$S_3$,
although this pattern is not uniform across loading magnitudes and sample
sizes. At $\sigma_X = 2$, the median rule remains broadly competitive and is
often modestly better for $S_1$--$S_3$. For $S_4$ the picture reverses, and the
frequency trade-off of Remark~\ref{rem:kernel-choice} accounts for it.

With $c_1 = 3$, no Gaussian scale uniformly dominates, and the ranking depends
on the covariate scale and the sample size. At
$\sigma_X = 0.5$, the median scale is the more powerful choice in nearly every
cell. At $\sigma_X = 1$, the realised median scale is $2.56$, so the median
rule behaves much like $\sigma = 2$; the fixed Gaussian scale $\sigma = 1$ is
nevertheless the better choice in most cells. At $\sigma_X = 2$, the median scale remains competitive at
$n = 500$ and $1{,}000$ but is overtaken by the fixed scales at
$n = 2{,}000$. 

\begin{table}[!tbp]
\caption{Empirical size of the omnibus statistic $T^{X\!P}_\sigma$ across covariate scales.}
\label{supp:tab:size_main_L}
\centering
\singlespacing
\setlength{\tabcolsep}{3pt}
\footnotesize
\begin{threeparttable}
\begin{tabular*}{\textwidth}{@{\extracolsep{\fill}}ll*{9}{r}@{}}
\toprule
& & \multicolumn{3}{c}{$\sigma_X=0.5$} & \multicolumn{3}{c}{$\sigma_X=1$} & \multicolumn{3}{c}{$\sigma_X=2$}\\
\cmidrule(lr){3-5}\cmidrule(lr){6-8}\cmidrule(lr){9-11}
$\alpha$ & $n$ & $\mathrm{med}$ & $1$ & $2$ & $\mathrm{med}$ & $1$ & $2$ & $\mathrm{med}$ & $1$ & $2$\\
\midrule
\multicolumn{11}{c}{\textit{Panel A}: linear outcome}\\
\midrule
0.1 & 500 & 0.110 & 0.097 & 0.100 & 0.118 & 0.103 & 0.117 & 0.112 & 0.051 & 0.059\\
 & 1,000 & 0.115 & 0.110 & 0.104 & 0.116 & 0.118 & 0.113 & 0.109 & 0.058 & 0.072\\
 & 2,000 & 0.108 & 0.107 & 0.100 & 0.097 & 0.108 & 0.093 & 0.093 & 0.061 & 0.080\\
\addlinespace[2pt]
0.05 & 500 & 0.054 & 0.051 & 0.044 & 0.053 & 0.044 & 0.052 & 0.049 & 0.018 & 0.024\\
 & 1,000 & 0.052 & 0.054 & 0.054 & 0.061 & 0.059 & 0.059 & 0.048 & 0.023 & 0.030\\
 & 2,000 & 0.043 & 0.049 & 0.050 & 0.058 & 0.051 & 0.059 & 0.036 & 0.027 & 0.031\\
\addlinespace[2pt]
0.01 & 500 & 0.007 & 0.003 & 0.004 & 0.010 & 0.010 & 0.012 & 0.014 & 0.004 & 0.004\\
 & 1,000 & 0.010 & 0.007 & 0.010 & 0.009 & 0.006 & 0.007 & 0.007 & 0.002 & 0.002\\
 & 2,000 & 0.014 & 0.012 & 0.014 & 0.017 & 0.009 & 0.015 & 0.008 & 0.007 & 0.006\\
\bottomrule
\end{tabular*}
\begin{tablenotes}[flushleft]\scriptsize
\item \emph{Note}: Empirical size over $1{,}000$ replications with $B = 500$
multiplier-bootstrap draws; $K = 5$ and $c_1=c_2=0$. Columns
$\mathrm{med}$, $1$, and $2$ denote the Gaussian scale $\sigma$.
Monte Carlo standard errors are $0.009$, $0.007$, and $0.003$ at
$\alpha=0.10$, $0.05$, and $0.01$.
\end{tablenotes}
\end{threeparttable}
\end{table}

\begin{table}[!tbp]
\caption{Power of $T^{X\!P}_\sigma$ across covariate scales, linear outcome, $c_1=0$ (covariate-dependent heterogeneity only).}
\label{supp:tab:power_L_c0}
\centering
\singlespacing
\setlength{\tabcolsep}{3pt}
\footnotesize
\begin{threeparttable}
\begin{tabular*}{\textwidth}{@{\extracolsep{\fill}}ll*{9}{r}@{}}
\toprule
& & \multicolumn{3}{c}{$\sigma_X=0.5$} & \multicolumn{3}{c}{$\sigma_X=1$} & \multicolumn{3}{c}{$\sigma_X=2$}\\
\cmidrule(lr){3-5}\cmidrule(lr){6-8}\cmidrule(lr){9-11}
$c_2$ & $n$ & $\mathrm{med}$ & $1$ & $2$ & $\mathrm{med}$ & $1$ & $2$ & $\mathrm{med}$ & $1$ & $2$\\
\midrule
\multicolumn{11}{c}{\makebox[0pt][c]{\textit{Panel A}: $S_1\propto X$ }}\\
\midrule
0.5 & 500 & 0.061 & 0.054 & 0.049 & 0.087 & 0.075 & 0.086 & 0.103 & 0.023 & 0.044\\
 & 1,000 & 0.099 & 0.103 & 0.088 & 0.147 & 0.123 & 0.145 & 0.185 & 0.042 & 0.093\\
 & 2,000 & 0.131 & 0.124 & 0.110 & 0.247 & 0.219 & 0.251 & 0.329 & 0.085 & 0.217\\
\addlinespace[2pt]
1 & 500 & 0.079 & 0.080 & 0.068 & 0.110 & 0.086 & 0.108 & 0.145 & 0.018 & 0.068\\
 & 1,000 & 0.143 & 0.137 & 0.107 & 0.235 & 0.200 & 0.238 & 0.316 & 0.065 & 0.200\\
 & 2,000 & 0.233 & 0.208 & 0.173 & 0.421 & 0.421 & 0.443 & 0.588 & 0.283 & 0.585\\
\addlinespace[2pt]
2 & 500 & 0.097 & 0.094 & 0.076 & 0.109 & 0.082 & 0.114 & 0.153 & 0.011 & 0.050\\
 & 1,000 & 0.184 & 0.175 & 0.120 & 0.257 & 0.229 & 0.266 & 0.383 & 0.078 & 0.287\\
 & 2,000 & 0.300 & 0.267 & 0.208 & 0.461 & 0.479 & 0.479 & 0.677 & 0.428 & 0.779\\
\midrule
\multicolumn{11}{c}{\makebox[0pt][c]{\textit{Panel B}: $S_2\propto e^{X/2}-Ee^{X/2}$ }}\\
\addlinespace[1pt]
0.5 & 500 & 0.085 & 0.075 & 0.070 & 0.124 & 0.115 & 0.125 & 0.089 & 0.018 & 0.043\\
 & 1,000 & 0.146 & 0.141 & 0.106 & 0.224 & 0.229 & 0.234 & 0.150 & 0.025 & 0.079\\
 & 2,000 & 0.245 & 0.209 & 0.163 & 0.446 & 0.481 & 0.478 & 0.296 & 0.074 & 0.202\\
\addlinespace[2pt]
1 & 500 & 0.107 & 0.100 & 0.094 & 0.155 & 0.148 & 0.158 & 0.093 & 0.011 & 0.030\\
 & 1,000 & 0.224 & 0.190 & 0.144 & 0.318 & 0.363 & 0.346 & 0.190 & 0.018 & 0.111\\
 & 2,000 & 0.400 & 0.346 & 0.256 & 0.655 & 0.758 & 0.698 & 0.421 & 0.119 & 0.376\\
\addlinespace[2pt]
2 & 500 & 0.126 & 0.111 & 0.100 & 0.159 & 0.140 & 0.168 & 0.085 & 0.003 & 0.018\\
 & 1,000 & 0.258 & 0.218 & 0.151 & 0.352 & 0.397 & 0.388 & 0.206 & 0.016 & 0.091\\
 & 2,000 & 0.486 & 0.432 & 0.304 & 0.710 & 0.837 & 0.766 & 0.460 & 0.114 & 0.441\\
\midrule
\multicolumn{11}{c}{\makebox[0pt][c]{\textit{Panel C}: $S_3\propto X^2-\mathrm{Var}(X)$ }}\\
\addlinespace[1pt]
0.5 & 500 & 0.876 & 0.860 & 0.791 & 0.703 & 0.655 & 0.715 & 0.292 & 0.057 & 0.142\\
 & 1,000 & 0.992 & 0.987 & 0.960 & 0.923 & 0.932 & 0.940 & 0.516 & 0.146 & 0.343\\
 & 2,000 & 1.000 & 1.000 & 1.000 & 1.000 & 0.999 & 1.000 & 0.908 & 0.370 & 0.741\\
\addlinespace[2pt]
1 & 500 & 0.969 & 0.952 & 0.911 & 0.893 & 0.901 & 0.907 & 0.485 & 0.080 & 0.286\\
 & 1,000 & 1.000 & 0.998 & 0.993 & 0.989 & 0.993 & 0.992 & 0.782 & 0.315 & 0.702\\
 & 2,000 & 1.000 & 1.000 & 1.000 & 1.000 & 1.000 & 1.000 & 0.990 & 0.793 & 0.995\\
\addlinespace[2pt]
2 & 500 & 0.986 & 0.972 & 0.933 & 0.926 & 0.934 & 0.941 & 0.551 & 0.093 & 0.357\\
 & 1,000 & 1.000 & 0.999 & 0.996 & 0.994 & 0.995 & 0.995 & 0.864 & 0.389 & 0.832\\
 & 2,000 & 1.000 & 1.000 & 1.000 & 1.000 & 1.000 & 1.000 & 0.994 & 0.916 & 0.999\\
\midrule
\multicolumn{11}{c}{\makebox[0pt][c]{\textit{Panel D}: $S_4\propto\sin(2X/\sigma_X)$ }}\\
\addlinespace[1pt]
0.5 & 500 & 0.203 & 0.162 & 0.108 & 0.223 & 0.260 & 0.244 & 0.127 & 0.051 & 0.099\\
 & 1,000 & 0.492 & 0.366 & 0.199 & 0.493 & 0.637 & 0.528 & 0.264 & 0.141 & 0.278\\
 & 2,000 & 0.895 & 0.790 & 0.452 & 0.861 & 0.971 & 0.911 & 0.488 & 0.429 & 0.648\\
\addlinespace[2pt]
1 & 500 & 0.431 & 0.320 & 0.180 & 0.424 & 0.531 & 0.464 & 0.224 & 0.086 & 0.233\\
 & 1,000 & 0.850 & 0.721 & 0.420 & 0.832 & 0.961 & 0.889 & 0.522 & 0.414 & 0.703\\
 & 2,000 & 0.997 & 0.991 & 0.828 & 0.997 & 1.000 & 0.999 & 0.870 & 0.939 & 0.990\\
\addlinespace[2pt]
2 & 500 & 0.558 & 0.444 & 0.225 & 0.525 & 0.657 & 0.577 & 0.274 & 0.114 & 0.305\\
 & 1,000 & 0.942 & 0.860 & 0.552 & 0.932 & 0.987 & 0.960 & 0.649 & 0.612 & 0.859\\
 & 2,000 & 1.000 & 0.998 & 0.929 & 0.999 & 1.000 & 0.999 & 0.967 & 0.997 & 1.000\\
\bottomrule
\end{tabular*}
\begin{tablenotes}[flushleft]\scriptsize
\item \emph{Note}: Rejection frequencies at the $5\%$ level over $1{,}000$
replications with $B = 500$ multiplier-bootstrap draws; $K = 5$. The
heterogeneity term is $\eta=(c_1+c_2S_j(X))U$ with $\mathrm{Var}(S_j)=5$.
Columns $\mathrm{med}$, $1$, and $2$ denote $\sigma$; their realised median
scales are $0.73$, $2.56$, and $10.92$. 
\end{tablenotes}
\end{threeparttable}
\end{table}

\begin{table}[!tbp]
\caption{Power of $T^{X\!P}_\sigma$ across covariate scales, linear outcome, $c_1=3$ (constant plus covariate-dependent heterogeneity).}
\label{supp:tab:power_L_c3}
\centering
\singlespacing
\setlength{\tabcolsep}{3pt}
\footnotesize
\begin{threeparttable}
\begin{tabular*}{\textwidth}{@{\extracolsep{\fill}}ll*{9}{r}@{}}
\toprule
& & \multicolumn{3}{c}{$\sigma_X=0.5$} & \multicolumn{3}{c}{$\sigma_X=1$} & \multicolumn{3}{c}{$\sigma_X=2$}\\
\cmidrule(lr){3-5}\cmidrule(lr){6-8}\cmidrule(lr){9-11}
$c_2$ & $n$ & $\mathrm{med}$ & $1$ & $2$ & $\mathrm{med}$ & $1$ & $2$ & $\mathrm{med}$ & $1$ & $2$\\
\midrule
\multicolumn{11}{c}{\makebox[0pt][c]{\textit{Panel A}: $S_1\propto X$ }}\\
\midrule
0.5 & 500 & 0.154 & 0.139 & 0.117 & 0.064 & 0.055 & 0.058 & 0.055 & 0.060 & 0.042\\
 & 1,000 & 0.370 & 0.319 & 0.254 & 0.101 & 0.160 & 0.107 & 0.075 & 0.228 & 0.141\\
 & 2,000 & 0.641 & 0.526 & 0.398 & 0.149 & 0.369 & 0.176 & 0.125 & 0.633 & 0.437\\
\addlinespace[2pt]
1 & 500 & 0.118 & 0.117 & 0.108 & 0.068 & 0.048 & 0.062 & 0.073 & 0.017 & 0.026\\
 & 1,000 & 0.282 & 0.248 & 0.212 & 0.120 & 0.150 & 0.128 & 0.118 & 0.106 & 0.126\\
 & 2,000 & 0.479 & 0.411 & 0.331 & 0.205 & 0.340 & 0.224 & 0.224 & 0.435 & 0.463\\
\addlinespace[2pt]
2 & 500 & 0.108 & 0.110 & 0.101 & 0.089 & 0.060 & 0.090 & 0.091 & 0.004 & 0.020\\
 & 1,000 & 0.227 & 0.216 & 0.183 & 0.170 & 0.162 & 0.172 & 0.205 & 0.036 & 0.128\\
 & 2,000 & 0.375 & 0.344 & 0.283 & 0.318 & 0.376 & 0.335 & 0.395 & 0.288 & 0.560\\
\midrule
\multicolumn{11}{c}{\makebox[0pt][c]{\textit{Panel B}: $S_2\propto e^{X/2}-Ee^{X/2}$ }}\\
\addlinespace[1pt]
0.5 & 500 & 0.109 & 0.093 & 0.081 & 0.053 & 0.064 & 0.051 & 0.046 & 0.071 & 0.049\\
 & 1,000 & 0.255 & 0.198 & 0.145 & 0.069 & 0.199 & 0.085 & 0.056 & 0.274 & 0.185\\
 & 2,000 & 0.495 & 0.363 & 0.228 & 0.142 & 0.509 & 0.187 & 0.091 & 0.681 & 0.555\\
\addlinespace[2pt]
1 & 500 & 0.068 & 0.063 & 0.065 & 0.050 & 0.062 & 0.054 & 0.037 & 0.024 & 0.024\\
 & 1,000 & 0.160 & 0.130 & 0.101 & 0.106 & 0.221 & 0.137 & 0.073 & 0.128 & 0.128\\
 & 2,000 & 0.326 & 0.237 & 0.151 & 0.262 & 0.607 & 0.321 & 0.156 & 0.488 & 0.493\\
\addlinespace[2pt]
2 & 500 & 0.069 & 0.065 & 0.062 & 0.081 & 0.066 & 0.075 & 0.049 & 0.004 & 0.014\\
 & 1,000 & 0.153 & 0.122 & 0.088 & 0.182 & 0.266 & 0.210 & 0.108 & 0.035 & 0.092\\
 & 2,000 & 0.290 & 0.237 & 0.153 & 0.451 & 0.719 & 0.533 & 0.254 & 0.238 & 0.485\\
\midrule
\multicolumn{11}{c}{\makebox[0pt][c]{\textit{Panel C}: $S_3\propto X^2-\mathrm{Var}(X)$ }}\\
\addlinespace[1pt]
0.5 & 500 & 0.242 & 0.203 & 0.161 & 0.191 & 0.332 & 0.212 & 0.093 & 0.166 & 0.154\\
 & 1,000 & 0.585 & 0.470 & 0.300 & 0.421 & 0.729 & 0.496 & 0.152 & 0.523 & 0.485\\
 & 2,000 & 0.961 & 0.903 & 0.674 & 0.823 & 0.990 & 0.914 & 0.335 & 0.907 & 0.898\\
\addlinespace[2pt]
1 & 500 & 0.662 & 0.599 & 0.470 & 0.485 & 0.626 & 0.537 & 0.191 & 0.142 & 0.198\\
 & 1,000 & 0.947 & 0.895 & 0.756 & 0.850 & 0.968 & 0.898 & 0.384 & 0.525 & 0.690\\
 & 2,000 & 0.998 & 0.998 & 0.981 & 0.996 & 1.000 & 1.000 & 0.819 & 0.950 & 0.994\\
\addlinespace[2pt]
2 & 500 & 0.899 & 0.865 & 0.756 & 0.757 & 0.842 & 0.803 & 0.338 & 0.115 & 0.281\\
 & 1,000 & 0.991 & 0.988 & 0.951 & 0.966 & 0.991 & 0.986 & 0.654 & 0.516 & 0.807\\
 & 2,000 & 1.000 & 1.000 & 0.999 & 1.000 & 1.000 & 1.000 & 0.969 & 0.957 & 0.999\\
\midrule
\multicolumn{11}{c}{\makebox[0pt][c]{\textit{Panel D}: $S_4\propto\sin(2X/\sigma_X)$ }}\\
\addlinespace[1pt]
0.5 & 500 & 0.252 & 0.228 & 0.195 & 0.099 & 0.131 & 0.105 & 0.060 & 0.102 & 0.073\\
 & 1,000 & 0.555 & 0.476 & 0.383 & 0.207 & 0.389 & 0.235 & 0.110 & 0.377 & 0.263\\
 & 2,000 & 0.888 & 0.806 & 0.633 & 0.419 & 0.802 & 0.484 & 0.176 & 0.849 & 0.688\\
\addlinespace[2pt]
1 & 500 & 0.306 & 0.273 & 0.222 & 0.182 & 0.241 & 0.198 & 0.091 & 0.075 & 0.090\\
 & 1,000 & 0.705 & 0.604 & 0.450 & 0.429 & 0.678 & 0.492 & 0.207 & 0.386 & 0.408\\
 & 2,000 & 0.979 & 0.922 & 0.765 & 0.812 & 0.986 & 0.885 & 0.423 & 0.929 & 0.905\\
\addlinespace[2pt]
2 & 500 & 0.449 & 0.373 & 0.259 & 0.335 & 0.448 & 0.367 & 0.157 & 0.063 & 0.158\\
 & 1,000 & 0.861 & 0.767 & 0.535 & 0.744 & 0.926 & 0.801 & 0.415 & 0.476 & 0.675\\
 & 2,000 & 0.999 & 0.990 & 0.876 & 0.991 & 1.000 & 0.995 & 0.778 & 0.984 & 0.995\\
\bottomrule
\end{tabular*}
\begin{tablenotes}[flushleft]\scriptsize
\item \emph{Note}: Design, notation, and scale values are as in
Table~\ref{supp:tab:power_L_c0}. 
\end{tablenotes}
\end{threeparttable}
\end{table}

\clearpage

\subsection{Directional statistics: linear design}\label{supp:di-linear}
With $c_1 = 0$, the deviation depends on covariates, whereas the directional statistics condition on $P$ alone and can detect only a component proportional to $\mathbb{E}[S_j(X)\mid P]$. Power is therefore limited in most cells,
although the median-scale kernel has meaningful power in some
cases. The clearest power gains come from the
higher-degree benchmark $\Qpoly{5}$: for example, at $\sigma_X=1$, $c_2=2$, and
$n=2{,}000$, it rejects with probability $0.750$ against $S_1$ and $0.661$
against $S_2$, compared with $0.178$ and $0.153$ for
$\TP{\mathrm{med}}$. This pattern suggests that the richer polynomial basis
can capture features of the projected departure that the directional kernel
statistics often miss.

With $c_1 = 3$, the departure is dominated by a component that depends on
$P$ alone, making this a more favorable setting for the directional
statistics. The three Gaussian scales often give similar power, so their
choice makes little difference in many cells. The polynomial benchmarks
frequently match or exceed the kernel statistics, though the gaps are often
modest. One clearer exception is $S_3$ at $\sigma_X=2$, $c_2=2$, and $n=2{,}000$,
where $\Qpoly{2}$ rejects with probability $0.528$, against $0.062$ for
$\TP{\mathrm{med}}$.

\begin{table}[!tbp]
\caption{Empirical size of the directional statistic $T^{P}_\sigma$ across covariate scales, linear outcome.}
\label{supp:tab:size_di_L}
\centering
\singlespacing
\setlength{\tabcolsep}{0pt}
\footnotesize
\begin{threeparttable}
\begin{tabular*}{\textwidth}{@{\extracolsep{\fill}}ll*{3}{r}@{\hspace{3pt}}*{2}{r}@{\hspace{6pt}}*{3}{r}@{\hspace{3pt}}*{2}{r}@{\hspace{6pt}}*{3}{r}@{\hspace{3pt}}*{2}{r}@{}}
\toprule
& & \multicolumn{5}{c}{$\sigma_X=0.5$} & \multicolumn{5}{c}{$\sigma_X=1$} & \multicolumn{5}{c}{$\sigma_X=2$}\\
\cmidrule(lr){3-7}\cmidrule(lr){8-12}\cmidrule(lr){13-17}
$\alpha$ & $n$ & $\mathrm{med}$ & $1$ & $2$ & $Q_{2}$ & $Q_{5}$ & $\mathrm{med}$ & $1$ & $2$ & $Q_{2}$ & $Q_{5}$ & $\mathrm{med}$ & $1$ & $2$ & $Q_{2}$ & $Q_{5}$\\
\midrule
0.1 & 500 & 0.092 & 0.101 & 0.099 & 0.081 & 0.085 & 0.094 & 0.095 & 0.093 & 0.084 & 0.073 & 0.112 & 0.103 & 0.108 & 0.072 & 0.088\\
 & 1,000 & 0.110 & 0.104 & 0.106 & 0.098 & 0.071 & 0.094 & 0.095 & 0.092 & 0.089 & 0.096 & 0.105 & 0.102 & 0.103 & 0.068 & 0.079\\
 & 2,000 & 0.096 & 0.104 & 0.103 & 0.089 & 0.081 & 0.098 & 0.106 & 0.106 & 0.077 & 0.079 & 0.107 & 0.109 & 0.108 & 0.079 & 0.105\\
\addlinespace[2pt]
0.05 & 500 & 0.043 & 0.048 & 0.049 & 0.044 & 0.042 & 0.041 & 0.048 & 0.048 & 0.036 & 0.034 & 0.056 & 0.064 & 0.060 & 0.030 & 0.042\\
 & 1,000 & 0.047 & 0.063 & 0.063 & 0.045 & 0.037 & 0.047 & 0.054 & 0.053 & 0.045 & 0.036 & 0.047 & 0.048 & 0.049 & 0.032 & 0.036\\
 & 2,000 & 0.050 & 0.051 & 0.051 & 0.050 & 0.041 & 0.050 & 0.041 & 0.043 & 0.039 & 0.036 & 0.052 & 0.054 & 0.053 & 0.033 & 0.060\\
\addlinespace[2pt]
0.01 & 500 & 0.006 & 0.008 & 0.007 & 0.008 & 0.011 & 0.014 & 0.011 & 0.012 & 0.007 & 0.003 & 0.011 & 0.010 & 0.010 & 0.004 & 0.006\\
 & 1,000 & 0.013 & 0.011 & 0.012 & 0.009 & 0.007 & 0.007 & 0.011 & 0.012 & 0.009 & 0.004 & 0.012 & 0.007 & 0.007 & 0.003 & 0.008\\
 & 2,000 & 0.009 & 0.009 & 0.008 & 0.005 & 0.006 & 0.010 & 0.008 & 0.008 & 0.012 & 0.004 & 0.009 & 0.009 & 0.009 & 0.001 & 0.013\\
\bottomrule
\end{tabular*}
\begin{tablenotes}[flushleft]\scriptsize
\item \emph{Note}: Empirical size over $1{,}000$ replications with $B = 500$
multiplier-bootstrap draws; $K = 5$ and $c_1=c_2=0$. Columns
$\mathrm{med}$, $1$, and $2$ denote the Gaussian scale $\sigma$;
$Q_2$ and $Q_5$ are HSU polynomial Wald benchmarks.
\end{tablenotes}
\end{threeparttable}
\end{table}

\begin{table}[!tbp]
\caption{Power of $T^{P}_\sigma$ across covariate scales, linear outcome, $c_1=0$ (covariate-dependent heterogeneity only).}
\label{supp:tab:power_di_L_c0}
\centering
\singlespacing
\setlength{\tabcolsep}{0pt}
\footnotesize
\begin{threeparttable}
\begin{tabular*}{\textwidth}{@{\extracolsep{\fill}}ll*{3}{r}@{\hspace{3pt}}*{2}{r}@{\hspace{6pt}}*{3}{r}@{\hspace{3pt}}*{2}{r}@{\hspace{6pt}}*{3}{r}@{\hspace{3pt}}*{2}{r}@{}}
\toprule
& & \multicolumn{5}{c}{$\sigma_X=0.5$} & \multicolumn{5}{c}{$\sigma_X=1$} & \multicolumn{5}{c}{$\sigma_X=2$}\\
\cmidrule(lr){3-7}\cmidrule(lr){8-12}\cmidrule(lr){13-17}
$c_2$ & $n$ & $\mathrm{med}$ & $1$ & $2$ & $Q_{2}$ & $Q_{5}$ & $\mathrm{med}$ & $1$ & $2$ & $Q_{2}$ & $Q_{5}$ & $\mathrm{med}$ & $1$ & $2$ & $Q_{2}$ & $Q_{5}$\\
\midrule
\multicolumn{17}{c}{\makebox[0pt][c]{\textit{Panel A}: $S_1\propto X$ }}\\
\midrule
0.5 & 500 & 0.054 & 0.054 & 0.053 & 0.044 & 0.075 & 0.067 & 0.051 & 0.053 & 0.047 & 0.111 & 0.079 & 0.050 & 0.047 & 0.041 & 0.086\\
 & 1,000 & 0.055 & 0.056 & 0.054 & 0.048 & 0.096 & 0.091 & 0.050 & 0.048 & 0.045 & 0.174 & 0.091 & 0.051 & 0.052 & 0.045 & 0.123\\
 & 2,000 & 0.056 & 0.045 & 0.051 & 0.041 & 0.144 & 0.110 & 0.057 & 0.050 & 0.041 & 0.319 & 0.156 & 0.071 & 0.062 & 0.042 & 0.213\\
\addlinespace[2pt]
1 & 500 & 0.054 & 0.050 & 0.049 & 0.048 & 0.132 & 0.084 & 0.054 & 0.054 & 0.033 & 0.197 & 0.110 & 0.059 & 0.055 & 0.041 & 0.141\\
 & 1,000 & 0.051 & 0.051 & 0.051 & 0.064 & 0.170 & 0.109 & 0.066 & 0.058 & 0.042 & 0.307 & 0.172 & 0.065 & 0.056 & 0.046 & 0.253\\
 & 2,000 & 0.057 & 0.060 & 0.059 & 0.053 & 0.290 & 0.179 & 0.066 & 0.060 & 0.039 & 0.607 & 0.326 & 0.102 & 0.079 & 0.050 & 0.544\\
\addlinespace[2pt]
2 & 500 & 0.046 & 0.047 & 0.047 & 0.050 & 0.185 & 0.094 & 0.064 & 0.060 & 0.045 & 0.294 & 0.148 & 0.063 & 0.057 & 0.049 & 0.202\\
 & 1,000 & 0.055 & 0.053 & 0.050 & 0.059 & 0.239 & 0.112 & 0.082 & 0.078 & 0.048 & 0.498 & 0.268 & 0.077 & 0.063 & 0.043 & 0.384\\
 & 2,000 & 0.066 & 0.065 & 0.066 & 0.062 & 0.395 & 0.178 & 0.080 & 0.073 & 0.041 & 0.750 & 0.501 & 0.118 & 0.076 & 0.050 & 0.744\\
\midrule
\multicolumn{17}{c}{\makebox[0pt][c]{\textit{Panel B}: $S_2\propto e^{X/2}-Ee^{X/2}$ }}\\
\addlinespace[1pt]
0.5 & 500 & 0.052 & 0.055 & 0.055 & 0.048 & 0.079 & 0.069 & 0.060 & 0.062 & 0.040 & 0.086 & 0.071 & 0.066 & 0.067 & 0.071 & 0.069\\
 & 1,000 & 0.058 & 0.054 & 0.055 & 0.036 & 0.096 & 0.085 & 0.055 & 0.052 & 0.048 & 0.153 & 0.082 & 0.076 & 0.076 & 0.062 & 0.085\\
 & 2,000 & 0.057 & 0.049 & 0.049 & 0.036 & 0.135 & 0.119 & 0.058 & 0.055 & 0.044 & 0.259 & 0.100 & 0.091 & 0.092 & 0.068 & 0.127\\
\addlinespace[2pt]
1 & 500 & 0.054 & 0.054 & 0.052 & 0.060 & 0.122 & 0.081 & 0.057 & 0.056 & 0.052 & 0.206 & 0.085 & 0.081 & 0.078 & 0.089 & 0.100\\
 & 1,000 & 0.059 & 0.060 & 0.060 & 0.050 & 0.183 & 0.102 & 0.067 & 0.061 & 0.066 & 0.337 & 0.125 & 0.102 & 0.099 & 0.101 & 0.137\\
 & 2,000 & 0.058 & 0.058 & 0.060 & 0.055 & 0.299 & 0.157 & 0.068 & 0.061 & 0.055 & 0.540 & 0.180 & 0.138 & 0.131 & 0.090 & 0.223\\
\addlinespace[2pt]
2 & 500 & 0.048 & 0.045 & 0.046 & 0.051 & 0.162 & 0.076 & 0.058 & 0.059 & 0.039 & 0.229 & 0.109 & 0.097 & 0.096 & 0.098 & 0.101\\
 & 1,000 & 0.061 & 0.053 & 0.049 & 0.049 & 0.273 & 0.101 & 0.071 & 0.068 & 0.049 & 0.430 & 0.164 & 0.125 & 0.116 & 0.123 & 0.188\\
 & 2,000 & 0.062 & 0.071 & 0.073 & 0.053 & 0.408 & 0.153 & 0.078 & 0.073 & 0.058 & 0.661 & 0.247 & 0.165 & 0.154 & 0.140 & 0.355\\
\midrule
\multicolumn{17}{c}{\makebox[0pt][c]{\textit{Panel C}: $S_3\propto X^2-\mathrm{Var}(X)$ }}\\
\addlinespace[1pt]
0.5 & 500 & 0.062 & 0.057 & 0.060 & 0.042 & 0.042 & 0.065 & 0.057 & 0.059 & 0.053 & 0.043 & 0.081 & 0.093 & 0.092 & 0.053 & 0.067\\
 & 1,000 & 0.063 & 0.073 & 0.073 & 0.048 & 0.058 & 0.065 & 0.072 & 0.072 & 0.049 & 0.049 & 0.111 & 0.130 & 0.128 & 0.072 & 0.064\\
 & 2,000 & 0.077 & 0.064 & 0.066 & 0.066 & 0.064 & 0.093 & 0.081 & 0.080 & 0.036 & 0.060 & 0.173 & 0.219 & 0.220 & 0.072 & 0.057\\
\addlinespace[2pt]
1 & 500 & 0.075 & 0.063 & 0.064 & 0.064 & 0.071 & 0.075 & 0.070 & 0.069 & 0.058 & 0.079 & 0.128 & 0.138 & 0.140 & 0.081 & 0.072\\
 & 1,000 & 0.079 & 0.096 & 0.095 & 0.081 & 0.067 & 0.085 & 0.099 & 0.094 & 0.069 & 0.100 & 0.194 & 0.219 & 0.220 & 0.087 & 0.073\\
 & 2,000 & 0.101 & 0.096 & 0.098 & 0.100 & 0.106 & 0.123 & 0.108 & 0.106 & 0.083 & 0.120 & 0.333 & 0.387 & 0.386 & 0.107 & 0.094\\
\addlinespace[2pt]
2 & 500 & 0.082 & 0.083 & 0.082 & 0.051 & 0.076 & 0.085 & 0.090 & 0.092 & 0.075 & 0.088 & 0.172 & 0.188 & 0.193 & 0.090 & 0.067\\
 & 1,000 & 0.090 & 0.120 & 0.118 & 0.081 & 0.100 & 0.099 & 0.111 & 0.109 & 0.079 & 0.109 & 0.268 & 0.295 & 0.295 & 0.107 & 0.092\\
 & 2,000 & 0.130 & 0.151 & 0.153 & 0.116 & 0.109 & 0.144 & 0.141 & 0.137 & 0.099 & 0.166 & 0.447 & 0.492 & 0.494 & 0.146 & 0.119\\
\midrule
\multicolumn{17}{c}{\makebox[0pt][c]{\textit{Panel D}: $S_4\propto\sin(2X/\sigma_X)$ }}\\
\addlinespace[1pt]
0.5 & 500 & 0.058 & 0.060 & 0.059 & 0.051 & 0.057 & 0.060 & 0.056 & 0.055 & 0.044 & 0.050 & 0.081 & 0.058 & 0.059 & 0.055 & 0.070\\
 & 1,000 & 0.057 & 0.060 & 0.061 & 0.057 & 0.043 & 0.074 & 0.056 & 0.052 & 0.041 & 0.078 & 0.074 & 0.062 & 0.059 & 0.046 & 0.103\\
 & 2,000 & 0.044 & 0.050 & 0.048 & 0.039 & 0.052 & 0.078 & 0.058 & 0.053 & 0.043 & 0.094 & 0.111 & 0.068 & 0.064 & 0.029 & 0.160\\
\addlinespace[2pt]
1 & 500 & 0.063 & 0.064 & 0.063 & 0.038 & 0.056 & 0.063 & 0.063 & 0.060 & 0.045 & 0.083 & 0.095 & 0.056 & 0.048 & 0.047 & 0.082\\
 & 1,000 & 0.066 & 0.069 & 0.069 & 0.051 & 0.072 & 0.095 & 0.062 & 0.060 & 0.037 & 0.120 & 0.138 & 0.072 & 0.060 & 0.038 & 0.191\\
 & 2,000 & 0.048 & 0.055 & 0.056 & 0.044 & 0.073 & 0.115 & 0.063 & 0.061 & 0.039 & 0.229 & 0.224 & 0.083 & 0.072 & 0.040 & 0.316\\
\addlinespace[2pt]
2 & 500 & 0.061 & 0.061 & 0.062 & 0.051 & 0.059 & 0.072 & 0.059 & 0.060 & 0.041 & 0.099 & 0.110 & 0.053 & 0.050 & 0.037 & 0.125\\
 & 1,000 & 0.065 & 0.064 & 0.064 & 0.049 & 0.064 & 0.105 & 0.058 & 0.056 & 0.053 & 0.165 & 0.199 & 0.082 & 0.075 & 0.047 & 0.255\\
 & 2,000 & 0.052 & 0.054 & 0.053 & 0.043 & 0.092 & 0.134 & 0.072 & 0.067 & 0.040 & 0.303 & 0.311 & 0.093 & 0.082 & 0.041 & 0.470\\
\bottomrule
\end{tabular*}
\begin{tablenotes}[flushleft]\scriptsize
\item \emph{Note}: Rejection frequencies at the $5\%$ level over $1{,}000$
replications with $B = 500$ multiplier-bootstrap draws; $K = 5$. The
heterogeneity term is $\eta=(c_1+c_2S_j(X))U$ with $\mathrm{Var}(S_j)=5$.
Columns $\mathrm{med}$, $1$, and $2$ denote $\sigma$; their realised median
scales are $0.73$, $2.56$, and $10.92$. $Q_2$ and $Q_5$ are HSU polynomial
Wald benchmarks. 
\end{tablenotes}
\end{threeparttable}
\end{table}

\begin{table}[!tbp]
\caption{Power of $T^{P}_\sigma$ across covariate scales, linear outcome, $c_1=3$ (constant plus covariate-dependent heterogeneity).}
\label{supp:tab:power_di_L_c3}
\centering
\singlespacing
\setlength{\tabcolsep}{0pt}
\footnotesize
\begin{threeparttable}
\begin{tabular*}{\textwidth}{@{\extracolsep{\fill}}ll*{3}{r}@{\hspace{3pt}}*{2}{r}@{\hspace{6pt}}*{3}{r}@{\hspace{3pt}}*{2}{r}@{\hspace{6pt}}*{3}{r}@{\hspace{3pt}}*{2}{r}@{}}
\toprule
& & \multicolumn{5}{c}{$\sigma_X=0.5$} & \multicolumn{5}{c}{$\sigma_X=1$} & \multicolumn{5}{c}{$\sigma_X=2$}\\
\cmidrule(lr){3-7}\cmidrule(lr){8-12}\cmidrule(lr){13-17}
$c_2$ & $n$ & $\mathrm{med}$ & $1$ & $2$ & $Q_{2}$ & $Q_{5}$ & $\mathrm{med}$ & $1$ & $2$ & $Q_{2}$ & $Q_{5}$ & $\mathrm{med}$ & $1$ & $2$ & $Q_{2}$ & $Q_{5}$\\
\midrule
\multicolumn{17}{c}{\makebox[0pt][c]{\textit{Panel A}: $S_1\propto X$ }}\\
\midrule
0.5 & 500 & 0.510 & 0.612 & 0.615 & 0.686 & 0.530 & 0.444 & 0.538 & 0.539 & 0.606 & 0.456 & 0.436 & 0.516 & 0.522 & 0.502 & 0.336\\
 & 1,000 & 0.831 & 0.891 & 0.892 & 0.927 & 0.822 & 0.720 & 0.805 & 0.804 & 0.865 & 0.747 & 0.691 & 0.780 & 0.781 & 0.812 & 0.610\\
 & 2,000 & 0.982 & 0.993 & 0.993 & 1.000 & 0.983 & 0.952 & 0.979 & 0.980 & 0.991 & 0.964 & 0.945 & 0.969 & 0.970 & 0.984 & 0.924\\
\addlinespace[2pt]
1 & 500 & 0.392 & 0.485 & 0.485 & 0.562 & 0.450 & 0.344 & 0.406 & 0.411 & 0.478 & 0.407 & 0.373 & 0.398 & 0.394 & 0.404 & 0.301\\
 & 1,000 & 0.695 & 0.765 & 0.763 & 0.845 & 0.694 & 0.599 & 0.671 & 0.667 & 0.761 & 0.675 & 0.597 & 0.641 & 0.639 & 0.664 & 0.536\\
 & 2,000 & 0.935 & 0.960 & 0.960 & 0.981 & 0.932 & 0.882 & 0.906 & 0.903 & 0.965 & 0.924 & 0.878 & 0.903 & 0.897 & 0.915 & 0.856\\
\addlinespace[2pt]
2 & 500 & 0.228 & 0.281 & 0.283 & 0.306 & 0.325 & 0.226 & 0.232 & 0.225 & 0.299 & 0.365 & 0.268 & 0.251 & 0.245 & 0.215 & 0.237\\
 & 1,000 & 0.403 & 0.465 & 0.460 & 0.542 & 0.496 & 0.379 & 0.379 & 0.372 & 0.495 & 0.590 & 0.462 & 0.401 & 0.386 & 0.419 & 0.477\\
 & 2,000 & 0.658 & 0.749 & 0.740 & 0.831 & 0.767 & 0.622 & 0.605 & 0.593 & 0.774 & 0.837 & 0.761 & 0.664 & 0.642 & 0.684 & 0.800\\
\midrule
\multicolumn{17}{c}{\makebox[0pt][c]{\textit{Panel B}: $S_2\propto e^{X/2}-Ee^{X/2}$ }}\\
\addlinespace[1pt]
0.5 & 500 & 0.505 & 0.585 & 0.588 & 0.694 & 0.515 & 0.422 & 0.533 & 0.536 & 0.610 & 0.454 & 0.351 & 0.451 & 0.451 & 0.519 & 0.317\\
 & 1,000 & 0.817 & 0.877 & 0.878 & 0.935 & 0.811 & 0.696 & 0.789 & 0.791 & 0.889 & 0.750 & 0.587 & 0.685 & 0.692 & 0.756 & 0.569\\
 & 2,000 & 0.981 & 0.992 & 0.992 & 0.996 & 0.991 & 0.937 & 0.971 & 0.970 & 0.996 & 0.973 & 0.855 & 0.915 & 0.918 & 0.964 & 0.871\\
\addlinespace[2pt]
1 & 500 & 0.384 & 0.479 & 0.484 & 0.583 & 0.467 & 0.300 & 0.389 & 0.393 & 0.491 & 0.430 & 0.240 & 0.280 & 0.285 & 0.380 & 0.261\\
 & 1,000 & 0.669 & 0.741 & 0.738 & 0.841 & 0.704 & 0.532 & 0.643 & 0.640 & 0.779 & 0.685 & 0.385 & 0.453 & 0.454 & 0.567 & 0.419\\
 & 2,000 & 0.928 & 0.954 & 0.954 & 0.983 & 0.935 & 0.812 & 0.879 & 0.877 & 0.971 & 0.930 & 0.632 & 0.694 & 0.696 & 0.856 & 0.734\\
\addlinespace[2pt]
2 & 500 & 0.215 & 0.277 & 0.275 & 0.340 & 0.343 & 0.171 & 0.201 & 0.202 & 0.276 & 0.331 & 0.115 & 0.131 & 0.131 & 0.234 & 0.189\\
 & 1,000 & 0.374 & 0.454 & 0.452 & 0.590 & 0.516 & 0.279 & 0.334 & 0.341 & 0.502 & 0.568 & 0.161 & 0.186 & 0.184 & 0.318 & 0.270\\
 & 2,000 & 0.605 & 0.710 & 0.700 & 0.862 & 0.802 & 0.475 & 0.546 & 0.545 & 0.777 & 0.836 & 0.286 & 0.266 & 0.259 & 0.547 & 0.525\\
\midrule
\multicolumn{17}{c}{\makebox[0pt][c]{\textit{Panel C}: $S_3\propto X^2-\mathrm{Var}(X)$ }}\\
\addlinespace[1pt]
0.5 & 500 & 0.471 & 0.601 & 0.605 & 0.728 & 0.556 & 0.421 & 0.550 & 0.551 & 0.695 & 0.492 & 0.307 & 0.421 & 0.427 & 0.505 & 0.322\\
 & 1,000 & 0.775 & 0.855 & 0.854 & 0.952 & 0.878 & 0.680 & 0.813 & 0.817 & 0.919 & 0.819 & 0.539 & 0.650 & 0.666 & 0.798 & 0.614\\
 & 2,000 & 0.976 & 0.993 & 0.991 & 0.998 & 0.990 & 0.935 & 0.977 & 0.978 & 0.998 & 0.992 & 0.810 & 0.904 & 0.908 & 0.984 & 0.922\\
\addlinespace[2pt]
1 & 500 & 0.321 & 0.449 & 0.452 & 0.663 & 0.465 & 0.257 & 0.395 & 0.398 & 0.577 & 0.434 & 0.174 & 0.240 & 0.245 & 0.398 & 0.252\\
 & 1,000 & 0.554 & 0.681 & 0.685 & 0.914 & 0.802 & 0.455 & 0.635 & 0.637 & 0.868 & 0.756 & 0.249 & 0.366 & 0.372 & 0.609 & 0.435\\
 & 2,000 & 0.850 & 0.928 & 0.925 & 0.993 & 0.977 & 0.739 & 0.884 & 0.893 & 0.987 & 0.962 & 0.429 & 0.597 & 0.609 & 0.896 & 0.756\\
\addlinespace[2pt]
2 & 500 & 0.113 & 0.200 & 0.202 & 0.463 & 0.324 & 0.099 & 0.171 & 0.183 & 0.376 & 0.294 & 0.064 & 0.083 & 0.085 & 0.207 & 0.132\\
 & 1,000 & 0.209 & 0.319 & 0.320 & 0.701 & 0.512 & 0.146 & 0.292 & 0.302 & 0.624 & 0.491 & 0.056 & 0.093 & 0.091 & 0.315 & 0.211\\
 & 2,000 & 0.319 & 0.505 & 0.518 & 0.945 & 0.847 & 0.235 & 0.444 & 0.463 & 0.899 & 0.812 & 0.062 & 0.102 & 0.106 & 0.528 & 0.380\\
\midrule
\multicolumn{17}{c}{\makebox[0pt][c]{\textit{Panel D}: $S_4\propto\sin(2X/\sigma_X)$ }}\\
\addlinespace[1pt]
0.5 & 500 & 0.519 & 0.606 & 0.608 & 0.689 & 0.499 & 0.458 & 0.570 & 0.569 & 0.609 & 0.441 & 0.447 & 0.566 & 0.573 & 0.554 & 0.368\\
 & 1,000 & 0.832 & 0.883 & 0.883 & 0.917 & 0.795 & 0.730 & 0.836 & 0.841 & 0.887 & 0.750 & 0.704 & 0.830 & 0.838 & 0.838 & 0.660\\
 & 2,000 & 0.987 & 0.995 & 0.994 & 0.997 & 0.987 & 0.964 & 0.987 & 0.987 & 0.990 & 0.959 & 0.951 & 0.988 & 0.989 & 0.989 & 0.949\\
\addlinespace[2pt]
1 & 500 & 0.383 & 0.486 & 0.485 & 0.561 & 0.368 & 0.351 & 0.445 & 0.441 & 0.494 & 0.364 & 0.378 & 0.461 & 0.462 & 0.432 & 0.310\\
 & 1,000 & 0.699 & 0.780 & 0.780 & 0.823 & 0.642 & 0.602 & 0.710 & 0.708 & 0.781 & 0.607 & 0.604 & 0.707 & 0.715 & 0.717 & 0.547\\
 & 2,000 & 0.925 & 0.961 & 0.959 & 0.978 & 0.914 & 0.887 & 0.936 & 0.939 & 0.964 & 0.891 & 0.898 & 0.947 & 0.950 & 0.952 & 0.876\\
\addlinespace[2pt]
2 & 500 & 0.215 & 0.282 & 0.283 & 0.291 & 0.215 & 0.217 & 0.248 & 0.247 & 0.267 & 0.222 & 0.270 & 0.269 & 0.261 & 0.238 & 0.260\\
 & 1,000 & 0.387 & 0.467 & 0.468 & 0.540 & 0.382 & 0.386 & 0.433 & 0.427 & 0.473 & 0.413 & 0.442 & 0.452 & 0.446 & 0.429 & 0.424\\
 & 2,000 & 0.646 & 0.730 & 0.729 & 0.798 & 0.625 & 0.608 & 0.680 & 0.664 & 0.761 & 0.691 & 0.719 & 0.731 & 0.721 & 0.725 & 0.721\\
\bottomrule
\end{tabular*}
\begin{tablenotes}[flushleft]\scriptsize
\item \emph{Note}: Design, notation, and scale values are as in
Table~\ref{supp:tab:power_di_L_c0}. 
\end{tablenotes}
\end{threeparttable}
\end{table}

\clearpage

\section{Additional Empirical Results}\label{supp:application}

Table~\ref{stab:robust-J} reports all eight statistics of
Table~\ref{main:tab:application} when the outcome means $\mu_d(X)$ are
specified as polynomials of degree $J = 1, 2, 3$ in the continuous covariates,
where the $J = 1$ rows reproduce the main text table. The conditioning set and
the multiplier draws are held fixed across $J$, so differences within a row
reflect the outcome-mean specification alone. Two features stand out. First,
the rejection for aspirations at midline is stable: the three omnibus
statistics give $p$-values between $0.014$ and $0.026$ for every $J$, while the
directional statistics stay above $0.24$. The remaining outcome--statistic
combinations are broadly stable across $J$, with no changes in rejection status
at the $5\%$ level. Second, no other statistic rejects at the $5\%$ level under
any specification. Among the kernel statistics, the closest calls are the
fixed-scale directional statistics for economic empowerment at midline, with
$p$-values between $0.053$ and $0.079$. The polynomial benchmarks reject
nowhere and their smallest $p$-values are $0.063$ to $0.075$, in the two midline
cells for control over the body and for aspirations.  Finally, Table~\ref{tab:application-diag}
reports the first-stage probit, the distribution of $\hat P$, and the
common-support limits for the six cells of Table~\ref{main:tab:application}.

\begin{table}[htbp]
\caption{Tests for essential heterogeneity in the ELA data under outcome
means of polynomial degree $J=1,2,3$: all statistics.}
\label{stab:robust-J}
\centering
\singlespacing
\setlength{\tabcolsep}{6pt}
\footnotesize
\begin{threeparttable}
\begin{tabular*}{\textwidth}{@{\extracolsep{\fill}}@{}l c rr rr rr@{}}
\toprule
& & \multicolumn{2}{c}{Economic} & \multicolumn{2}{c}{Control over}
& \multicolumn{2}{c}{Aspirations}\\
& & \multicolumn{2}{c}{empowerment} & \multicolumn{2}{c}{the body} & \multicolumn{2}{c}{}\\
\cmidrule(lr){3-4}\cmidrule(lr){5-6}\cmidrule(lr){7-8}
Statistic & $J$ & Mid & End & Mid & End & Mid & End\\
\midrule
\multicolumn{8}{@{}l}{\textit{Panel A: omnibus, $W=(X^\top,P)^\top$}}\\
$\TXP{\mathrm{med}}$ & 1 & 0.383 & 0.181 & 0.225 & 0.754 & \textbf{0.014} & 0.499\\
 & 2 & 0.556 & 0.154 & 0.389 & 0.841 & \textbf{0.024} & 0.544\\
 & 3 & 0.514 & 0.316 & 0.245 & 0.746 & \textbf{0.021} & 0.618\\
\addlinespace[2pt]
$\TXP{1}$ & 1 & 0.570 & 0.355 & 0.160 & 0.752 & \textbf{0.023} & 0.784\\
 & 2 & 0.647 & 0.353 & 0.263 & 0.829 & \textbf{0.026} & 0.827\\
 & 3 & 0.610 & 0.574 & 0.186 & 0.854 & \textbf{0.023} & 0.830\\
\addlinespace[2pt]
$\TXP{2}$ & 1 & 0.376 & 0.203 & 0.226 & 0.764 & \textbf{0.014} & 0.564\\
 & 2 & 0.554 & 0.183 & 0.393 & 0.848 & \textbf{0.025} & 0.602\\
 & 3 & 0.508 & 0.369 & 0.249 & 0.785 & \textbf{0.023} & 0.676\\
\midrule
\multicolumn{8}{@{}l}{\textit{Panel B: directional, $W=P$}}\\
$\TP{\mathrm{med}}$ & 1 & 0.770 & 0.264 & 0.291 & 0.629 & 0.283 & 0.593\\
 & 2 & 0.674 & 0.215 & 0.680 & 0.553 & 0.246 & 0.578\\
 & 3 & 0.659 & 0.149 & 0.616 & 0.512 & 0.244 & 0.598\\
\addlinespace[2pt]
$\TP{1}$ & 1 & 0.053 & 0.246 & 0.415 & 0.962 & 0.791 & 0.523\\
 & 2 & 0.079 & 0.230 & 0.369 & 0.929 & 0.852 & 0.348\\
 & 3 & 0.063 & 0.593 & 0.803 & 0.810 & 0.874 & 0.527\\
\addlinespace[2pt]
$\TP{2}$ & 1 & 0.053 & 0.338 & 0.453 & 0.973 & 0.783 & 0.606\\
 & 2 & 0.077 & 0.229 & 0.331 & 0.963 & 0.802 & 0.407\\
 & 3 & 0.062 & 0.677 & 0.782 & 0.875 & 0.819 & 0.588\\
\midrule
\multicolumn{8}{@{}l}{\textit{Panel C: polynomial benchmarks}}\\
$\Qpoly{2}$ & 1 & 0.390 & 0.588 & 0.066 & 0.101 & 0.063 & 0.221\\
 & 2 & 0.445 & 0.637 & 0.073 & 0.122 & 0.072 & 0.211\\
 & 3 & 0.466 & 0.548 & 0.063 & 0.121 & 0.075 & 0.198\\
\addlinespace[2pt]
$\Qpoly{5}$ & 1 & 0.204 & 0.351 & 0.322 & 0.612 & 0.373 & 0.623\\
 & 2 & 0.244 & 0.354 & 0.340 & 0.656 & 0.373 & 0.625\\
 & 3 & 0.235 & 0.196 & 0.347 & 0.660 & 0.377 & 0.610\\
\bottomrule
\end{tabular*}
\begin{tablenotes}[flushleft]\scriptsize
\item \emph{Note}: $\mu_d(X)$ is a degree-$J$ polynomial in the baseline
index and age (standardised to unit variance); branch dummies enter linearly.
Kernel-statistic $p$-values use $B=999$ village-level two-point multiplier
draws \citep{mammen1993bootstrap}; $\Qpoly{2}$ and $\Qpoly{5}$ use a wild
cluster bootstrap-$t$ with $B=999$. Bold marks rejection at the $5\%$ level.
\end{tablenotes}
\end{threeparttable}
\end{table}

\begin{table}[htbp]
\caption{Empirical application: first stage and propensity-score support.}
\label{tab:application-diag}
\centering
\begin{threeparttable}
\singlespacing
\setlength{\tabcolsep}{3pt}
\footnotesize
\begin{tabular}{@{}l rr rr rr@{}}
\toprule
& \multicolumn{2}{c}{Economic} & \multicolumn{2}{c}{Control over}
& \multicolumn{2}{c}{Aspirations}\\
& \multicolumn{2}{c}{empowerment} & \multicolumn{2}{c}{the body} & \multicolumn{2}{c}{}\\
\cmidrule(lr){2-3}\cmidrule(lr){4-5}\cmidrule(lr){6-7}
& Mid & End & Mid & End & Mid & End\\
\midrule
\multicolumn{7}{@{}l}{\textit{First-stage probit}}\\
Pseudo-$R^2$ & 0.085 & 0.080 & 0.086 & 0.079 & 0.084 & 0.080\\
\shortstack[l]{LR $\chi^2(5)$,\\excluded instruments} & 172.7 & 120.6 & 175.8 & 120.8 & 172.9 & 122.3\\
\shortstack[l]{$F(5,\cdot)$, linear\\probability model} & 30.9 & 23.3 & 31.6 & 23.4 & 31.0 & 23.6\\
\addlinespace[2pt]
\multicolumn{7}{@{}l}{\textit{Propensity score}}\\
Mean $\hat P$ & 0.130 & 0.174 & 0.130 & 0.174 & 0.130 & 0.174\\
\shortstack[l]{Standard\\deviation of $\hat P$} & 0.084 & 0.106 & 0.085 & 0.105 & 0.084 & 0.106\\
Minimum $\hat P$ & 0.005 & 0.015 & 0.003 & 0.012 & 0.005 & 0.014\\
Maximum $\hat P$ & 0.568 & 0.802 & 0.598 & 0.770 & 0.560 & 0.786\\
5th--95th percentile & $[0.024,0.277]$ & $[0.044,0.323]$ & $[0.022,0.279]$ & $[0.043,0.325]$ & $[0.024,0.276]$ & $[0.042,0.324]$\\
$\Pr(\hat P<0.30)$ & 0.962 & 0.901 & 0.963 & 0.898 & 0.962 & 0.899\\
\shortstack[l]{Common support\\$[\underline{p},\bar p]$} & $[0.023,0.484]$ & $[0.031,0.712]$ & $[0.018,0.510]$ & $[0.033,0.713]$ & $[0.023,0.500]$ & $[0.040,0.721]$\\
\quad as a share of $[0,1]$ & 0.461 & 0.682 & 0.492 & 0.680 & 0.476 & 0.680\\
\addlinespace[2pt]
\multicolumn{7}{@{}l}{\textit{Sample}}\\
\shortstack[l]{Take-up, treatment\\communities} & 0.170 & 0.221 & 0.170 & 0.221 & 0.170 & 0.221\\
\shortstack[l]{Take-up, control\\communities} & 0.043 & 0.079 & 0.043 & 0.079 & 0.043 & 0.078\\
Complier share & 0.127 & 0.143 & 0.127 & 0.142 & 0.127 & 0.143\\
Observations & 3,984 & 2,942 & 3,984 & 2,943 & 3,984 & 2,942\\
Clusters (villages) & 150 & 148 & 150 & 148 & 150 & 148\\
\bottomrule
\end{tabular}

\begin{tablenotes}[flushleft]\scriptsize
\item \emph{Note}: The first stage is a
probit of club participation on the excluded instruments (distance to the
nearest club, its square, the treatment dummy and the two distance--treatment
interactions) and on $X$ (the baseline index, nine branch dummies and age),
estimated on the sample retained by \texttt{mtefe}'s common-support trim.
Common support is the interval on which both participation states are
observed.
\end{tablenotes}
\end{threeparttable}
\end{table}

\end{document}